\documentclass[11pt]{article}
\usepackage[utf8]{inputenc}
\usepackage[a4paper, left=0.9in, right=0.9in, top=0.9in, bottom=0.9in]{geometry}

\usepackage{amsthm,amssymb,amsmath,mathtools}
\usepackage{bm}
\usepackage{array}
\usepackage{diagbox}
\usepackage{algorithm}
\usepackage{algpseudocode}
\usepackage{booktabs}
\usepackage{graphicx}
\usepackage{multirow}
\usepackage{float}
\usepackage{enumitem}
\usepackage{booktabs}  
\usepackage{pdflscape}  
\usepackage{rotating}
\usepackage{changepage,adjustbox}
\usepackage{geometry}
\usepackage{comment}

\usepackage[spaces,hyphens]{url}
\usepackage{hyperref}
\usepackage{blindtext}
\usepackage{titlesec}

\hypersetup{
  colorlinks   = true, 
  urlcolor     = black, 
  linkcolor    = blue, 
  citecolor   = red 
}
\usepackage{pdfpages}
\usepackage[all]{hypcap}
\usepackage[font=small]{caption}
\usepackage[font=small]{subcaption}
\usepackage[rightcaption]{sidecap}
\newtheorem{theorem}{Theorem}[section]
\newtheorem{corollary}[theorem]{Corollary}
\newtheorem{lemma}[theorem]{Lemma}
\newtheorem{proposition}[theorem]{Proposition}

\newtheorem{remark}{Remark}

\title{Liquidity Provision and Rebate Design in Option Markets}
\author{Samuel N. Cohen\footnote{samuel.cohen@maths.ox.ac.uk, Mathematical Institute and Oxford-Man Institute, University of Oxford, Andrew Wiles Building, Woodstock Road, Oxford, OX6 2GG, UK}, Lyndon Drake\footnote{lyndon.drake@theology.ox.ac.uk, Faculty of Theology and Religion, University of Oxford, Schwarzman Centre, Woodstock Road, Oxford, OX6 2GG, UK}, Zihan Guo\footnote{Primary Author; zihan.guo@maths.ox.ac.uk, Mathematical Institute, University of Oxford, Andrew Wiles Building, Woodstock Road, Oxford, OX6 2GG, UK}, Christoph Reisinger\footnote{christoph.reisinger@maths.ox.ac.uk, Mathematical Institute and Oxford-Man Institute, University of Oxford, Andrew Wiles Building, Woodstock Road, Oxford, OX6 2GG, UK}}
\date{}

\begin{document}
\maketitle



\begin{abstract}
  We provide a model for the nested optimisation problem of market making and rebate design problems in option markets and find optimal strategies. 
  A single market maker trades multiple European call options in a local-stochastic volatility option market with both make and take strategies,
 modeled, respectively, as continuous and impulse controls.
  Her objective is to maximize, over all admissible make-take strategies, net profit of option portfolio value and cumulative rebate revenue,
  subject to a penalty on residual portfolio delta and vega. 
  In addition, we demonstrate how an exchange can incentivize a market maker to improve market liquidity by setting suitable fee rebates, thereby resolving its own liquidity attraction problem.
  To this end, we propose a three-step rebate design scheme with flexibility to accommodate specific liquidity targets imposed by an exchange. 
  Numerical results are provided to validate the effectiveness of the proposed scheme.
\end{abstract}


\section{Introduction}{\label{Sec 1}}
The electronification of financial markets and the widespread adoption of computerized execution algorithms has led modern market making towards automation.
In most order-driven markets, \emph{market makers} today refers to those liquidity providers who use high-frequency algorithms to continuously place buy and sell orders 
on the limit order book (LOB).

Inspired by the early work of Ho and Stoll \cite{ho1981optimal}, Avellaneda and Stoikov proposed in their seminal paper 
\cite{avellaneda2008high} a fundamental framework for algorithmic market making with limit orders.
Formulating the market making problem as a utility maximization problem of the terminal Profit and Loss (PnL), 
the limit orders posted at optimal depth levels given arrival rates of counterparty market orders can be found by means of stochastic optimal control.

Since then, various extensions have been proposed in the literature. 
Gu{\'e}ant et al. showed in \cite{gueant2013dealing} that the Hamilton–Jacobi–Bellman (HJB) equation in the Avellaneda–Stoikov \cite{avellaneda2008high} model 
can be simplified to a linear system of ordinary differential equations (ODEs) under the assumption of limited inventory position and exponential execution intensity.
They also provided a closed-form approximation for the asymptotic behavior of the optimal market making strategy as the trading horizon approaches infinity. 
These results were further extended by Gu{\'e}ant \cite{gueant2017optimal} to a wider class of execution intensities and performance criteria,
with more details available in his textbook \cite{gueant2016financial}.
In parallel, Cartea and Jaimungal, along with their collaborators, developed alternative frameworks for market making with both limit and market orders.
In a series of papers \cite{cartea2014buy}, \cite{cartea2015risk}, \cite{cartea2017algorithmic}, and the book \cite{cartea2015algorithmic}, 
they added new features to market making modelling, including market impact, short-term alpha, adverse selection, ambiguity aversion, etc. 
Their modelling framework is different from the one proposed by Avellaneda and Stoikov \cite{avellaneda2008high},
but still relies on the standing assumption that limit orders can be placed at any price.\footnote{
However, except for small-tick assets, limit orders can only be placed at a few price levels,
see \cite[Section 11.1]{gueant2016financial} for more details.}
This assumption was relaxed by Guilbaud and Pham in the context of market making \cite{guilbaud2013optimal}, \cite{guilbaud2015optimal}, and by Cartea et al. in the context of market making \cite{cartea2016algorithmic}, optimal acquisition \cite{cartea2015optimal} and optimal hedging \cite{cartea2019hedge}.

As far as asset classes are concerned, most market making models are dedicated to equity markets.
However, the growing volume of derivatives trading has drawn increased attention to option market making over the past decade.
The first paper to address the option market making problem known to us was due to Stoikov and Sa{\u{g}}lam \cite{stoikov2009option}. 
They studied market making of a single European call option and obtained optimal market making strategies by discrete time approximation.
Subsequent work by El Aoud and Abergel \cite{el2015stochastic} and Baldacci et al. \cite{baldacci2021algorithmic} explored option market making in a stochastic volatility model 
for single and multiple options respectively. 
The latter also introduced a useful technique for dimension reductions via a constant approximation of the option Greeks.
More recently, Lucic and Tse \cite{lucic2024optimal} proposed an option market making model that enables option traders to express their views on volatility. 
Using a linear approximation as in \cite{avellaneda2008high}, they derived closed-form optimal quotes in terms of the expected volatility arbitrage profits in a single-option model. 
They further demonstrated how their baseline model could be extended to encompass additional features including position limits, general payoffs, etc.

Regardless of the asset class under consideration, most of the models mentioned above do not take into account interventions of the exchange.
Due to the fragmentation of financial markets, trading nowadays can take place on different venues, leading to competitions among exchanges (see \cite{lehalle2018market}).
A common mechanism employed by exchanges
to attract liquidity onto their platform is known as the \emph{make-take fee system}.\footnote{
Another tool for liquidity attraction is to set a suitable tick size, see \cite{baldacci2023bid} and \cite{dayri2015large} for example.} 
It is a policy that asymmetrically treats liquidity providers and liquidity consumers 
by compensating the former while taxing the latter. 
In practice, 
the make-take fee system redistributes the take fees (transaction fees) collected from the liquidity consumers 
as the make fees (fee rebates) offered to the liquidity providers. 

The theoretical study of make-take fees was initiated in \cite{euch2021optimal}, combining the Avellaneda–Stoikov \cite{avellaneda2008high} 
model with the Principal-Agent framework, for liquidity improvement on a single asset. 
There is an extensive literature on Principal-Agent problems,
with the pioneering work in a continuous-time setup by Holmstrom and Milgrom \cite{holmstrom1987aggregation},
who considered a model where the agent's unobserved efforts influence the drift of the output process impacting the wealth of the principal. 
This was later extended by Cvitani{\'c} et al. \cite{cvitanic2009optimal} with a lump-sum payment for general utility functions and by Sannikov \cite{sannikov2008continuous} 
with continuous payments over an infinite horizon.
Using backward stochastic differential equation (BSDE) theory, Cvitani{\'c} et al.\cite{cvitanic2018dynamic} allow the agent to control also the volatility of the output process.
In a similar spirit to these work, Euch et al. \cite{euch2021optimal} framed the make-take fees problem as a Principal-Agent problem:
the exchange (principal) aims to design suitable make-take fees that incentivize the market maker (agent) to increase the transaction flows (output), 
but he cannot base the policy on the spreads (efforts) for reasons detailed in \cite{euch2021optimal}.
This problem was then extended in various ways to include multiple competitive market makers \cite{baldacci2021optimal}, dark pool trading \cite{baldacci2023market},
and option market making \cite{baldacci2024design}. 
See also \cite{baldacci2024optimal} for a different solution to the make-take fees problem using a stochastic partial differential equation (SPDE) control approach. 

Nevertheless, the performance criteria in the aforementioned models neglect two important features.
First, while liquidity attraction is the primary objective of the exchange, it is not explicitly reflected in the exchange's performance criteria used by these models.
In a narrow sense, the quality of the market's liquidity is measured by the sizes of the bid-ask spreads.
However, it is only numerically verified that the optimal bid-ask spread is effectively reduced in the presence of a make fee incentive.
The only exception is \cite{baldacci2024design}, which penalizes insufficient liquidity in its performance criterion.
Second, the market maker's performance criteria in these models are only to maximize the expected utility of the total PnL while disregarding inventory risk. 
Inventory management is at the heart of market making strategies so as to avoid significant losses when extreme adverse price moves occur.

We develop a new paradigm to resolve the above issues of market making and relevant make-take fees. We present the framework in a local-stochastic-volatility option market.
From the market maker's perspective, we analyze a single market maker, 
who uses both limit and market orders to trade multiple European call options on an option exchange operating under the price-time microstructure. 
From the exchange's perspective, we demonstrate how to design a suitable rebate policy for which the market maker is incentivized to narrow the spread,
rather than simply increase the number of transactions.
Moreover, our framework is novel in the following two aspects.

\begin{itemize}[label=\textbullet] 
  \item \textbf{Inventory management for market maker:} 
  Our underlying asset is \emph{nontradable}\footnote{
  By nontradable, we refer to cases where the underlying is either physically nontradable, e.g.\ VIX options,
  or practically nontradable due to illiquidity and market frictions (transaction costs), e.g.\ SPXW options.
  } so that no delta-one instrument is available for a perfect delta hedge.
  As a result, the market maker chooses to trade the option themselves to manage her delta and vega exposures.
  If the market maker aims to maintain a delta-vega-neutral portfolio, she cannot achieve this solely through skewing her make strategies as limit order executions are uncertain. 
  Consequently, she may resort to market orders for an immediate execution.  
  However, since market orders are costly, the market maker may prefer to tolerate a residual portfolio delta or vega, provided it remains within an acceptable risk limit.
  Inspired by \cite{gueant2017optimal} (and also \cite{guilbaud2013optimal}, \cite{guilbaud2015optimal}, \cite{baldacci2021algorithmic}), 
  we capture the terminal values and the pathwise variations of the portfolio delta and vega, as penalty terms in the market maker's performance criterion,
  which reflect the market maker's risk aversion towards the cost of holding delta and vega inventory.
  \item \textbf{Liquidity attraction for exchange:}
  In contrast to \cite{euch2021optimal} where a \emph{second-best} case of rebate design problem is considered,
  our model focuses on a \emph{first-best} case where the efforts of the market maker are \emph{contractable}.\footnote{
  In contract theory, \emph{contractible} effort characterizes the \emph{first-best} benchmark: the principal can condition payments directly on the agent’s effort. When effort is non-contractible, the problem becomes \emph{second-best} (moral hazard), since the contract must be based on observable outcomes rather than effort itself.
  }
  One major challenge with \emph{noncontractable} efforts, as noted in \cite{euch2021optimal}, 
  is that limit orders do not guarantee executions, making it difficult for the exchange to justify subsidizing the market maker solely based on posted limit orders.
  In practice, however, market makers also provide liquidity for options with wide spreads --- such as deeply out-of-the-money (OTM) options --- 
  where executions rarely happen.
  From the exchange viewpoint, ensuring sufficient liquidity across all listed options is crucial for attracting transactions on its platform.
  In such cases, the exchange also compensates the market maker as long as her limit orders remain active, rather than only upon executions.
  As the exchange can monitor the exact actions of the market maker,
  the fee rebates are continuously paid to the market maker based on her actual quoting behavior.
  In our first-best framework, the exchange is able to anticipate the market maker's best response to any rebate policy, and also to influence her best response by devising an appropriate rebate policy, without solving a contract optimization problem formulated in \cite{euch2021optimal} as a leader-follower game.
  Consequently, our framework allows the exchange to implement a family of rebate policies tailored to different liquidity goals.\footnote{
  By liquidity goals, we mean that the exchange imposes specific liquidity targets for each listed option. 
  In practice, the exchange aims to maintain the small spreads for highly liquid options and attempts to narrow the large spreads for illiquid options 
  by incentivizing more competitive quotes. 
  }
  A relevant application is when the listed options are illiquid overall: the exchange is able to design a rebate policy 
  that incentivizes the market maker to consistently improve the best available prices (whenever possible), 
  thus narrowing the spreads and improving the market liquidity.
\end{itemize}

The paper is organized as follows. Section \ref{Sec 2} presents our option market making model.
In Section \ref{Sec 3}, we formulate the market making problem and solve it using dynamic programming methods. 
We then tackle the rebate design problem, for which a flexible three-step rebate design scheme is proposed to accommodate specific requirements of the exchange.
Section \ref{Sec 4} is devoted to the numerical algorithm.
Section \ref{Sec 5} provides numerical experiments and the resulting outcomes, using option order book data from the Chicago Board Options Exchange (CBOE).

\section{Model}{\label{Sec 2}}
In this section, we present our option market making model, which extends the framework of Guilbaud and Pham \cite{guilbaud2013optimal}
by incorporating exchange intervention through rebate incentives and allowing market making on multiple options.
Throughout, we work on a probability space \((\Omega,\mathcal{F},\mathbb{P})\) equipped with a filtration \(\mathbb{F} = (\mathcal{F}_t)_{t \geq 0}\) 
satisfying the usual conditions. All random variables and stochastic processes are defined on the stochastic basis \((\Omega ,\mathcal{F} ,\mathbb{F} ,\mathbb{P})\).   

\subsection{Option price and spread processes}{\label{Sec 2.1}}
We consider a single market maker for a list of European call options with prices \(C^i,\ i \in \mathcal{I} = \{1,2,\dots,M\}\), 
written on a non-tradable underlying \(S\) (for instance, a volatility or an equity index). 
The call options \(C^i\) are struck at levels \(K^i\), and expire at a common\footnote{
Our modelling framework can be generalized to vary the expiry dates, but for simplicity we consider only a common expiry date.} future date \(\tau\),
close to the present date \(0\), 
while market making takes place over a short time interval \([0,T]\) with \(T \ll \tau\).
Given the short duration of market making, we assume throughout that the short rate is zero during the market making horizon.
The price dynamics of the underlying \(S\) are described by a local stochastic volatility (LSV)\footnote{
We account for stochastic volatility even over a short horizon. 
This is because the market maker is trading options,
and volatility affects option valuation through the vega term, even over short horizons.
}
model of the generic form:
\begin{align}
  dS_t & = \sqrt{\nu_t} \sigma(t,S_t) S_t \ dW_t^S, \label{Eq 1} \\
  d \nu_t & = a(t,\nu_t) dt + b(t,\nu_t) \ dW_t^{\nu}, \label{Eq 2}
\end{align}  
with initial values \(S_0\) and \(\nu_0\) in \(\mathbb{R}_+\), where 
\begin{enumerate}
  \item \((W^S,W^{\nu})\) is a pair of standard \(1\)-dimensional Brownian motions under \(\mathbb{P}\), with quadratic covariation 
  \(\rho = \frac{d\langle W^S,W^{\nu}\rangle}{dt}\in (-1,1)\);
  \item \(a:[0,T] \times \mathbb{R}_+ \to \mathbb{R}\) and \(b: [0,T]\times \mathbb{R}_+ \to \mathbb{R}_+\) 
  are two deterministic functions, sufficiently regular to ensure that \eqref{Eq 2} has a pathwise unique strong solution \(\nu\) satisfying \(\nu >0\) \(\mathbb{P}\)-a.s. 
  and \(|\mathbb{E}[\int_t^T a(u,\nu_u)du|\nu_t=v]| \leq C(v)\) for all \(t\in [0,T]\) with \(C(v)\) a constant depending on \(v\) only;
  \item \(\sigma:[0,T]\times \mathbb{R}_+\to \mathbb{R}_+\) is a deterministic function,
  sufficiently regular to ensure that \eqref{Eq 1} has a pathwise unique strong solution \(S\) satisfying \(S>0\) \(\mathbb{P}\)-a.s.
\end{enumerate}

\begin{remark}
  The LSV model of the form \eqref{Eq 1} and \eqref{Eq 2} is general enough to cover most of the classical stochastic volatility (SV) models, for instance the Heston model \cite{heston1993closed}, in which 
  \(\sigma \equiv 1\), \(a:(t,\nu)\mapsto \kappa(\theta-\nu)\) and \(b:(t,\nu)\mapsto \xi \sqrt{\nu}\) with the mean-reversion speed \(\kappa>0\), 
  the mean-reversion level \(\theta>0\), and the volatility-of-volatility \(\xi \geq 0\) satisfying the Feller condition \(2 \kappa \theta > \xi^2\)
  (so that the squared volatility process \(\nu\) remains strictly positive). In this case, existence and uniqueness of the solution \(\nu\) follow from the Yamada–Watanabe Theorem 
  (see \cite[Proposition 5.2.13 and Corollary 5.3.23]{karatzas1998brownian}) due to the Lipschitz continuity of the drift coefficient 
  and the \(\frac{1}{2}\)-Hölder continuity of the volatility coefficient, and it holds for all \(t\in [0,T]\) that 
  \(\mathbb{E}[\int_t^T a(u,\nu_u)du|\nu_t=v] = (v-\theta)[e^{-\kappa(T-t)}-1]\) (see \cite{jeanblanc2009mathematical} Theorem 6.3.3.1), 
  which can be bounded by a constant depending on \(v\) only.        
\end{remark}

The market is incomplete since the uncertainty stemming from the Brownian motion \((W^S,W^{\nu})\) cannot be eliminated due to the lack of liquid hedging instruments.   
As indicated by the zero drift term in \eqref{Eq 1}, the market maker \emph{does not expect}\footnote{
As opposed to long/short investors who require compensation for taking the risk of underlying movements,
we consider a risk-neutral market maker who provides two-sided quotes without preference on the price trend.
}
 additional risk premium under the objective measure \(\mathbb{P}\), so she uses this measure for valuation. 
In other words, the pricing measure \(\mathbb{Q}\) chosen by the market maker coincides with the objective measure \(\mathbb{P}\). 
It follows that the arbitrage-free price process \((C^i_t)_{0 \leq t \leq \tau}\) of each call \(C^i\) is given by \((C^i(t,S_t,\nu_t))_{0 \leq t \leq \tau}\), 
where \(C^i\) solves the partial differential equation (PDE)
\[
  0 = \bigg\{\partial_t + a(t,\nu)\partial_{\nu} + \frac{1}{2}\nu \sigma^2(t,S) S^2 \partial_{SS} + \frac{1}{2}b^2(t,\nu)\partial_{\nu \nu} + \sigma(t,S)b(t,\nu)\rho S\sqrt{\nu} \partial_{S\nu}\bigg\} C^i(t,S,\nu)
\] with terminal condition 
\[
  C^i(\tau,S,\cdot) = (S-K^i)^+.
\]   

The market maker trades all the options on the same
option exchange\footnote{The market maker can trade on different venues, but trading on the same one can alleviate latency effect.}
where prices are formed through a continuous double auction implemented by a central LOB. 
By this we mean that both buy and sell orders are continuously placed and matched in a LOB visible to all market participants.  
We assume that the LOB operates under a price-time priority rule. That is, limit orders with best price and earliest submission time are given the highest execution priority. 

To simplify matters, we assume that the reference price\footnote{
Here, the reference price is understood as the mid price formed by the rest of the market, excluding the market maker herself.
In other words, it is the mid price on the order book as if the quotes of the market maker were absent.
}
process of each call \(C^i\) agrees with its arbitrage-free price process.
On the other hand, inspired by \cite{guilbaud2013optimal}
in the context of a single asset, 
we jointly model the bid-ask spreads of the calls \(C^1,\cdots,C^M\) by an \emph{exogenous}
(homogeneous\footnote{For simplicity we assume that the spread Markov chain is homogeneous.}) 
continuous time Markov chain (CTMC) \(D=(D^1,\cdots,D^M)\) parameterized by the generator \(Q = (q_{jk})_{j,k\in \mathbb{S}^M}\), 
where \(\mathbb{S}^M \coloneqq \{\delta ,2\delta,\ldots,m \delta\}^M\) is the state space with \(m\in \mathbb{N}\) and \(\delta>0\) the tick size common to each call.\footnote{
An implicit assumption here is that the spread of each call is always a multiple of the common tick size.} 
We assume that the Markov chain \(D\) is irreducible, in particular, does not have absorbing states, i.e. \(q(j) \coloneqq q_j \coloneqq \sum_{k \neq j}q_{j,k} \neq 0\) for all \(j\in \mathbb{S}^M\).

Given the generator \(Q\), the spread Markov chain \(D\) can be constructed via the \emph{jump matrix} \(P=(p_{jk})_{j,k\in \mathbb{S}^M}\) 
and the \emph{holding times} \(S_1,S_2,\dots\). 
The jump matrix \(P\) is defined by\footnote{
Empirical literature (see \cite{wyart2008relation} for instance) suggests that the spreads are proportional to the volatility:
the spreads tend to widen when the volatility increases.
For simplicity, we do not express a volatility dependence \(p_{jk}(\nu)\) in the transition rates,
but we remark that this relation can be useful in practice for inferring transition rates from volatility. 
} 
\[
  p_{jj} \coloneqq 0,\quad p_{jk} \coloneqq \frac{q_{jk}}{q_j},\quad j,k\in \mathbb{S}^M.
\]
Let \(\hat{D}\) be a discrete time Markov chain (DTMC) with transition matrix \(P\), and let \(\tau_1,\tau_2,\dots\) be independent exponentials with rate \(1\), independent of \(\hat{D}\). 
In addition, \(\hat{D}\) and \(\tau_1,\tau_2,\dots\) are chosen to be independent of \(S\) and \(\nu\). Then
\[
  S_n \coloneqq \frac{\tau_n}{q(\hat{D}_{n-1})},\quad n \in \mathbb{N}
\] defines the holding time before the \(n\)-th jump, and
\[
  T_n \coloneqq \sum_{\ell=1}^{n} S_\ell,\quad n \in \mathbb{N}
\] defines the \(n\)-th jump time. Thus,
\[
  D_t \coloneqq \hat{D}_{n},\quad T_n \leq t < T_{n+1}
\] defines the required spread Markov chain. In our context,
the jump times \(T_1,T_2,\dots\) represent the tick times at which the bid-ask spreads are influenced by trading activities of other market participants,
and the states of the embedded chain \(\hat{D}\) represent the bid-ask spreads in tick time. 

Now for each call \(C^i\), its best bid price process \(\underline{C^i}\) is given by 
\[
  \underline{C^i_t} = C^i_t - \frac{D^i_t}{2},\ t \geq 0,
\] and similarly, its best ask price process \(\overline{C^i}\) is given by 
\[
  \overline{C^i_t} = C^i_t+\frac{D_t^i}{2},\ t \geq 0,
\] where \(D^i\) is the \(i\)-th component of the joint spread process \(D\).

\subsection{Make-take strategies and fees}{\label{Sec 2.2}}  
The market maker is able to trade all the calls \(C^i\) with both limit and market orders. 
To clarify, a limit (resp.\ market) order is an order placed behind (resp.\ at) the best available price on the opposite side,
so while a market order results in an immediate execution, a limit order in general does not guarantee an execution.

The market maker is a liquidity provider as she continuously submits limit orders to facilitate trades in the market. 
For her limit order (make) strategy, she sends limit buy (resp.\ sell) orders at unit quantity by specifying the price she commits to 
pay (resp.\ receive), but needs to wait in a queue for the incoming market sell (resp.\ buy) orders to match her posted limit orders.
Instead of allowing the market maker to post her limit orders at any price, we restrict her make strategy to the following \emph{regimes}.

At any time \(t\), the market maker can only post her limit orders on either side according to one of the choices below:\footnote{
Make strategies of these forms have been investigated in a few papers, 
see \cite{guilbaud2013optimal}, \cite{cartea2015optimal}, \cite{guilbaud2015optimal}, \cite{cartea2016algorithmic} for example.}
\begin{itemize}
  \item Send limit orders \emph{at the touch}, that is, 
  send limit buy (resp.\ sell) orders at the current best bid (resp.\ ask) price \(\underline{C^i_t}\) (resp.\ \(\overline{C^i_t}\)); 
  \item Send limit orders with an improvement on the best prices by one tick, that is, 
  send limit buy (resp.\ sell) orders at the price \(\underline{C_t^{i}}^+= \underline{C_t^i}+\delta\) \big(resp. \(\overline{C_t^{i}}^-= \overline{C_t^i}-\delta\)\big); 
\end{itemize} 

For the first choice, posting limit orders behind the best available prices typically leads to the risk of not being executed 
because only a small fraction of market orders can consume the first slice of the LOB (see \cite{cartea2015optimal}).
The market maker therefore tends to quote at the touch when she decides to add more liquidity. 
On the other hand, the second choice is occasionally adopted by the market maker in order to gain priority in limit order executions (see \cite{guilbaud2013optimal}). 

It is noteworthy that if the time-\(t\) spread of the call \(C^i\) is one tick and the market maker takes the second choice, 
then she is actually sending a market order, in which case she walks the spread with an immediate execution for one unit of \(C^i\). 
Hence, as far as make strategies are concerned, we enforce the market maker to choose the first regime whenever the spread is one tick.

The market maker has full flexibility in her make strategy: she may quote conservatively by tracking the best available price (following the first regime), 
or aggressively by improving it by one tick (following the second regime).
This decision involves a tradeoff between a larger spread profit with slower execution (when quoting conservatively) 
and a smaller spread profit with faster execution (when quoting aggressively).
As the market maker continuously provides liquidity, she incurs no charges for submission or cancellation of her limit orders.
In addition, the market maker is \emph{small} in the sense that her trading activities do not impact the \emph{ambient}\footnote{
By ambient we emphasize that the bid-ask spreads, described by the exogenously given Markov chain,
are only determined by the other participants.} bid-ask spreads. 
In practice, this means the small market maker instantaneously updates her quotes only when the changes of the spreads 
are caused by the trading behavior of the other market participants, so that the spreads (excluding her own quotes) remain constant between her updates.

The limit order (make) strategy of the market maker is modelled by a continuous time predictable control process \(\alpha^{\mathrm{make}}\): 
\[
  \alpha_t^{\mathrm{make}} \coloneqq (\alpha_t^{b,i},\alpha_t^{a,i})_{i=1}^M,\ 0 \leq t \leq T,
\] where \(\alpha^{b,i}\) (resp.\ \(\alpha^{a,i}\)) represents the regime of the market maker's limit buy (resp. sell) strategy of the \(i\)-th call, 
and takes values in the action space \(\{0,1\}\) whose elements represent the first and the second regime respectively. 

Since the market maker cannot follow the second regime when the spread is one tick, for each \(t \geq 0\) and \(i\in \mathcal{I}\), 
\((\alpha_t^{b,i},\alpha_t^{a,i})\) takes values in 
\(\mathcal{A} (D^i_{t^-})\), where 
\[
  \mathcal{A} (d) \coloneqq \mathcal{A}^b(d) \times \mathcal{A}^a(d) \coloneqq \begin{dcases}
    \{0\} \times \{0\}, &\text{ if \(d=\delta\) }  ;\\
    \{0,1\} \times \{0,1\}, &\text{ if \(d>\delta \) }  .
  \end{dcases}
\]   

Thus, the bid price \(\pi^b(\alpha_t^{b,i},C_t^i,D_{t^-}^i)\) and the ask price \(\pi^a(\alpha_t^{a,i},C_t^i,D_{t^-}^i)\) of the call \(C^i\) 
submitted by the market maker at time \(t\) are given by
\[
  \begin{dcases}
    \pi^b(\alpha^b,c,d) = c-\frac{d}{2}+\delta \alpha^b,\\ 
    \pi^a(\alpha^a,c,d) = c + \frac{d}{2}-\delta \alpha^a.
  \end{dcases}
\] 

The limit orders are executed when they are matched by counterparty market orders. 
For each call \(C^i\), the arrivals of such market orders are modelled by two Cox processes \(N^{b,i} = (N_t^{b,i})_{t \geq 0}\)
and \(N^{a,i} = (N^{a,i}_t)_{t \geq 0}\) with controlled intensities \(\Lambda^{b,i}_t\) and \(\Lambda^{a,i}_t\) of the form 
\[
  \Lambda^{b,i}_t \coloneqq \lambda^{b,i}(\alpha_t^{b,i},D_{t^-}^i)\mathbb{I}_{{\mathcal{Q}}^{b,i}_{t^-}},\quad 
  \Lambda^{a,i}_t \coloneqq \lambda^{a,i}(\alpha_t^{a,i},D_{t^-}^i)\mathbb{I}_{{\mathcal{Q}}^{a,i}_{t^-}},
\]  
where \(\mathcal{Q}^{b,i}\) and \(\mathcal{Q}^{a,i}\) are events that describe trading constraints arising from inventory management, to be specified later, and    
\(\lambda^{b,i},\lambda^{a,i}\) are two deterministic functions satisfying
\begin{equation}{\label{Eq 3}}
  \lambda^{b,i}(0,\cdot) \leq \lambda^{b,i}(1,\cdot), \quad  \lambda^{a,i}(0,\cdot) \leq \lambda^{a,i}(1,\cdot),
\end{equation}
meaning that limit buy (resp.\ sell) orders with higher (resp.\ lower) prices are more frequently hit by counterparty market orders.
With probability one, the jumps of the Cox processes \(N^{b,i},N^{a,i},i\in \mathcal{I}\) do not occur simultaneously.

The market maker can manage her delta exposures via trading a basket of options.
But the executions of her limit orders are uncertain, so she may exploit market orders for immediate executions when necessary.
Unlike limit orders, continuous submission of market orders involves a continuous payment of the spreads.
Following \cite{guilbaud2013optimal} and \cite{guilbaud2015optimal}, the market order (take) strategy of the market maker is modelled by an impulse control 
\[
  \alpha^{\mathrm{take}} = (\tau_n,\xi_n)_{n=1}^\infty,
\]
where \(\{\tau_n\}\) is an increasing sequence of stopping times with limit \(T\), indicating the time when the market maker sends market orders, and each \(\xi_n\) 
is an \(\mathcal{E}^M\)-valued \(\mathcal{F}_{\tau_n}\)-measurable random variable, indicating the number of units of \(C^i,\ i\in \mathcal{I}\) 
purchased (resp.\ sold) at the best ask (resp.\ bid) price if \(\xi_n^i \geq 0\) (resp.\ \(\xi_n^i \leq 0\)), 
where \(\mathcal{E} =\{-\bar \ell,-\bar \ell +1,\dots,\bar \ell-1,\bar \ell\}\) for some threshold \(\bar \ell\in \mathbb{N}\). 
In this sense, the market maker is also a liquidity consumer, as she sends market orders to consume the liquidity offered by the existing limit orders. 

For a market order trade, we introduce the cost function \(h:\mathcal{E} \times \mathbb{R}_+ \times \mathbb{S} \to \mathbb{R}\) defined by
\begin{equation}{\label{Eq 4}}
  h(\ell,c,d) \coloneqq \ell c + |\ell| \Big(\frac{d}{2}+\epsilon\Big) + \tilde{\epsilon} \mathbb{I}_{\{\ell \neq 0\}}, 
\end{equation}
which indicates the amount to be paid immediately when passing a market order for a call with size \(\ell\) given the mid price \(c\) and the spread \(d\) of that call. 
Here, \(\epsilon>0\) represents the transaction fee per unit of call\footnote{An implicit assumption here is that the unit transaction fee is common to each call.} 
and \(\tilde{\epsilon}>0\) represents the fixed brokerage cost per market order. 
In other words, the pair \((\epsilon,\tilde{\epsilon})\) is the take fee that a liquidity consumer needs to pay for use of a market order. 

For each \(i\in \mathcal{I}\), denote by \(Y^i = (Y^i_t)_{0 \leq t \leq T}\) the inventory process of the call \(C^i\),
and by \(X = (X_t)_{0 \leq t \leq T}\) the cash holdings process. 
Under a make-take strategy \(\alpha = (\alpha^{\mathrm{make}},\alpha^{\mathrm{take}})\), the processes \(Y^i,\ i\in \mathcal{I}\) and \(X\) 
evolve according to the following dynamics:
\begin{equation}{\label{Eq 5}}
  \begin{aligned}
    \begin{dcases}
      d Y^i_t & = d N_t^{b,i} - d N_t^{a,i},\ \tau_n<t<\tau_{n+1},\ n \geq 0,\ i\in \mathcal{I},\\
      Y^i_{\tau_{n}} & = Y^i_{\tau_{n}^-} + \xi_n^i,\ n \geq 1,\\
      dX_t & = \sum_{i=1}^{M} (-\pi^b(\alpha_t^{b,i},C^i_t,D^i_{t^-}) dN_t^{b,i} +  \pi^a(\alpha_t^{a,i},C^i_t,D^i_{t^-})dN_t^{a,i}),\ \tau_n<t<\tau_{n+1},\ n \geq 0,\\
      X_{\tau_{n}} & = X_{\tau_{n}^-} - \sum_{i=1}^{M}  h(\xi_n^i,C^i_{\tau_n},D^i_{\tau_n^-}),\ n \geq 1,
    \end{dcases}
  \end{aligned}
\end{equation}
where \(\tau_0 \equiv 0\).

To define an admissible strategy, we impose the admissibility condition that, after each market order trade,
the portfolio delta (resp.\ vega) should lie\footnote{
In practice, the market maker directly manages her inventories rather than indirectly monitoring her portfolio delta and vega, 
but the latter is a simple while effective first-order proxy for inventory management (although the gamma risk is neglected).}
within a bounded interval \([-\bar{\Delta},\bar{\Delta}]\) (resp. \([-\bar{\mathcal{V}},\bar{\mathcal{V}}]\)), that is,
\[
  \begin{dcases}
    \Delta^\pi_{\tau_n^-} + \sum_{i=1}^{M} \xi_n^i \Delta^i_{\tau_n} \in [-\bar\Delta,\bar \Delta],\quad n\in \mathbb{N},\\
    \mathcal{V}^\pi_{\tau_n^-} + \sum_{i=1}^{M} \xi_n^i \mathcal{V}^i_{\tau_n} \in [-\bar{\mathcal{V}},\bar{\mathcal{V}}],\quad n\in \mathbb{N},
  \end{dcases}
\]
where \(\bar{\Delta}>0\) (resp.\ \(\bar{\mathcal{V}}>0\)) represents the total delta (resp.\ vega) limit of the market maker,
\(\Delta^i_t \coloneqq \partial_S C^i(t,S_t,\nu_t)\) (resp. \(\mathcal{V}^i_t \coloneqq \partial_{\nu} C^i(t,S_t,\nu_t)\)) 
is the time-\(t\) delta (resp. vega) of the call \(C^i\), and \(\Delta^{\pi}_t \coloneqq \sum_{i=1}^{M} Y_t^i \Delta^i_t\) 
(resp.\ \(\mathcal{V}^\pi_t \coloneqq \sum_{i=1}^{M} Y_t^i \mathcal{V}^i_t\)) is the time-\(t\) portfolio delta (resp.\ vega).
Note that the vega \(\mathcal{V}_t^i\) here is defined as the first order sensitivity of the price with respect to the \emph{squared} volatility,
rather than the classical one defined with respect to the volatility \(\partial_{\sqrt{\nu}} C^i(t,S_t,\nu_t)\) 
in stochastic volatility model (see \cite{baldacci2021algorithmic}),
but by the chain rule the latter matches the former up to a multiplier \(\sqrt{\nu_t}\).
We denote by \(\mathcal{A}\) the family of all admissible make-take strategies \(\alpha = (\alpha^{\mathrm{make}},\alpha^{\mathrm{take}})\).

\subsection{Limit order execution intensities}{\label{Sec 2.3}}
We now specify the events \(\mathcal{Q}^{b,i}\) and \(\mathcal{Q}^{a,i}\) appearing in the intensities of the execution processes \(N^{b,i}\) and \(N^{a,i}\). 
Inspired by \cite{baldacci2021algorithmic} (and also \cite{baldacci2024design}, \cite{lucic2024optimal}), 
we invoke a constant approximation to the delta and the vega of each option within a short market making period. 
Precisely, the delta and the vega of each option over \([0,T]\) are approximated by their initial values:
\[
  \begin{dcases}
    \Delta^i_t = \partial_S C^i(t,S_t,\nu_t) \approx \partial_S C^i(0,S_0,\nu_0) \coloneqq \Delta^i,\quad t \in [0,T], \quad i \in \mathcal{I},\\
    \mathcal{V}^i_t = \partial_{\nu} C^i(t,S_t,\nu_t) \approx \partial_{\nu} C^i(0,S_0,\nu_0) \coloneqq \mathcal{V}^i,\quad t \in [0,T], \quad i \in \mathcal{I}.
  \end{dcases}
\]
Under this constant approximation, the portfolio delta and vega now depend only on the inventory:
\[
  \Delta^\pi_t = \sum_{i=1}^{M} Y_t^i \Delta^i,\quad \mathcal{V}^\pi_t = \sum_{i=1}^{M} Y_t^i \mathcal{V}^i,\quad t\in [0,T],
\]
and we now define the events
\[
  \mathcal{Q}^{b,i}_{t^-} \coloneqq \{\Delta^\pi_{t^-} + \Delta^i \leq \bar{\Delta},\ \mathcal{V}^\pi_{t^-}+\mathcal{V}^i \leq \bar{\mathcal{V}}\},\quad 
  \mathcal{Q}^{a,i}_{t^-} \coloneqq \{\Delta^\pi_{t^-} - \Delta^i \geq -\bar{\Delta},\ \mathcal{V}^\pi_{t^-}-\mathcal{V}^i \geq -\bar{\mathcal{V}}\}.
\]
Then the intensities \(\Lambda_t^{b,i}\) and \(\Lambda_t^{a,i}\) of the Cox processes \(N^{b,i}\) and \(N^{a,i}\) become 
\[
  \Lambda_t^{b,i} \coloneqq  \lambda^{b,i}(\alpha_t^{b,i},D_{t^-}^i)\mathbb{I}_{\{\Delta^\pi_{t^-} + \Delta^i \leq \bar{\Delta},\ \mathcal{V}^\pi_{t^-}+\mathcal{V}^i \leq \bar{\mathcal{V}}\}},\quad 
  \Lambda_t^{a,i} \coloneqq  \lambda^{a,i}(\alpha_t^{a,i},D_{t^-}^i)\mathbb{I}_{\{\Delta^\pi_{t^-} - \Delta^i \geq -\bar{\Delta},\ \mathcal{V}^\pi_{t^-}-\mathcal{V}^i \geq -\bar{\mathcal{V}}\}}.
\]  

The definition of the intensities \(\Lambda_t^{b,i}\) and \(\Lambda_t^{a,i}\) 
implies that the market maker will continue to post limit orders 
only if her portfolio delta and vega do not touch the boundary \(\bar \Delta\) and \(\bar{\mathcal{V}}\) after the next limit order execution.
This enables the market maker to skew her make strategy by immediately ceasing trading on the side where her portfolio delta or vega is about to cross the boundary.

In particular, the skewness of the make strategies, together with the admissibility condition, ensures that the market maker's portfolio delta (resp.\ vega)
is bounded by \(\bar{\Delta}\) (resp.\ \(\bar{\mathcal{V}}\)) during the entire market making horizon. 
We set \(\lambda^{b,i}_{\max} = \max\limits_{\alpha^{b,i},d^i} \lambda^{b,i}(\alpha^{b,i},d^i)\) and 
\(\lambda^{a,i}_{\max} = \max\limits_{\alpha^{a,i},d^i} \lambda^{a,i}(\alpha^{a,i},d^i)\) for all \(i\in \mathcal{I}\), and let 
\(\Lambda_{\max} \coloneqq \max\limits_{i\in \mathcal{I}} (\lambda^{b,i}_{\max}\vee \lambda^{a,i}_{\max})\). 
Then the intensities \(\Lambda_t^{b,i}\) and \(\Lambda_t^{a,i}\) are uniformly bounded by \(\Lambda_{\max}\). 

Following an addmissible strategy \(\alpha\in \mathcal{A}\), the process \(V = (V_t)_{0 \leq t \leq T}\) of the Mark-to-Market (MtM) portfolio value, 
using the options' mid prices as benchmark, is given by 
\begin{equation}{\label{Eq 6}}
  V_t \coloneqq X_t + \sum_{i=1}^{M} Y_t^i C^i(t,S_t,\nu_t),\ 0 \leq t \leq T.
\end{equation}
Using \eqref{Eq 1}, \eqref{Eq 2}, \eqref{Eq 4} and \eqref{Eq 5}, a direct integration by parts shows that its dynamics are governed by 
\begin{equation}{\label{Eq 7}}
  \begin{dcases}
    d V_t = & \sum_{i=1}^{M} \bigg(\Big(\frac{D^i_{t^-}}{2}-\delta \alpha_t^{b,i}\Big) d N_t^{b,i} + \Big(\frac{D^i_{t^-}}{2}-\delta \alpha_t^{a,i}\Big) d N_t^{a,i}\bigg)+
    \Delta^{\pi}_{t^-}dS_t + \mathcal{V}^\pi_{t^-} \big(d \nu_t - a(t,\nu_t)dt\big),\\
    & \tau_n<t<\tau_{n+1},\ n \geq 0,\\
    V_{\tau_{n}} = & V_{\tau_{n}^-} - \sum_{i=1}^{M} \bigg[|\xi^i_n| \Big(\frac{D^i_{\tau_n^-}}{2}+\epsilon\Big)+\tilde{\epsilon}\mathbb{I}_{\{\xi^i_n \neq 0\}}\bigg],\ n \geq 1.  
  \end{dcases}
\end{equation}

\subsection{Rebates for liquidity attraction}
As mentioned earlier, the exchange aims to attract liquidity to its platform through the make-take fee system. 
In our model, the make fees refer to the fee rebates offered by the exchange to the market maker as incentives for reducing the bid-ask spreads. 
Unlike \cite{euch2021optimal} and its extensions, the fee rebates in our model are continuously delivered and may depend on the behavior of the market maker. 
The market maker receives fee rebates as long as her limit orders remain active, rather than only upon the executions of those orders. 
This design mechanism is particularly applicable in scenarios where the market maker undertakes market making for illiquid options, for example, deeply OTM options.  
The instantaneous fee rebates \(F^{b,i},\ F^{a,i},\ i\in \mathcal{I}\) are bounded \(\mathbb{R}\)-valued\footnote{
Note that negative fee rebates are possible, meaning that the market maker possibly incurs charges due to failure of liquidity provision responsibility.
} functions of the generic form:
\[
  F^{b,i} = F^{b,i}(\alpha^{b,i},d^i,\Delta^\pi,\mathcal{V}^\pi),\quad F^{a,i} = F^{a,i}(\alpha^{a,i},d^i,\Delta^\pi,\mathcal{V}^\pi),
\]
where \(\alpha^{b,i}\in \mathcal{A}^b(d^i),\ \alpha^{a,i}\in \mathcal{A}^a(d^i),\ d^i\in \mathbb{S},\ \Delta^\pi \in [-\bar{\Delta},\bar{\Delta}],\ \mathcal{V}^\pi \in [-\bar{\mathcal{V}},\bar{\mathcal{V}}]\). 

We call \(F = (F^{b,i},F^{a,i})_{i=1}^M\) a \emph{feasible rebate policy} if it satisfies
\begin{equation}{\label{Eq 8}}
  F^{b,i}(0,\cdot) \leq F^{b,i}(1,\cdot),\quad F^{a,i}(0,\cdot) \leq F^{a,i}(1,\cdot),\quad i \in \mathcal{I},
\end{equation}
which ensures that the market maker is never less compensated when narrowing the spreads. 
In the sequel, we consider only feasible rebate policies and their collection is denoted by \(\mathfrak{F}\).

Finally, for a given feasible rebate policy \(F \in \mathfrak{F}\), the process \(R = (R_t)_{0 \leq t \leq T}\) of the cumulative rebate revenue is given by\footnote{
The market maker does not receive fee rebates if she ``stops'' posting limit orders due to the delta or vega inventory constraint.}
\begin{equation}{\label{Eq 9}}
  \begin{aligned}
    R_t = \sum_{i=1}^{M} \int_0^t \big(F^{b,i}(\alpha^{b,i}_u,D^i_{u^-},\Delta^\pi_{u^-},\mathcal{V}^\pi_{u^-})\mathbb{I}_{\mathcal{Q}_{u^-}^{b,i}}
    + F^{a,i}(\alpha^{a,i}_u,D^i_{u^-},\Delta^\pi_{u^-},\mathcal{V}^\pi_{u^-})\mathbb{I}_{\mathcal{Q}_{u^-}^{a,i}}\big)du,\ 0 \leq t \leq T,
  \end{aligned}
\end{equation}
and the resulting total PnL process \(\mathrm{PL} = (\mathrm{PL}_t)_{0 \leq t \leq T}\) of the market maker is given by
\begin{equation}{\label{Eq 10}}
  \mathrm{PL} _t = V_t + R_t,\ 0 \leq t \leq T.
\end{equation}

\section{Market making and rebate design}{\label{Sec 3}}
In this section, we consider the complementary problems of market making and rebate design. 
Given a feasible rebate policy, we formulate the market making problem as a well-posed dynamic optimization problem,
and derive the associated dynamic programming equation.
We then focus on the special case of the Black–Scholes model.
We show that a dimension reduction is possible using an appropriate ansatz, 
and hence we can simplify the dynamic programming equation and obtain the optimal make-take strategy in feedback form. 
Further, we propose a three-step scheme to solve the rebate design problem. 
Notably, the proposed scheme is flexible enough to enable the exchange to set suitable fee rebates 
according to different liquidity targets on the traded options.    

\subsection{Market maker's optimization problem}{\label{Sec 3.1}}
Following the academic literature on (option) market making 
(see \cite{guilbaud2013optimal}, \cite{guilbaud2015optimal}, \cite{gueant2017optimal}, \cite{baldacci2021algorithmic}), 
for a given feasible rebate policy \(F\in \mathfrak{F}\) and an admissible make-take strategy \((\alpha^{\mathrm{make}},\alpha^{\mathrm{take}})\in \mathcal{A}\),
we consider a mean-quadratic performance criterion:
\begin{equation}{\label{Eq 11}}
  J(\alpha^{\mathrm{make}},\alpha^{\mathrm{take}};F) \coloneqq \mathbb{E} \bigg[\mathrm{PL}_T - \gamma_1^\prime (\Delta^\pi_T)^2 - \gamma_2^\prime  (\mathcal{V}^\pi_T)^2 - \gamma_1\int_0^T (\Delta_t^\pi)^2 dt-\gamma_2 \int_0^T (\mathcal{V}_t^\pi)^2 dt \bigg],
\end{equation}
which is the expected total PnL from market making, with a penalty on the running portfolio delta and vega. We specifically include a terminal penalty on the portfolio delta and vega, however, we do not necessarily interpret this as a `liquidation cost' of the portfolio. The reason for this is to allow us to easily encode the difference of time scales in the market -- questions of market making strategy are often taken over very short horizons (as they depend on microstructure of the market, which typically operates at a scale of microseconds) whereas liquidation values of options are only determined at expiry, which (even for short-dated options) is at the scale of hours or days. We include the penalty on the terminal portfolio delta and vega as a proxy for the risk preference of the market maker, who in practice will continue to trade beyond time horizon \(T\); by dynamic programming, there should exist a value at time \(T\) which reflects her preference, and the quadratic penalty gives a simple approximation of this value. In particular, the choice of parameters \(\gamma_1, \gamma_2, \gamma_1^\prime, \gamma_2^\prime\) reflects the risk preference of the market maker over the (short) period \([0,T]\), both dynamically (via the running penalties) and at the end of the horizon (via the terminal penalties).

From \eqref{Eq 6}, \eqref{Eq 9}, and \eqref{Eq 10}, the objective of the market maker is to maximize, over all admissible make-take strategies \((\alpha^{\mathrm{make}},\alpha^{\mathrm{take}})\),
the rewritten performance criterion:
\begin{equation}{\label{Eq 12}}
  \begin{aligned}
    & J(\alpha^{\mathrm{make}},\alpha^{\mathrm{take}};F)\\
    = \ & \mathbb{E}\bigg[X_T + \sum_{i=1}^{M}Y^i_T C^i(T,S_T,\nu_T) - \gamma_1^\prime (\Delta^\pi_T)^2 - \gamma_2^\prime (\mathcal{V}^\pi_T)^2- \gamma_1 \int_0^T (\Delta_t^\pi)^2dt - \gamma_2 \int_0^T (\mathcal{V}_t^\pi)^2dt\\
    &\quad + \sum_{i=1}^{M} \int_0^T \Big(F^{b,i}(\alpha^{b,i}_t,D^i_{t^-},\Delta^\pi_{t^-},\mathcal{V}^\pi_{t^-})\mathbb{I}_{\mathcal{Q}_{t^-}^{b,i}}
    + F^{a,i}(\alpha^{a,i}_t,D^i_{t^-},\Delta^\pi_{t^-},\mathcal{V}^\pi_{t^-})\mathbb{I}_{\mathcal{Q}_{t^-}^{a,i}}  \Big)dt \bigg].
  \end{aligned}
\end{equation}

As the performance criterion \eqref{Eq 12} is determined by the state variables \((S,\nu,X,Y,D)\), the value function \(u\) for a given feasible rebate policy \(F\in \mathfrak{F}\) is defined by 
\begin{equation}{\label{Eq 13}}
  \begin{aligned}
    u(t,s,v,x,y,d;F)
    & \coloneqq \sup\limits_{\alpha\in \mathcal{A}} \bigg\{\mathbb{E}^{t,s,v,x,y,d} \bigg[X_T + \sum_{i=1}^{M}Y^i_T C^i(T,S_T,\nu_T)\\
    & \quad - \gamma_1^\prime (\Delta^\pi_T)^2 -\gamma_2^\prime (\mathcal{V}^\pi_T)^2
    - \gamma_1 \int_t^T (\Delta^{\pi}_u)^2 du- \gamma_2 \int_t^T (\mathcal{V}^\pi_u)^2du \\
    & \quad + \sum_{i=1}^{M} \int_t^T \Big(F^{b,i}(\alpha^{b,i}_u,D^i_{u^-},\Delta^\pi_{u^-},\mathcal{V}^\pi_{u^-})\mathbb{I}_{\mathcal{Q}_{u^-}^{b,i}} 
    +F^{a,i}(\alpha^{a,i}_u,D^i_{u^-},\Delta^\pi_{u^-},\mathcal{V}^\pi_{u^-})\mathbb{I}_{\mathcal{Q}_{u^-}^{a,i}} \Big)du \bigg] \bigg\},
  \end{aligned}
\end{equation}
where \((t,s,v,x,y,d) \in [0,T]\times \mathbb{R}_+^2 \times \mathbb{R} \times \mathcal{Q} \times \mathbb{S}^M\) 
with \(\mathcal{Q}\) the set of authorized inventories:
\[
  \mathcal{Q} \coloneqq \bigg\{y \in \mathbb{Z}^M:\ \sum_{i=1}^{M} y^i \Delta^i \in [-\bar \Delta,\bar \Delta],\ \sum_{i=1}^{M} y^i \mathcal{V}^i \in [-\bar{\mathcal{V}},\bar{\mathcal{V}}]\bigg\}. 
\] 

We first check that the dynamic optimization problem \eqref{Eq 13} is well-posed.
In particular, the following proposition shows that the value function is finite and locally bounded. 

\begin{proposition}{\label{Prop 3.1}}
  For any given feasible rebate policy \(F\in \mathfrak{F}\), the associated value function \(u\) satisfies
  \[
    \begin{aligned}
      & x + \sum_{i=1}^{M} y^i C^i(t,s,v) - \bar{\mathcal{V}}C(v) - (\gamma_1 {\bar{\Delta}}^2 + \gamma_2 {\bar{\mathcal{V}}}^2 +K)(T-t)- \gamma_1^\prime {\bar{\Delta}}^2 - \gamma_2^\prime {\bar{\mathcal{V}}}^2 \\
      \leq \ & u(t,s,v,x,y,d) \leq  x + \sum_{i=1}^{M} y^i C^i(t,s,v) + \bar{\mathcal{V}} C(v) + (M m \delta \Lambda_{\max}+K)(T-t)
    \end{aligned}
  \] for all \((t,s,v,x,y,d) \in [0,T]\times \mathbb{R}_+^2 \times \mathbb{R} \times \mathcal{Q} \times \mathbb{S}^M\), 
  where \(K\) is a constant depending only on the feasible rebate policy \(F\).  
\end{proposition}
\begin{proof}
  Let \((t,s,v,x,y,d) \in [0,T]\times \mathbb{R}_+^2 \times \mathbb{R} \times \mathcal{Q} \times \mathbb{S}^M\).
  For any admissible strategy \((\alpha^{\mathrm{make}},\alpha^{\mathrm{take}})\in \mathcal{A}\),
  \[
    \begin{aligned}
      &\mathbb{E}^{t,s,v,x,y,d}\bigg[X_T + \sum_{i=1}^{M}Y^i_T C^i(T,S_T,\nu_T) - \gamma_1^\prime (\Delta^\pi_T)^2 - \gamma_2^\prime (\mathcal{V}^\pi_T)^2\bigg]\\
      \leq\ & x + \sum_{i=1}^{M} y^i C^i(t,s,v) + \sum_{i=1}^{M} \mathbb{E}^{t,s,v,x,y,d} \bigg[\int_t^T\big[(\frac{D^i_{u^-}}{2}-\delta \mathbb{I}_{\{\alpha_u^{b,i}=1\}}) d N_u^{b,i} + (\frac{D^i_{u^-}}{2}-\delta \mathbb{I}_{\{\alpha_u^{a,i}=1\}}) d N_u^{a,i} \big]\bigg]\\ 
      & + \mathbb{E}^{t,s,v,x,y,d}  \bigg[\int_t^T \big[\Delta^\pi_{u^-} dS_u +\mathcal{V}^\pi_{u^-} \big(d \nu_u-a(u,\nu_u)du \big)\big]\bigg] \\
      \leq\ & x + \sum_{i=1}^{M} y^i C^i(t,s,v) + M m \delta \Lambda_{\max}(T-t) + \bar{\mathcal{V}} C(v),
    \end{aligned}
  \]
  where the first inequality follows from \eqref{Eq 6} and \eqref{Eq 7} together with the fact that jumps of \(V\) due to impulse controls are negative,
  and the second inequality holds since both stochastic integrals \(\Delta^\pi \bullet S\) and \(\mathcal{V}^\pi \bullet \nu\) are true martingales by boundedness of \(\Delta^\pi\) and \(\mathcal{V}^\pi\), 
  and the intensities of the Cox processes \(N^{b,i},\ N^{a,i},\ i\in \mathcal{I}\) are uniformly bounded by \(\Lambda_{\max}\). Therefore,
  \[
    \begin{aligned}
      u(t,s,v,x,y,d) \leq\ & x + \sum_{i=1}^{M} y^i C^i(t,s,v) + M m \delta \Lambda_{\max}(T-t) + \bar{\mathcal{V}} C(v) + K(T-t) \\
      =\ & x + \sum_{i=1}^{M} y^i C^i(t,s,v) + \bar{\mathcal{V}} C(v) + (M m \delta \Lambda_{\max}+K)(T-t),
    \end{aligned}
  \] where \(K \coloneqq 2 M \tilde K\) for some constant \(\tilde K\) bounding \(F^{b,i}\) and \(F^{a,i},\ i \in \mathcal{I}\).\\
  Conversely, consider the particular strategy consisting of the make strategy \(\alpha^{\mathrm{make}} \equiv (0,0)_{i=1}^M\) and the take strategy \(\alpha^{\mathrm{take}}\) with \(\xi_n \equiv 0\) for all \(n\in \mathbb{N}\).\\ 
  For this strategy, with the same reasoning, but this time noting that jumps of \(V\) due to impulse controls are zero and also that the Cox processes \(N^{b,i},\ N^{a,i},\ i\in \mathcal{I}\) are increasing,
  we deduce that
  \[
    \mathbb{E}^{t,s,x,y,d}\bigg[X_T + \sum_{i=1}^{M}Y^i_T C^i(T,S_T) - \gamma_1^\prime (\Delta^\pi_T)^2 -\gamma_2^\prime (\mathcal{V}^\pi_T)^2\bigg] \geq\ x + \sum_{i=1}^{M} y^i C^i(t,s,v) - \bar{\mathcal{V}}C(v) - \gamma_1^\prime \bar{\Delta}^2 - \gamma_2^\prime \bar{\mathcal{V}}^2, 
  \]   
  and hence that 
  \[
    \begin{aligned}
      & u(t,s,v,x,y,d)\\
      \geq \ &x + \sum_{i=1}^{M} y^i C^i(t,s,v) - \bar{\mathcal{V}}C(v) -  {\bar{\Delta}}^2 \big(\gamma_1(T-t) + \gamma_1^\prime\big) -  {\bar{\mathcal{V}}}^2 \big(\gamma_2(T-t) + \gamma_2^\prime\big) - K(T-t)\\
      =\ & x + \sum_{i=1}^{M} y^i C^i(t,s,v) - \bar{\mathcal{V}}C(v) - (\gamma_1 {\bar{\Delta}}^2 + \gamma_2 {\bar{\mathcal{V}}}^2 +K)(T-t)- \gamma_1^\prime {\bar{\Delta}}^2 - \gamma_2^\prime {\bar{\mathcal{V}}}^2.
    \end{aligned}
  \]
\end{proof}

\subsection{Hamilton–Jacobi–Bellman Quasi-Variational-Inequality}{\label{Sec 3.2}}
Given a feasible rebate policy \(F \in \mathfrak{F}\), the associated control problem \eqref{Eq 13} can be solved via dynamic programming methods.
To do so, we introduce for every \(\alpha^{\mathrm{make}} = (\alpha^{b,i},\alpha^{a,i})_{i=1}^M\) 
an operator \(\mathcal{L}^{\alpha^{\mathrm{make}}}\) associated with the state variables \((S,\nu,X,Y,D)\):
\begin{equation}{\label{Eq 14}}
    \begin{aligned}
      & \mathcal{L}^{\alpha^{\mathrm{make}}}u(t,s,v,x,y,d) \\
      \coloneqq &\ a(t,v)\frac{\partial u}{\partial v} + \frac{1}{2}v \sigma^2(t,s)s^2 \frac{\partial^2 u}{\partial s^2} + \frac{1}{2}b^2(t,v)\frac{\partial^2 u}{\partial v^2} + \sigma(t,s)b(t,v)\rho s \sqrt{v} \frac{\partial^2 u}{\partial s\partial v}\\ 
      & + \sum_{i=1}^{M}\big\{\lambda^{b,i}(\alpha^{b,i},d^i)\mathbb{I}_{\{\Delta^\pi+\Delta^i \leq \bar{\Delta},\mathcal{V}^\pi + \mathcal{V}^i \leq \bar{\mathcal{V}}\}}\big[u\big(t,s,v,\Gamma^b(t,s,v,x,y,d^i,\alpha^{b,i}),d\big)-u(t,s,v,x,y,d)\big]\big\}\\
      & + \sum_{i=1}^{M} \big\{\lambda^{a,i}(\alpha^{a,i},d^i)\mathbb{I}_{\{\Delta^\pi-\Delta^i \geq -\bar{\Delta},\mathcal{V}^\pi -\mathcal{V}^i \geq -\bar{\mathcal{V}}\}} \big[u\big(t,s,v,\Gamma^a(t,s,v,x,y,d^i,\alpha^{a,i}),d\big)-u(t,s,v,x,y,d)\big]\big\}\\
      & + \sum_{d^\prime \in \mathbb{S}^M} q_{d,d^\prime}[u(t,s,v,x,y,d^\prime)-u(t,s,v,x,y,d)]
    \end{aligned}
\end{equation}
for functions \(u\) defined on \([0,T]\times \mathbb{R}_+^2 \times \mathbb{R} \times \mathcal{Q} \times \mathbb{S}^M\),
where \(\Delta^\pi = \sum_{i=1}^{M}y^i \Delta^i\), \(\mathcal{V}^\pi =\sum_{i=1}^{M} y^i \mathcal{V}^i\), and, \(\Gamma^b\) and \(\Gamma^a\) are defined from 
\([0,T]\times \mathbb{R}_+^2 \times \mathbb{R} \times \mathcal{Q} \times \mathbb{S}\times \{0,1\}\) into \(\mathbb{R}\times \mathbb{Z}^M\) by 
\begin{equation}{\label{Eq 15}}
  \begin{dcases}
    \Gamma^b(t,s,v,x,y,d^i,\alpha^{b,i}) \coloneqq & \Big(x-\pi^b\big(\alpha^{b,i},C^i(t,s,v),d^i\big),y+\mathbf{e}^i\Big),\\
    \Gamma^a(t,s,v,x,y,d^i,\alpha^{a,i}) \coloneqq & \Big(x+\pi^a\big(\alpha^{a,i},C^i(t,s,v),d^i\big),y-\mathbf{e}^i\Big),
  \end{dcases}
\end{equation}
with \((\mathbf{e}^i)_{i=1}^M\) the canonical basis of \(\mathbb{R}^M\). 

The first and the last line in \eqref{Eq 14} corresponds to the generator of the diffusion price-volatility process \((S,\nu)\) and the spread Markov chain \(D\) respectively.
The sums in the second and the third line correspond to the nonlocal operator induced by the jumps of the cash process \(X\) and the inventory processs \(Y\) 
when applying an instantaneous make strategy \(\alpha^{\mathrm{make}} = (\alpha^{b,i},\alpha^{a,i})_{i=1}^M\). 

On the other hand, consider the intervention operator \(\mathcal{M}\) of the impulse control (take strategy):
\begin{equation}{\label{Eq 16}}
  \mathcal{M} u(t,s,v,x,y,d) \coloneqq \sup\limits_{\substack{\ell \in \mathcal{E}^M:\\y+\sum_{i=1}^{M} \ell^i \mathbf{e}^i\in \mathcal{Q}}}
  u\big(t,s,v,\Gamma^{\mathrm{take}}(t,s,v,x,y,d,\ell),d\big)
\end{equation}
for functions \(u\) defined on \([0,T]\times \mathbb{R}_+^2 \times \mathbb{R} \times \mathcal{Q} \times \mathbb{S}^M\).

Here, \(\Gamma^{\mathrm{take}}\) is the impulse transaction function defined from 
\([0,T]\times \mathbb{R}_+^2 \times \mathbb{R} \times \mathcal{Q} \times \mathbb{S}^M \times \mathcal{E}^M\) into \(\mathbb{R} \times \mathcal{Q}\) by 
\begin{equation}{\label{Eq 17}}
  \Gamma^{\mathrm{take}}(t,s,v,x,y,d,\ell) \coloneqq \Big(x-\sum_{i=1}^{M}h\big(\ell^i,C^i(t,s,v),d^i\big),y+\sum_{i=1}^{M} \ell^i \mathbf{e}^i\Big).
\end{equation}

The dynamic programming equation associated with the control problem \eqref{Eq 13} is given by the following Hamilton–Jacobi–Bellman Quasi-Variational-Inequality (HJBQVI):
\begin{equation}{\label{Eq 18}}
  \min \Big\{-\frac{\partial u}{\partial t}-\sup\limits_{\alpha^i\in \mathcal{A}(d^i),\forall i}\big(\mathcal{L}^{\alpha^{\mathrm{make}}}u
  +\sum_{i=1}^{M}(F^{b,i}\mathbb{I}_{\mathcal{Q}^{b,i}}+F^{a,i}\mathbb{I}_{\mathcal{Q}^{a,i}})\big)
  +\gamma_1 (\Delta^\pi)^2+\gamma_2 (\mathcal{V}^\pi)^2,\ u-\mathcal{M} u\Big\} = 0
\end{equation}
on \([0,T)\times \mathbb{R}_+^2 \times \mathbb{R}\times \mathcal{Q} \times \mathbb{S}^M\), along with the terminal condition
\begin{equation}{\label{Eq 19}}
  u(T,s,v,x,y,d) = x+ \sum_{i=1}^{M} y^i C^i(T,s,v) - \gamma_1^\prime (\Delta^\pi)^2 - \gamma_2^\prime (\mathcal{V}^\pi)^2
\end{equation}
for \((s,v,x,y,d)\in \mathbb{R}_+^2 \times \mathbb{R} \times \mathcal{Q} \times \mathbb{S}^M\), where 
\[
  \mathcal{Q}^{b,i} \coloneqq \{\Delta^\pi + \Delta^i \leq \bar{\Delta},\mathcal{V}^\pi + \mathcal{V}^i \leq \bar{\mathcal{V}}\},\quad 
  \mathcal{Q}^{a,i} \coloneqq \{\Delta^\pi - \Delta^i \geq -\bar{\Delta},\mathcal{V}^\pi - \mathcal{V}^i \geq -\bar{\mathcal{V}}\}.
\]

By \cite[Theorem 12.8 and 12.11]{oksendal2019stochastic}, one can verify via standard arguments that the value function \(u\) is the unique viscosity solution 
to the HJBQVI \eqref{Eq 18} and \eqref{Eq 19}. 

\subsection{Solution to the market making problem}{\label{Sec 3.3}}
The optimal make-take strategy can be found by solving the HJBQVI \eqref{Eq 18} and \eqref{Eq 19}. 
Before doing that, we comment that the value function \(u\) has \(2M+3\) state variables, so when it comes to solving the dynamic programming equation,
even numerical methods are of little help when \(M\) is large. 
To circumvent this issue, we show that a dimension reduction is possible via a change of variable \big(see \eqref{Eq 25}\big),
which yields optimal make-take strategy in feedback form \big(see \eqref{Eq 30}, \eqref{Eq 31}, \eqref{Eq 32} and \eqref{Eq 33}\big).
To simplify matters, from now on we focus on the special case where the price dynamics of the underlying \(S\) are described by the Black–Scholes model,
with a constant\footnote{
In contrast to the stochastic volatility models, here we neglect the vega risk in the option portfolio,
thus assuming a constant volatility over a short horizon.
}
volatility \(\sigma>0\):
\[
  d S_t = \sigma S_t\ dW_t,\quad S_0>0.
\] 
In this case, the price function \(C^i\) solves the Black--Scholes PDE
\[
  0 = \bigg\{\partial_t + \frac{1}{2} \sigma^2 S^2\partial_{SS}\bigg\}C^i(t,S) 
\] with the terminal condition 
\[
  C^i(\tau,S) = (S-K^i)^+.
\]

Since there is no vega risk (due to a constant volatility), the market maker only needs to manage her delta exposures. 
Under the constant delta approximation as in Section \ref{Sec 2.2}, we simplify the intensities \(\Lambda^{b,i}_t\) and \(\Lambda^{a,i}_t\) of the Cox processes \(N^{b,i}\) and \(N^{a,i}\) to
\[
  \Lambda^{b,i}_t \coloneqq \lambda^{b,i}(\alpha^{b,i}_t,D^i_{t^-})\mathbb{I}_{\mathcal{Q}_{t^-}^{b,i}},\quad
  \Lambda^{a,i}_t \coloneqq \lambda^{a,i}(\alpha^{a,i}_t,D^i_{t^-})\mathbb{I}_{\mathcal{Q}_{t^-}^{a,i}},
\]   
where 
\[
  \mathcal{Q}_{t^-}^{b,i} = \{\Delta^\pi_{t^-}+\Delta^i \leq \bar{\Delta}\},\quad \mathcal{Q}_{t^-}^{a,i} = \{\Delta^\pi_{t^-}-\Delta^i \geq - \bar{\Delta}\}.
\]
Likewise, the admissible impulse controls are now such that the resulting portfolio delta remains within the bounded interval \([-\bar{\Delta},\bar{\Delta}]\) after each trade at market,
and the rebate policy now does not depend on the portfolio vega, so the process \((R_t)_{0 \leq t \leq T}\) of the cumulative rebate revenue becomes 
\[
  R_t = \sum_{i=1}^{M} \int_0^t \Big(F^{b,i}(\alpha^{b,i}_u,D^i_{u^-},\Delta^\pi_{u^-})\mathbb{I}_{\mathcal{Q}_{u^-}^{b,i}}
  + F^{a,i}(\alpha^{a,i}_u,D^i_{u^-},\Delta^\pi_{u^-})\mathbb{I}_{\mathcal{Q}_{u^-}^{a,i}}  \Big)du,\ 0 \leq t \leq T.
\] 

Following the original mean-quadratic performance criterion, the value function \(u\) for a given feasible rebate policy \(F\) is now defined by 
\begin{equation}{\label{Eq 20}}
  \begin{aligned}
    u(t,s,x,y,d;F)&\coloneqq \sup\limits_{\alpha\in \mathcal{A}}\bigg\{ \mathbb{E}^{t,s,x,y,d}\bigg[X_T + \sum_{i=1}^{M} Y_T^i C^i(T,S_T)- \gamma^\prime (\Delta^\pi_T)^2 - \gamma \int_t^T (\Delta_u^\pi)^2 du\\
    & \qquad + \sum_{i=1}^{M} \int_t^T  \Big(F^{b,i}(\alpha^{b,i}_u,D^i_{u^-},\Delta^\pi_{u^-})\mathbb{I}_{\mathcal{Q}_{u^-}^{b,i}}
    + F^{a,i}(\alpha^{a,i}_u,D^i_{u^-},\Delta^\pi_{u^-})\mathbb{I}_{\mathcal{Q}_{u^-}^{a,i}}  \Big)du \bigg]\bigg\},
  \end{aligned}
\end{equation}
where \((t,s,x,y,d)\in [0,T]\times \mathbb{R}_+ \times \mathbb{R} \times \mathcal{Q} \times \mathbb{S}^M\) with \(\mathcal{Q}\) the set of admissible inventories now defined by 
\[
  \mathcal{Q} \coloneqq \bigg\{y \in \mathbb{Z}^M:\ \sum_{i=1}^{M} y^i \Delta^i \in [-\bar \Delta,\bar \Delta]\bigg\},
\] and \(\gamma, \gamma^\prime > 0\) are risk aversion parameters on the running and terminal portfolio delta.

With the same arguments as in Section \ref{Sec 3.2}, the dynamic programming equation to the control problem \eqref{Eq 20} is given by the following HJBQVI: 
\begin{equation}{\label{Eq 21}}
  \min \bigg\{-\frac{\partial u}{\partial t}-\sup\limits_{\alpha^i\in \mathcal{A}(d^i),\forall i}\big(\mathcal{L}^{\alpha^{\mathrm{make}}}u
  +\sum_{i=1}^{M}(F^{b,i}\mathbb{I}_{\mathcal{Q}^{b,i}}  +F^{a,i}\mathbb{I}_{\mathcal{Q}^{a,i}} )\big)
  +\gamma (\Delta^\pi)^2,\ u-\mathcal{M} u\bigg\} = 0
\end{equation}
on \([0,T)\times \mathbb{R}_+ \times \mathbb{R} \times \mathcal{Q} \times \mathbb{S}^M\), along with the terminal condition 
\begin{equation}{\label{Eq 22}}
  u(T,s,x,y,d) = x + \sum_{i=1}^{M} y^i C^i(T,s) - \gamma^\prime (\Delta^\pi)^2
\end{equation} 
for \((s,x,y,d)\in \mathbb{R}_+ \times \mathbb{R} \times \mathcal{Q} \times \mathbb{S}^M\), where \(\Delta^\pi = \sum_{i=1}^{M} y^i \Delta^i\), 
\[
  \mathcal{Q}^{b,i} = \{\Delta^\pi + \Delta^i \leq \bar{\Delta}\},\quad \mathcal{Q}^{a,i} = \{\Delta^\pi - \Delta^i \geq - \bar{\Delta}\},
\]
\(\mathcal{L}^{\alpha^{\mathrm{make}}}\) is the controlled generator associated with the state variables \((S,X,Y,D)\),
and \(\mathcal{M}\) is the intervention operator of the impulse control.  

For convenience, we introduce a family of auxiliary value functions \(u_d(t,s,x,y) \coloneqq u(t,s,x,y,d)\) for \(d\in \mathbb{S}^M\). 
By abuse of notation, we shall identify the value function \(u\) with the \(\mathbb{R}^{m^M}\)-valued function \(u=(u_d)_{d\in \mathbb{S}^M}\) 
defined on \([0,T]\times \mathbb{R}_+ \times \mathbb{R} \times \mathcal{Q}\). 
Then we can explicitly reformulate the HJBQVI \eqref{Eq 21} and \eqref{Eq 22} as a coupled system of HJBQVIs for auxiliary value functions \(u_d,\ d\in \mathbb{S}^M\):
\begin{equation}{\label{Eq 23}}
  \begin{aligned}
    0 = &\min \Bigg\{-\frac{\partial u_d}{\partial t}- \frac{1}{2}\sigma^2 s^2 \frac{\partial^2 u_d}{\partial s^2}-\sum_{d^\prime \in \mathbb{S}^M}q_{d,d^\prime}(u_{d^\prime}-u_d) + \gamma (\Delta^\pi)^2 \\
    - & \sum_{i=1}^{M} \sup\limits_{\alpha^{b,i}\in \mathcal{A}^b(d^i)} \bigg\{\mathbb{I}_{\mathcal{Q}^{b,i}}  \bigg[\lambda^{b,i}_{d^i}(\alpha^{b,i}) 
    \bigg(u_d\Big(t,s,x-\pi^b_{d^i}\big(\alpha^{b,i},C^i(t,s)\big),y+\mathbf{e}^i\Big) -u_d\bigg) + F^{b,i}_{d^i}(\alpha^{b,i},\Delta^\pi)\bigg]\bigg\} \\
    - & \sum_{i=1}^{M} \sup\limits_{\alpha^{a,i}\in \mathcal{A}^a(d^i)} \bigg\{\mathbb{I}_{\mathcal{Q}^{a,i}} \bigg[\lambda^{a,i}_{d^i}(\alpha^{a,i}) 
    \bigg(u_d\Big(t,s,x+\pi^a_{d^i}\big(\alpha^{a,i},C^i(t,s)\big),y-\mathbf{e}^i\Big) -u_{d}\bigg)+F^{a,i}_{d^i}(\alpha^{a,i},\Delta^\pi)\bigg]\bigg\}, \\
    &\ u_d-\sup\limits_{\substack{\ell\in \mathcal{E}^M:\\y+\sum_{i=1}^{M} \ell^i \mathbf{e}^i\in \mathcal{Q}}}u_d\Big(t,s,x-\sum_{i=1}^{M} h\big(\ell^i,C^i(t,s),d^i\big),y+\sum_{i=1}^{M} \ell^i \mathbf{e}^i\Big)\Bigg\}
  \end{aligned}
\end{equation}
on \([0,T)\times \mathbb{R}_+ \times \mathbb{R} \times \mathcal{Q}\), 
where the variable \(d^i\) in \(\lambda^{b,i},\lambda^{a,i},\pi^b,\pi^a,F^{b,i},F^{a,i}\) is suppressed as a subscript,
along with the terminal condition 
\begin{equation}{\label{Eq 24}}
  u_d(T,s,x,y) = x+ \sum_{i=1}^{M} y^i C^i(T,s) - \gamma^\prime (\Delta^\pi)^2
\end{equation}
for \((s,x,y)\in \mathbb{R}_+ \times \mathbb{R} \times \mathcal{Q}\).

To solve the system of HJBQVIs \eqref{Eq 23} and \eqref{Eq 24}, we propose for each \(d\in \mathbb{S}^M\) the ansatz 
\begin{equation}{\label{Eq 25}}
  u_d(t,s,x,y) = x + \sum_{i=1}^{M}y^i C^i(t,s) + \varphi_d(t,\Delta^\pi)
\end{equation}
via a change of variable \(\Delta^\pi = \sum_{i=1}^{M} y^i \Delta^i\). 
We observe that 
 \(\varphi_d\) only needs to be defined on \([0,T]\times ([-\bar \Delta,\bar \Delta]\cap \mathcal{Z})\). Here 
\[
  \mathcal{Z} = \sum_{i=1}^{M} \Delta^i \mathbb{Z} \coloneqq \Big\{\sum_{i=1}^{M} \Delta^i y^i:\ y^i\in \mathbb{Z},\ i=1,\cdots,M\Big\}
\]
denotes the Minkowski sum, which gives the set of achievable portfolio deltas, using investments in (whole numbers of indivisible) options.

\begin{remark}
  We refer to the function \(\varphi =(\varphi_d)_{d\in \mathbb{S}^M}\) the \emph{reduced value function}, 
  or just the \emph{value function} for short when there is no ambiguity. 
  Under the constant delta approximation, this ansatz reduces the dimension of the system (excluding the bid-ask spreads) from \(M+2\) to \(2\).
  This reduction highlights the value of the constant parameter approximation (see Section \ref{Sec 2.3}), which is to reduce the inventory variables to a single penalty variable;
  in our case this corresponds to the portfolio delta, but in general can indicate other sources of portfolio risk (for instance the vega risk in a general LSV model).
\end{remark}

The interpretation of the ansatz \eqref{Eq 25} is fairly intuitive. 
The first two terms on the right-hand side account for the market maker's portfolio value. 
The last term does not depend on the underlying price, but on the bid-ask spreads 
as well as the inventory holdings (only through the portfolio delta).
It is a correction term that reconciles the benefit of rebate revenue and the penalty on residual portfolio delta. 
Indeed, the ansatz \eqref{Eq 25} is motivated by the growth condition in \hyperref[Prop 3.1]{Proposition 3.1}, which says 
the correction term \(\varphi_d(t,\Delta^\pi)\) is subject to the bounds
\[
  - (\gamma {\bar{\Delta}}^2 +K)(T-t) - \gamma^\prime \bar{\Delta}^2 \leq \varphi_d(t,\Delta^\pi) \leq (M m \delta \Lambda_{\max}+K)(T-t).
\]

Upon a substitution of the ansatz \eqref{Eq 25} into the coupled system of HJBQVIs \eqref{Eq 23} and \eqref{Eq 24}, we obtain a reduced coupled system of 
HJBQVIs for the \(2\)-dimensional functions \(\varphi_d,\ d\in \mathbb{S}^M\):
\begin{align}{\label{Eq 26}}
  0 = \min \Bigg\{ & -\frac{\partial \varphi_d}{\partial t}(t,\Delta^\pi) -\sum_{d^\prime \in \mathbb{S}^M}q_{d,d^\prime}\big(\varphi_{d^\prime}(t,\Delta^\pi)-\varphi_d(t,\Delta^\pi)\big) + \gamma(\Delta^\pi)^2 \nonumber \\ \displaybreak[1]
  & - \sum_{i=1}^{M} \sup\limits_{\alpha^{b,i}\in \mathcal{A}^b(d^i)} \Big\{\mathbb{I}_{\mathcal{Q}^{b,i}}  \big[H^{b,i}_{d^i}(\alpha^{b,i},\Delta^\pi,t,d^{\bar i}) + F^{b,i}_{d^i}(\alpha^{b,i},\Delta^\pi)\big]\Big\}\nonumber \\
  & - \sum_{i=1}^{M} \sup\limits_{\alpha^{a,i}\in \mathcal{A}^a(d^i)} \Big\{\mathbb{I}_{\mathcal{Q}^{a,i}}  \big[H^{a,i}_{d^i}(\alpha^{a,i},\Delta^\pi,t,d^{\bar i}) + F^{a,i}_{d^i}(\alpha^{a,i},\Delta^\pi)\big] \Big\},\nonumber \\
  & \varphi_d(t,\Delta^\pi)-\sup\limits_{\substack{\ell\in \mathcal{E}^M:\\\Delta^\pi + \sum_{i=1}^{M}\ell^i \Delta^i \in [-\bar \Delta,\bar \Delta]}}\bigg\{-\sum_{i=1}^{M}\bigg(|\ell^i|\Big(\frac{d^i}{2}+\epsilon\Big)+\tilde{\epsilon}\mathbb{I}_{\{\ell^i \neq 0\}}\bigg)+\varphi_d\bigg(t,\Delta^\pi+\sum_{i=1}^{M}\ell^i \Delta^i\bigg)\bigg\} \Bigg\}
\end{align}
on \([0,T)\times ([-\bar \Delta,\bar \Delta]\cap \mathcal{Z})\), along with the terminal condition 
\begin{equation}{\label{Eq 27}}
  \varphi_d(T,\Delta^\pi) = - \gamma^\prime (\Delta^\pi)^2,
\end{equation} 
where the notation \(d^{\bar i}\) means the vector \(d\) except for the \(i\)-th coordinate, and for the sake of notational simplicity we define
\begin{equation}{\label{Eq 28}}
  H^{b,i}_{d^i}(\alpha^{b,i},\Delta^\pi,t,d^{\bar i}) \coloneqq \lambda^{b,i}_{d^i} (\alpha^{b,i})\Big[\Big(\frac{d^i}{2}-\delta \alpha^{b,i}\Big)+\varphi_d(t,\Delta^\pi+\Delta^i)-\varphi_d(t,\Delta^\pi)\Big],
\end{equation}
and 
\begin{equation}{\label{Eq 29}}
  H^{a,i}_{d^i}(\alpha^{a,i},\Delta^\pi,t,d^{\bar i}) \coloneqq \lambda^{a,i}_{d^i} (\alpha^{a,i})\Big[\Big(\frac{d^i}{2}-\delta \alpha^{a,i}\Big)+\varphi_d(t,\Delta^\pi-\Delta^i)-\varphi_d(t,\Delta^\pi)\Big].
\end{equation}

Once we solve the coupled system of HJBQVIs \eqref{Eq 26} and \eqref{Eq 27} and get the value function \(\varphi\), 
the optimal make strategy \(\alpha^{\mathrm{make},*}\) is given by the pairs 
\(\big(\alpha^{b,i,*},\alpha^{a,i,*}\big)_{i=1}^M\): 
\[
  \begin{dcases}
    \alpha^{b,i,*}_t = \chi^{b,i}(t,\Delta^\pi_{t^-},D_{t^-}),\quad t\in [0,T],\\
    \alpha^{a,i,*}_t = \chi^{a,i}(t,\Delta^\pi_{t^-},D_{t^-}),\quad t\in [0,T],
  \end{dcases}
\] where 
\begin{equation}{\label{Eq 30}}
  \chi^{b,i}(t,\Delta^\pi,d) = 
  \begin{dcases}
    1,\quad &\text{ if } d^i > \delta,\ \Delta^\pi +\Delta^i \leq \bar \Delta, \text{ and }\\ 
    & H^{b,i}_{d^i}(1,\Delta^\pi,t,d^{\bar i}) + F^{b,i}_{d^i}(1,\Delta^\pi)
    > H^{b,i}_{d^i}(0,\Delta^\pi,t,d^{\bar i}) + F^{b,i}_{d^i}(0,\Delta^\pi);\\
    0,\quad &\text{ otherwise},
  \end{dcases}
\end{equation}
and 
\begin{equation}{\label{Eq 31}}
  \chi^{a,i}(t,\Delta^\pi,d) = 
  \begin{dcases} 
    1,\quad &\text{ if } d^i > \delta,\ \Delta^\pi -\Delta^i \geq -\bar \Delta, \text{ and }\\ 
    & H^{a,i}_{d^i}(1,\Delta^\pi,t,d^{\bar i}) + F^{a,i}_{d^i}(1,\Delta^\pi) 
    > H^{a,i}_{d^i}(0,\Delta^\pi,t,d^{\bar i}) + F^{a,i}_{d^i}(0,\Delta^\pi);\\
    0,\quad &\text{ otherwise}.
  \end{dcases}
\end{equation}
Further, the optimal take strategy \(\alpha^{\mathrm{take},*}\) is given by the sequence of impulse trades \((\tau_n^*,\xi_n^*)_{n=1}^\infty\), inductively defined by  
\[
  \begin{dcases}
    \tau_{n+1}^* \coloneqq \inf\big\{t>\tau_n^*:(t,\Delta^\pi_{t^-},D_{t^-})\notin \mathcal{C}\big\},\quad n \geq 0;\\
    \xi_{n+1}^* \coloneqq \hat{\ell}\big(\tau_{n+1}^*,\Delta^\pi_{(\tau_{n+1}^*)^-},D_{(\tau_{n+1}^*)^-}\big),\quad n \geq 0
  \end{dcases}
\] with \(\tau_0^* \equiv 0\), where \(\mathcal{C}\) is the continuation region:
\begin{equation}{\label{Eq 32}}
  \begin{aligned}
    \mathcal{C} = \bigg\{&(t,\Delta^\pi,d)\in [0,T]\times ([-\bar \Delta,\bar \Delta]\cap \mathcal{Z}) \times \mathbb{S}^M:\\
    &\varphi_d(t,\Delta^\pi) > \sup\limits_{\substack{\ell \in \mathcal{E}^M:\\\Delta^\pi + \sum_{i=1}^{M}\ell^i \Delta^i \in [-\bar \Delta,\bar \Delta]}}\bigg[-\sum_{i=1}^{M}\bigg(|\ell^i|\Big(\frac{d^i}{2}+\epsilon\Big)+\tilde{\epsilon}\mathbb{I}_{\{\ell^i \neq 0\}}\bigg)+\varphi_d\bigg(t,\Delta^\pi+\sum_{i=1}^{M}\ell^i \Delta^i\bigg)\bigg]\bigg\}
  \end{aligned}
\end{equation}
and 
\begin{equation}{\label{Eq 33}}
  \hat{\ell}(t,\Delta^\pi,d) \in \text{argmax} \Bigg(-\sum_{i=1}^{M}\bigg(|\ell^i|\Big(\frac{d^i}{2}+\epsilon\Big)+\tilde{\epsilon}\mathbb{I}_{\{\ell^i \neq 0\}}\bigg)+\varphi_d\bigg(t,\Delta^\pi+\sum_{i=1}^{M}\ell^i \Delta^i\bigg)\Bigg)
\end{equation}
over \(\{\ell\in \mathcal{E}^M:\Delta^\pi + \sum_{i=1}^{M}\ell^i \Delta^i \in [-\bar \Delta,\bar \Delta]\}\).   

In summary, we are able (by solving the HJBQVI) to identify the optimal make-take strategy \((\alpha^{\mathrm{make},*},\alpha^{\mathrm{take},*})\). A simple financial interpretation of this strategy is the following.
For the make strategy,
when the spread exceeds one tick, it is optimal to quote aggressively if the resulting fee rebates fully offset the potential loss from improving the best price (unless a trade would lead the portfolio delta to exceed the bound);
otherwise, it is optimal to quote conservatively.
For the take strategy, it is optimal to send a market order only when the portfolio delta is reduced effectively after intervention,
with an order size that simultaneously minimizes the transaction costs and the penalty on the residual portfolio delta. 
We will see (in \hyperref[MFR]{Remark 5}) that this simple intuitive construction results in a heuristic approximation of the optimal rebate policy.

\subsection{A three-step scheme for rebate design}{\label{Sec 3.4}}
As a business, an exchange may face a variety of objectives in developing relationships with the users of its platform. In much of the academic literature (see \cite{euch2021optimal}, \cite{baldacci2021optimal}, and \cite{baldacci2023market}), these are reduced to optimizing fee income over a short time scale. Over longer time scales, liquidity attraction is a primary concern of the exchange, since lack of sufficient liquidity on its platform may lead to the loss of clients or members, resulting in much larger losses than a reduction in short-term fee income. 
Therefore, in order to remain competitive from a liquidity perspective, the exchange may aim to design a rebate policy that incentivizes market makers to provide more liquidity, and wishes to do so in a cost-effective way.
Specifically, we set as our out goal \emph{to design a rebate policy for which the optimal make strategy of the market maker 
is to consistently improve the best available prices by one tick whenever possible}, with a minimal level of rebate offered. 

In the previous section, we considered the market making problem under a given rebate policy and derived the corresponding optimal make-take strategy in feedback form;
this enables the exchange to anticipate the best response of the market maker to any specific rebate policy.
In \cite{euch2021optimal} and its extensions, this approach is used in the \emph{second best} framework where the fee rebates do not depend on the efforts made by the market maker,
and the exchange needs to solve a Principal-Agent problem to determine the best rebate policy that maximizes its fee income. 
However, this does not focus on how the exchange can improve market liquidity with the presence of rebate incentives. 
More importantly, the Principal-Agent framework requires the same time scale, which does not hold in our setting:
the market maker maximizes her net PnL over a short horizon, while the exchange aims to improve market liquidity, in order to satisfy long-term business objectives. 
In contrast, our \emph{first best} framework assumes that the exchange can perfectly monitor the actual behavior of the market maker, which is a reasonable assumption in practice, as the limit orders posted by the market maker are known to the exchange. 
As we allow the time scales of the market maker and the exchange to differ, we determine a rebate policy which depends on the risk preference of the market maker, but otherwise does not depend on time. In particular, the exchange is able to design a rebate policy that aligns the market maker's behavior with its liquidity attraction objective, rather than simply focussing on fee income.

In what follows, we generalize this idea by enabling the exchange to 
Based on these targets, the exchange then designs individualized rebate policies to meet them.
For consistency, we continue to work with the Black–Scholes case as in Section \ref{Sec 3.3}.

Since the market maker can only quote aggressively when the spread is larger than one tick,
the objective of the exchange is to determine a binary rebate rule so that the best response of the market maker is to quote aggressively 
if the exchange wishes to narrow the spreads and to quote conservatively otherwise.
Mathematically, the exchange aims to design a rebate policy for which the optimal make strategy \(\alpha^{\mathrm{make},*}\) of the market maker is given by 
\begin{equation}{\label{Eq 34}}
  \begin{dcases}
    \alpha^{b,i,*}_t = \overline{\alpha^{b,i}} \mathbb{I}_{\{D^i_{t^-} > \delta\}},\quad t\in [0,T],\quad i\in \mathcal{I},\\
    \alpha^{a,i,*}_t = \overline{\alpha^{a,i}} \mathbb{I}_{\{D^i_{t^-} > \delta\}},\quad t\in [0,T],\quad i\in \mathcal{I},
  \end{dcases}
\end{equation}
where \(\overline{\alpha^{b,i}},\overline{\alpha^{a,i}}\) are constants valued in \(\{0,1\}\), 
reflecting the individualized liquidity targets set by the exchange.
In particular, if the exchange requires the market maker to consistently improve the best prices whenever possible,
it can set \(\overline{\alpha^{b,i}}=\overline{\alpha^{a,i}}=1\) for all \(i\in \mathcal{I}\).      

To this end, we propose a three-step rebate design scheme.
The main idea is to \emph{structure a rebate policy such that any other make strategy is suboptimal to \eqref{Eq 34}}. 
The details are as follows.
\begin{enumerate}
  \item [(a)] First, the exchange specifies the values
  \[
    F^{b,i}(\overline{\alpha^{b,i}},d^i,\Delta^\pi) \text{ and }  F^{a,i}(\overline{\alpha^{a,i}},d^i,\Delta^\pi)
  \]
  for all \(i\in \mathcal{I}, d^i > \delta, \Delta^\pi \in [-\bar{\Delta},\bar{\Delta}]\cap \mathcal{Z}\), and sets
  \[
    F^{b,i}(0,d^i,\Delta^\pi) = F^{a,i}(0,d^i,\Delta^\pi)\equiv 0
  \]
  for all \(i\in \mathcal{I}, d^i = \delta, \Delta^\pi \in [-\bar{\Delta},\bar{\Delta}]\cap \mathcal{Z}\),
  but leaves for all \(i\in \mathcal{I}, d^i > \delta, \Delta^\pi \in [-\bar{\Delta},\bar{\Delta}]\cap \mathcal{Z}\) the values
  \[
    F^{b,i}(1-\overline{\alpha^{b,i}},d^i,\Delta^\pi) \text{ and }  F^{a,i}(1-\overline{\alpha^{a,i}},d^i,\Delta^\pi)
  \]
  to be determined later.\\
  That is, the exchange prescribes the fee rebates that it promises to pay if the market maker follows the make strategy \eqref{Eq 34}. 
  The zero fee rebates for quoting at the touch when the spread is one tick are set for convenience only. 
  From the exchange's perspective, the market maker does not need compensations for adding liquidity to options with tight spreads already. 
  \item [(b)] Second, we assume from the outset what we want to show,
  namely that the rebate policy \(F = (F^{b,i},F^{a,i})_{i=1}^M\) specified in (a) leads to the optimal make strategy \eqref{Eq 34}. 
  In view of the optimal feedback functions \eqref{Eq 30} and \eqref{Eq 31}, we may eliminate the suprema in the coupled system of HJBQVIs \eqref{Eq 26} 
  and hence reduce it to a coupled system of Quasi-Variational-Inequalities (QVIs):
  \begin{align}{\label{Eq 35}}
    &0 = \min \Bigg\{-\frac{\partial \varphi_d}{\partial t}(t,\Delta^\pi) -\sum_{d^\prime \in \mathbb{S}^M}q_{d,d^\prime}\big(\varphi_{d^\prime}(t,\Delta^\pi)-\varphi_d(t,\Delta^\pi)\big) + \gamma(\Delta^\pi)^2 \nonumber \\ \displaybreak[1]
    & - \sum_{i=1}^{M} \mathbb{I}_{\mathcal{Q}^{b,i}} \Big\{\mathbb{I}_{\{d^i=\delta\}} \big[H^{b,i}_{d^i}(0,\Delta^\pi,t,d^{\bar i}) + F^{b,i}_{d^i}(0,\Delta^\pi)\big]
    + \mathbb{I}_{\{d^i>\delta\}}\big[H^{b,i}_{d^i}(\overline{\alpha^{b,i}},\Delta^\pi,t,d^{\bar i}) + F^{b,i}_{d^i}(\overline{\alpha^{b,i}},\Delta^\pi)\big]\Big\} \nonumber \\ 
    & - \sum_{i=1}^{M} \mathbb{I}_{\mathcal{Q}^{a,i}} \Big\{\mathbb{I}_{\{d^i=\delta\}} \big[H^{a,i}_{d^i}(0,\Delta^\pi,t,d^{\bar i}) + F^{a,i}_{d^i}(0,\Delta^\pi)\big]
    + \mathbb{I}_{\{d^i>\delta\}}\big[H^{a,i}_{d^i}(\overline{\alpha^{a,i}},\Delta^\pi,t,d^{\bar i}) + F^{a,i}_{d^i}(\overline{\alpha^{a,i}},\Delta^\pi)\big]\Big\}, \nonumber \\ 
    & \varphi_d(t,\Delta^\pi)-\sup\limits_{\substack{\ell\in \mathcal{E}^M:\\\Delta^\pi + \sum_{i=1}^{M}\ell^i \Delta^i \in [-\bar \Delta,\bar \Delta]}}\bigg\{-\sum_{i=1}^{M}\bigg(|\ell^i|\Big(\frac{d^i}{2}+\epsilon\Big)+\tilde{\epsilon}\mathbb{I}_{\{\ell^i \neq 0\}}\bigg)+\varphi_d\bigg(t,\Delta^\pi+\sum_{i=1}^{M}\ell^i \Delta^i\bigg)\bigg\} \Bigg\}
  \end{align}
  on \([0,T) \times ([-\bar \Delta, \bar \Delta] \cap \mathcal{Z})\), along with the terminal condition
  \begin{equation}{\label{Eq 36}}
    \varphi_d(T,\Delta^\pi) = - \gamma^\prime (\Delta^\pi)^2.
  \end{equation}  
  Solving the coupled system of QVIs \eqref{Eq 35} and \eqref{Eq 36}, we obtain the candidate value function \(\varphi\). 
  \item [(c)] Finally, we verify that for the rebate policy specified in (a), the candidate value function \(\varphi\) obtained in (b), 
  is a solution to the associated market making problem, for which the optimal make strategy is given by \eqref{Eq 34}. 
  To do so, it is enough to check 
  \begin{enumerate}
    \item [i.] for every fixed \(i\in \mathcal{I}, d^i > \delta\), and \(\Delta^\pi \in [-\bar \Delta,\bar \Delta] \cap \mathcal{Z}\) with \(\Delta^\pi + \Delta^i \leq \bar \Delta\) that
    \begin{equation}{\label{Eq 37}}
    H^{b,i}_{d^i}(\overline{\alpha^{b,i}},\Delta^\pi,t,d^{\bar i}) + F^{b,i}_{d^i}(\overline{\alpha^{b,i}},\Delta^\pi) \geq H^{b,i}_{d^i}(1-\overline{\alpha^{b,i}},\Delta^\pi,t,d^{\bar i}) + F^{b,i}_{d^i}(1-\overline{\alpha^{b,i}},\Delta^\pi)
    \end{equation} 
    holds for all \((t,d^{\bar i})\in [0,T) \times \mathbb{S}^{M-1}\),
    \item [ii.] and for every fixed \(i\in \mathcal{I}, d^i > \delta\), and \(\Delta^\pi \in [-\bar \Delta,\bar \Delta] \cap \mathcal{Z}\) with \(\Delta^\pi - \Delta^i \geq -\bar \Delta\) that
    \begin{equation}{\label{Eq 38}}
    H^{a,i}_{d^i}(\overline{\alpha^{a,i}},\Delta^\pi,t,d^{\bar i}) + F^{a,i}_{d^i}(\overline{\alpha^{a,i}},\Delta^\pi) \geq H^{a,i}_{d^i}(1-\overline{\alpha^{a,i}},\Delta^\pi,t,d^{\bar i}) + F^{a,i}_{d^i}(1-\overline{\alpha^{a,i}},\Delta^\pi)
    \end{equation} 
    holds for all \((t,d^{\bar i})\in [0,T) \times \mathbb{S}^{M-1}\),
  \end{enumerate}
  where the notation \(d^{\bar i}\) means the vector \(d\) except for the \(i\)-th coordinate,  
  \(H^{b,i}\) and \(H^{a,i}\) are given by \eqref{Eq 28} and \eqref{Eq 29}.
  Rearranging the terms in \eqref{Eq 37} and \eqref{Eq 38}, it suffices to check
  \begin{enumerate}
    \item [i.] for every fixed \(i\in \mathcal{I}, d^i > \delta\), and \(\Delta^\pi \in [-\bar \Delta,\bar \Delta] \cap \mathcal{Z}\) with \(\Delta^\pi + \Delta^i \leq \bar \Delta\) that 
    the unspecified value \(F^{b,i}_{d^i}(1-\overline{\alpha^{b,i}},\Delta^\pi)\) satisfies the upper bounds: 
    \begin{equation}{\label{Eq 39}}
      \begin{aligned}
        F^{b,i}_{d^i}(1-\overline{\alpha^{b,i}},\Delta^\pi) \leq \ & H^{b,i}_{d^i}(\overline{\alpha^{b,i}},\Delta^\pi,t,d^{\bar i}) + F^{b,i}_{d^i}(\overline{\alpha^{b,i}},\Delta^\pi) - H^{b,i}_{d^i}(1-\overline{\alpha^{b,i}},\Delta^\pi,t,d^{\bar i})\\
        \eqqcolon \ & G^{b,i}_{d^i} (\overline{\alpha^{b,i}},\Delta^\pi,t,d^{\bar i}) 
      \end{aligned}
    \end{equation} 
    for all \((t,d^{\bar i}) \in [0,T) \times \mathbb{S}^{M-1}\), 
    \item [ii.] and for every fixed \(i\in \mathcal{I}, d^i > \delta\), and \(\Delta^\pi \in [-\bar \Delta,\bar \Delta] \cap \mathcal{Z}\) with \(\Delta^\pi - \Delta^i \geq -\bar \Delta\) that
    the unspecified value \(F^{a,i}_{d^i}(1-\overline{\alpha^{a,i}},\Delta^\pi)\) satisfies the upper bounds: 
    \begin{equation}{\label{Eq 40}}
      \begin{aligned}
        F^{a,i}_{d^i}(1-\overline{\alpha^{a,i}},\Delta^\pi) \leq \ & H^{a,i}_{d^i}(\overline{\alpha^{a,i}},\Delta^\pi,t,d^{\bar i}) + F^{a,i}_{d^i}(\overline{\alpha^{a,i}},\Delta^\pi) - H^{a,i}_{d^i}(1-\overline{\alpha^{a,i}},\Delta^\pi,t,d^{\bar i})\\
        \eqqcolon \ & G^{a,i}_{d^i} (\overline{\alpha^{a,i}},\Delta^\pi,t,d^{\bar i}) 
      \end{aligned}
    \end{equation}  
    for all \((t,d^{\bar i}) \in [0,T) \times \mathbb{S}^{M-1}\).
  \end{enumerate}
  By the verification theorem, if these unspecified values are chosen to satisfy the upper bounds in \eqref{Eq 39} and \eqref{Eq 40} 
  (for all \((t,d^{\bar i})\in [0,T) \times \mathbb{S}^{M-1}\)) together with the property \eqref{Eq 8}, 
  then the candidate value function \(\varphi\) found in (b) is the true value function for the resulting feasible rebate policy \(F\), 
  for which the optimal make strategy is given by \eqref{Eq 34}.\\
  Moreover, given \(i \in \mathcal{I}, d^i>\delta,\Delta^\pi \in [-\bar \Delta, \bar \Delta] \cap \mathcal{Z}\) with \(\Delta^\pi + \Delta^i \leq \bar \Delta\),
  the least upper bound 
  \begin{equation}{\label{Eq 41}}
    \text{LUB}^{b,i}_{d^i}(\Delta^\pi) \coloneqq \inf\limits_{(t,d^{\bar i})\in [0,T)\times \mathbb{S}^{M-1}} G^{b,i}_{d^i} (\overline{\alpha^{b,i}},\Delta^\pi;t,d^{\bar i})
  \end{equation}
  represents the highest possible fee rebate that the exchange should offer to the market maker 
  for the reverse make strategy \(\alpha^{b,i,\mathrm{reverse}} \coloneqq 1-\overline{\alpha^{b,i}}\) on the bid side,
  similarly, given \(i \in \mathcal{I}, d^i>\delta,\Delta^\pi \in [-\bar \Delta, \bar \Delta] \cap \mathcal{Z}\) with \(\Delta^\pi - \Delta^i \geq -\bar \Delta\),
  the least upper bound 
  \begin{equation}{\label{Eq 42}}
    \text{LUB}^{a,i}_{d^i}(\Delta^\pi) \coloneqq \inf\limits_{(t,d^{\bar i})\in [0,T)\times \mathbb{S}^{M-1}} G^{a,i}_{d^i} (\overline{\alpha^{a,i}},\Delta^\pi;t,d^{\bar i})
  \end{equation}
  represents the highest possible fee rebate that the exchange should offer to the market maker 
  for the reverse make strategy \(\alpha^{a,i,\mathrm{reverse}} \coloneqq 1-\overline{\alpha^{a,i}}\) on the ask side. 
  In other words, for each option \(C^i\), given its spread \(d^i\) and the portfolio delta \(\Delta^\pi\) held by the market maker,
  if the bid (resp.\ ask) side fee rebates \(F^{b,i}_{d^i}(\alpha^{b,i,\mathrm{reverse}},\Delta^\pi)\) (resp. \(F^{a,i}_{d^i}(\alpha^{a,i,\mathrm{reverse}},\Delta^\pi)\))
  for the reverse make strategy exceed the least upper bound \(\text{LUB}^{b,i}_{d^i}(\Delta^\pi)\) (resp. \(\text{LUB}^{a,i}_{d^i}(\Delta^\pi)\) ),
  then the market maker lose incentives to adopt the make strategy \(\overline{\alpha^{b,i}}\) (resp. \(\overline{\alpha^{a,i}}\)) that the exchange wishes her to follow.
\end{enumerate}    

\begin{remark}
  In a similar vein to \cite{euch2021optimal}, if the exchange aims to increase the short-term transaction volumes so as to increase the transaction fee revenues,
  the optimal rebate policy that maximizes its short-term PnL can be found by solving the following utility maximization problem: 
  \[
    \sup\limits_{(\overline{\alpha^{b,i}},\overline{\alpha^{a,i}})_{i=1}^M} \mathbb{E}\Big[-\exp\Big(-\eta\big(\sum_{i=1}^{M}(\epsilon+\tilde{\epsilon})(N_T^{b,i,\alpha^{\mathrm{make},*}}+N_T^{a,i,\alpha^{\mathrm{make},*}})-R_T^{\alpha^{\mathrm{make},*}}   \big) \Big)\Big],
  \] where \(\eta>0\) is the exchange's risk aversion parameter, \((\overline{\alpha^{b,i}},\overline{\alpha^{b,i}})_{i=1}^M\) is the make strategy anticipated by the exchange, 
  \(\alpha^{\mathrm{make},*}\) is the optimal make strategy given by \eqref{Eq 34}, \(N^{b,i,\alpha^{\mathrm{make},*}},N^{a,i,\alpha^{\mathrm{make},*}}\)
  are the market maker's execution processes under the optimal make strategy \(\alpha^{\mathrm{make},*}\), and 
  \(R^{\alpha^{\mathrm{make},*}}\) is the market maker's cumulative rebate revenue process under the optimal make strategy \(\alpha^{\mathrm{make},*}\)
  with the fee rebates \(F^{b,i}_{d^i}(\overline{\alpha^{b,i}},\Delta^\pi)\) and \(F^{a,i}_{d^i}(\overline{\alpha^{a,i}},\Delta^\pi)\) specified in (a).

  To this end, the exchange can first apply the three-step scheme to each candidate make strategy \((\overline{\alpha^{b,i}},\overline{\alpha^{a,i}})_{i=1}^M\) 
  to construct a feasible rebate policy that ensures the market maker will indeed follow the anticipated strategy,
  and then solve the utility maximization problem above to determine the optimal rebate policy that attains the supremum.
\end{remark}

\begin{remark}
  The rebate design presented above assumes that the exchange can monitor the position (in particular the net delta exposure) of the market maker, and use this when determining the rebate. This can be extended in various ways:

  A first extension is to observe that, if the exchange wishes to ensure a market maker consistently submits aggressive limit orders, then it needs to set the rebate sufficiently high, but setting it higher than a certain threshold has no further impact (beyond the exchange incurring additional costs). Therefore, the rebate value calculated using the method above, maximized over admissible deltas, will provide a lower bound on how large a constant rebate would need to be in order to ensure consistent market maker behaviour.

  A second extension is to set a bound on positions where the exchange accepts that market makers may not be willing to quote aggressively when their positions in delta become too large. This can be used as a target policy (rather than the uniformly aggressive policy above), and the three-step scheme followed through to determine the optimal rebate (either delta-dependent, or constant as in the first extension). This can then be explored iteratively to confirm the impact of this rebate on market makers' strategies, and hence to determine more refined rebate schemes.
\end{remark}

\section{Numerical scheme}{\label{Sec 4}}
Using dynamic programming methods, we have shown that, in the Black–Scholes model with a constant delta approximation, 
the problems of market making and rebate design boil down to solving an associated coupled system of HJBQVIs.
In this section, we propose a semi-implicit scheme to solve the coupled system of HJBQVIs \eqref{Eq 26} and \eqref{Eq 27} of the market making problem, 
which includes, as a special case, the coupled system of QVIs \eqref{Eq 35} and \eqref{Eq 36} of the rebate design problem. 
To this end, we implement a penalty method to transform the coupled system of HJBQVIs into a system of penalized HJB equations, and then solve it via a semi-implicit Euler scheme. 

\subsection{Penalty method}
Following \cite{reisinger2020error}, the coupled system of HJBQVIs \eqref{Eq 26} and \eqref{Eq 27} can be numerically approximated by a penalty scheme.
For \(\rho >0\), denote by \(\varphi^\rho = (\varphi^\rho_d)_{d \in \mathbb{S}^M}\) the solution to the coupled system of \emph{penalized} HJB equations:
\begin{align}{\label{Eq 43}}
  0 = - &\frac{\partial \varphi^\rho_d}{\partial t}(t,\Delta^\pi) -\sum_{d^\prime \in \mathbb{S}^M}q_{d,d^\prime}\big(\varphi^\rho_{d^\prime}(t,\Delta^\pi)-\varphi^\rho_d(t,\Delta^\pi)\big) + \gamma(\Delta^\pi)^2 \nonumber \\ \displaybreak[1]
  - & \sum_{i=1}^{M} \sup\limits_{\alpha^{b,i}\in \mathcal{A}^b(d^i)} \big\{\mathbb{I}_{\mathcal{Q}^{b,i}} \big[H^{b,i}_{d^i}(\alpha^{b,i},\Delta^\pi,t,d^{\bar i}) + F^{b,i}_{d^i}(\alpha^{b,i},\Delta^\pi)\big]\big\} \nonumber \\
  - & \sum_{i=1}^{M} \sup\limits_{\alpha^{a,i}\in \mathcal{A}^a(d^i)} \big\{\mathbb{I}_{\mathcal{Q}^{a,i}} \big[H^{a,i}_{d^i}(\alpha^{a,i},\Delta^\pi,t,d^{\bar i}) + F^{a,i}_{d^i}(\alpha^{a,i},\Delta^\pi)\big]\big\} \nonumber \\
  - & \rho \Big(\sup\limits_{\substack{\ell\in \mathcal{E}^M:\\\Delta^\pi + \sum_{i=1}^{M}\ell^i \Delta^i \in [-\bar \Delta,\bar \Delta]}}\big\{-\sum_{i=1}^{M}\big(|\ell^i|(\frac{d^i}{2}+\epsilon)+\tilde{\epsilon}\mathbb{I}_{\{\ell^i \neq 0\}}\big)+\varphi^\rho_d(t,\Delta^\pi+\sum_{i=1}^{M}\ell^i \Delta^i)\big\} - \varphi^\rho_d(t,\Delta^\pi)\Big)^+
\end{align}
on \([0,T)\times ([-\bar \Delta, \bar \Delta] \cap \mathcal{Z})\), along with the terminal condition
\begin{equation}{\label{Eq 44}}
  \varphi^\rho_d(T,\Delta^\pi) = - \gamma^\prime (\Delta^\pi)^2.
\end{equation}

Comparing equations \eqref{Eq 26} and \eqref{Eq 43}, the last line of \eqref{Eq 43}, which penalizes the obstacle term in \eqref{Eq 26},
can be regarded as the effect of an additional generator of a compound Poisson process for market order executions, 
acting on the function \(\varphi^\rho_d\), with \(\rho\) indicating the intensity of the execution process. 
In financial language, this term has an interpretation of a \emph{latency effect}. 
The larger the penalty parameter \(\rho\), the faster the market orders are executed, or in other words, the lower the latency effect.

As shown in \cite{reisinger2020error}, the solution \(\varphi^\rho\) to the system of penalized HJB equations \eqref{Eq 43} and \eqref{Eq 44} converges to 
the exact solution \(\varphi\) to the system of HJBQVIs \eqref{Eq 26} and \eqref{Eq 27}, as \(\rho \to \infty\).

Therefore, for \(\rho\) large enough, we may think of \(\varphi^\rho\) as a valid approximation of \(\varphi\).

\subsection{Discretization}
To solve the system of penalized HJB equations \eqref{Eq 43} and \eqref{Eq 44}, we implement an Euler timestepping scheme 
that is explicit with respect to the generators in the first three lines of \eqref{Eq 43}, but for stability reasons, 
implicit with respect to the penalty term for the obstacle. 

We first consider a time discretization of the interval \([0,T]\) with uniform time step \(h = \frac{T}{N}\) for some \(N\in \mathbb{N}\).  
For the space discretization, we may assume\footnote{
The first assumption is innocent as the total delta limit is typically determined by the total delta exposure of the maximal positions in traded options. 
The second assumption is more technical. In the Black–Scholes model, the delta of a European call option can take any value in \([0,1]\),  
which of course includes irrational numbers. This assumption, while potentially controversial, 
is imposed to ensure that \([-\bar \Delta,\bar \Delta] \cap \mathcal{Z}\) is finite
(see \hyperref[Cor A.3]{Corollary A.3} for what may happen when this assumption does not hold). 
In practice, an irrational \(\Delta^i\) can be approximated by a rational number to a desired level of precision.}
that \(\bar \Delta\in \mathcal{Z}\) and that \(\Delta^i \in \mathbb{Q}\) for all \(i\in \mathcal{I}\).
Under these assumptions, the intersection \([-\bar \Delta,\bar \Delta]\cap \mathcal{Z}\) contains finitely many elements (see \hyperref[Lemma A.1]{Lemma A.1}). 
We label the elements \(\{z_0,\ldots,z_{N^\prime}\}\) in \([-\bar \Delta,\bar \Delta]\cap \mathcal{Z}\) 
in an increasing order with \(z_0=-\bar \Delta\) and \(z_{N^\prime}= \bar \Delta\).
Now each \((t_n,z_{n^\prime}),\ n= 0,\ldots,N,\ n^\prime=0,\ldots,N^\prime\) represents a time-space node of 
\([0,T]\times ([-\bar \Delta,\bar \Delta]\cap \mathcal{Z})\), where \(t_n = nh\). 

We introduce operators \(\tilde{\mathcal{L}}\) and \(\tilde{\mathcal{M}}\) associated with the time-space discretization of \eqref{Eq 43}, defined by
\begin{equation}{\label{Eq 45}}
  \begin{aligned}
    \tilde{\mathcal{L}} \varphi^\rho_d (t_n,z_{n^\prime})
    \coloneqq &\sum_{d^\prime \in \mathbb{S}^M}q_{d,d^\prime}\big(\varphi^\rho_{d^\prime}(t_n,z_{n^\prime})-\varphi^\rho_d(t_n,z_{n^\prime})\big)\\ 
    +  &\sum_{i=1}^{M} \sup\limits_{\alpha^{b,i}\in \mathcal{A}^b(d^i)} \Big\{\mathbb{I}_{\{z_{n^\prime}+\Delta^i \leq \bar{\Delta}\}} \big[H^{b,i}_{d^i}(\alpha^{b,i},z_{n^\prime},t_n,d^{\bar i}) + F^{b,i}_{d^i}(\alpha^{b,i},z_{n^\prime})\big]\Big\}\\
    + & \sum_{i=1}^{M} \sup\limits_{\alpha^{a,i}\in \mathcal{A}^a(d^i)} \Big\{\mathbb{I}_{\{z_{n^\prime}-\Delta^i \geq -\bar{\Delta}\}} \big[H^{a,i}_{d^i}(\alpha^{a,i},z_{n^\prime},t_n,d^{\bar i}) + F^{a,i}_{d^i}(\alpha^{a,i},z_{n^\prime})\big]\Big\}\\
  \end{aligned}
\end{equation}   
and 
\begin{equation}{\label{Eq 46}}
  \begin{aligned}
    \tilde{\mathcal{M}} \varphi^\rho_d (t_n,z_{n^\prime}) \coloneqq 
    \sup\limits_{\substack{\ell\in \mathcal{E}^M:\\z_{n^\prime} + \sum_{i=1}^{M}\ell^i \Delta^i \in [-\bar \Delta,\bar \Delta]}} \bigg\{-\sum_{i=1}^{M}\bigg(|\ell^i|\Big(\frac{d^i}{2}+\epsilon\Big)+\tilde{\epsilon}\mathbb{I}_{\{\ell^i \neq 0\}}\bigg)+\varphi^\rho_d\bigg(t_n,z_{n^\prime}+\sum_{i=1}^{M}\ell^i \Delta^i\bigg)\bigg\}.
  \end{aligned}
\end{equation}

Following the semi-implicit Algorithm \ref{Algo 1} in backward induction, 
we obtain the discretized solution \(\varphi^\rho_d(t_n,z_{n^\prime}),\ d \in \mathbb{S}^M\) at each node \((t_n,z_{n^\prime})\).
\begin{algorithm}
  \caption{Backward Semi-Implicit Algorithm}\label{Algo 1}
  \begin{algorithmic}[1]
  \State \textbf{Time step} \(t_N=T\): For all \(d \in \mathbb{S}^M\) and \(n^\prime \in \{0,\cdots,N^\prime\}\), set \(\varphi^\rho_d(t_N,z_{n^\prime})= - \gamma^\prime z_{n^\prime}^2\).
  \For{\(n = N-1, \dots, 0\)}
    \For{each \(d \in \mathbb{S}^M\)}
      \For{each \(n^\prime \in \{0,\cdots,N^\prime\}\)}
          \State Compute \(\tilde{\mathcal{L}}\varphi^\rho_d(t_{n+1},z_{n^\prime})\) from \eqref{Eq 45}, and store \((\alpha^{b,i,*},\ \alpha^{a,i,*})_{i=1}^M\) the argmax.
          \State Compute \(\tilde{\mathcal{M}}\varphi^\rho_d(t_{n+1},z_{n^\prime})\) from \eqref{Eq 46}, and store \((\ell^{1,*},\dots,\ell^{M,*})\) the argmax.
          \State Find \(\varphi^\rho_d(t_n,z_{n^\prime})\) by solving implicitly 
          \[
            \begin{aligned}
              & \varphi^\rho_d(t_n,z_{n^\prime}) - \rho h \big(\tilde{\mathcal{M}}\varphi^\rho_d(t_{n+1},z_{n^\prime})-\varphi^\rho_d(t_n,z_{n^\prime})\big)^+\\ 
          = & \; \varphi^\rho_d(t_{n+1},z_{n^\prime}) + h \big(\tilde{\mathcal{L}}\varphi^\rho_d(t_{n+1},z_{n^\prime})-\gamma z_{n^\prime}^2\big).
            \end{aligned}
          \]  
      \EndFor
    \EndFor
  \EndFor
  \end{algorithmic}
\end{algorithm}

\section{Numerical experiment}{\label{Sec 5}}
In this section, we provide numerical results for the rebate design problem formulated in Section \ref{Sec 3.4}. In particular, we will use our three-step scheme to determine a maximal level of rebate which should be offered to market makers when they quote conservatively, given a fixed rebate offered to market makers when they quote aggressively.

We consider the case where the exchange wishes to incentivize the market maker to consistently improve the best prices whenever possible, i.e. the case that \(\overline{\alpha^{b,i}}=\overline{\alpha^{a,i}}=1\) for all \(i\in \mathcal{I}\) in \eqref{Eq 34}. 

In step (a),  we fix rebate functions for aggressive quotes, which are symmetric on bid and ask sides, constant in portfolio delta,
and increasing in the spread:
\[
  F^{b,i}_{d^i}(0,\cdot) = F^{a,i}_{d^i}(0,\cdot) = 0,\quad d^i = \delta
\] and 
\[
  F^{b,i}_{d^i}(1,\cdot) = F^{a,i}_{d^i}(1,\cdot) = 0.1 * (d^i-\delta),\quad d^i > \delta.
\]

Specifically, no rebate is offered when the spread is one tick, otherwise the
rebate is proportional to the spread (if the market maker quotes aggressively). This implies that more rebate is paid to the market maker when she reduces a larger spread.  

\subsection{Data}{\label{Sec 5.1}}
We consider two (\(M=2\)) one-day-to-expiry (1DTE) SPXW call options \(C^1,C^2\) 
that expired on October 24, 2023, with strike prices \(K^1=\$4225\) and \(K^2=\$4230\) respectively.
The tick size for both options is \(\delta = \$0.01\). 
In the sequel, all the parameters are estimated from Level 1 data at the best quotes,
recorded during regular trading hours from 9:30 to 16:15 on October 23, 2023 (one day before expiry). We calibrate market parameters by considering the opening values of the market on this date, in particular we use the opening value of the VIX as a volatility estimate, and the Black--Scholes deltas calculated using the opening option prices and underlying SPX values for the constant-delta approximation (rounded to two decimal places).

Before estimation, we first analyze the time series of high-frequency intraday bid-ask spreads. 
Figure \ref{Fig 1} presents the intraday spread dynamics resampled at second frequency, excluding spread values exceeding ten ticks.
Notably, neither figure reveals any evident intraday seasonality effect,
supporting the homogeneity assumption on the spread Markov chain.\footnote{
The period of heightened spreads in the \(\$4225\) strike option does suggest that there is a case for a model with longer memory, such as a regime switching model. While we omit this extension here, as our model is concerned with short timeframes (at the level of seconds), this could easily be incorporated, and would require a more involved calibration of the spread dynamics.}
We note, however, that while this dataset does not exhibit any prominent intraday spread patterns, such effects may be present on other trading days or for longer-dated options.

Moreover, the histograms in Figure \ref{Fig 2} indicate that the first six spread levels dominate the empirical distribution and suggest disregarding larger spread values (as they occur infrequently, and hence cannot be reliably calibrated). This results in a state space \(\mathbb{S}=\{\delta,2 \delta,\dots, 6 \delta\}\) (\(m=6\)).

The estimation procedure is detailed in \hyperref[Appendix B]{Appendix B}, which also includes an inspection of the correlation in spreads (see Figure \ref{Fig 6}).
In the subsequent numerical experiment, we use the estimated parameters of the spread Markov chain (Table \ref{Table 2})
and the execution processes (Table \ref{Table 3}), together with other user-specified parameters (Table \ref{Table 1}).

\begin{figure}[h]
  \centering
  \begin{subfigure}[b]{0.5\textwidth}
    \centering
    \includegraphics[width=\linewidth]{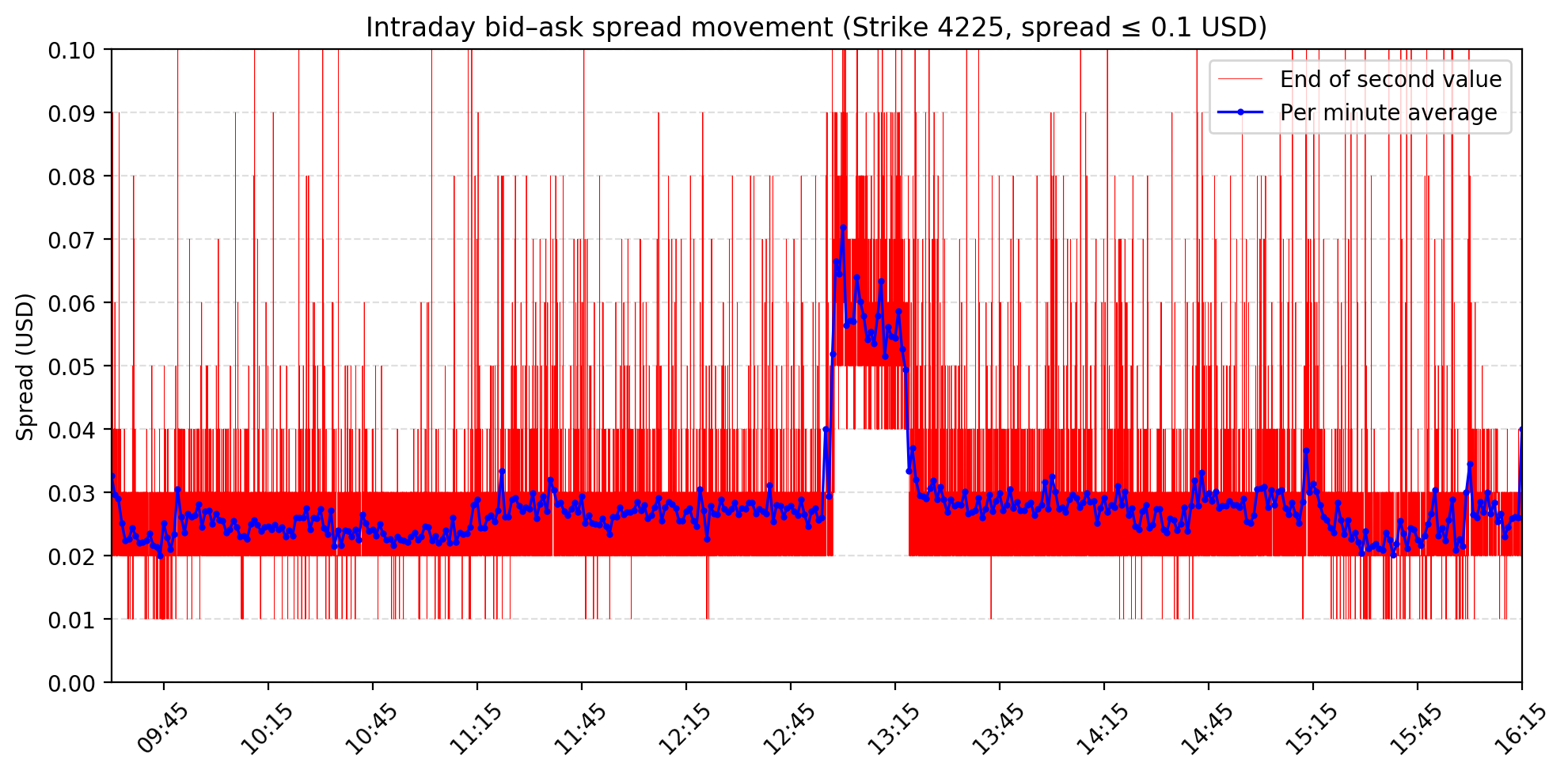}
  \end{subfigure}
  
  \vspace{0.3cm}

  \begin{subfigure}[b]{0.5\textwidth}
    \centering
    \includegraphics[width=\linewidth]{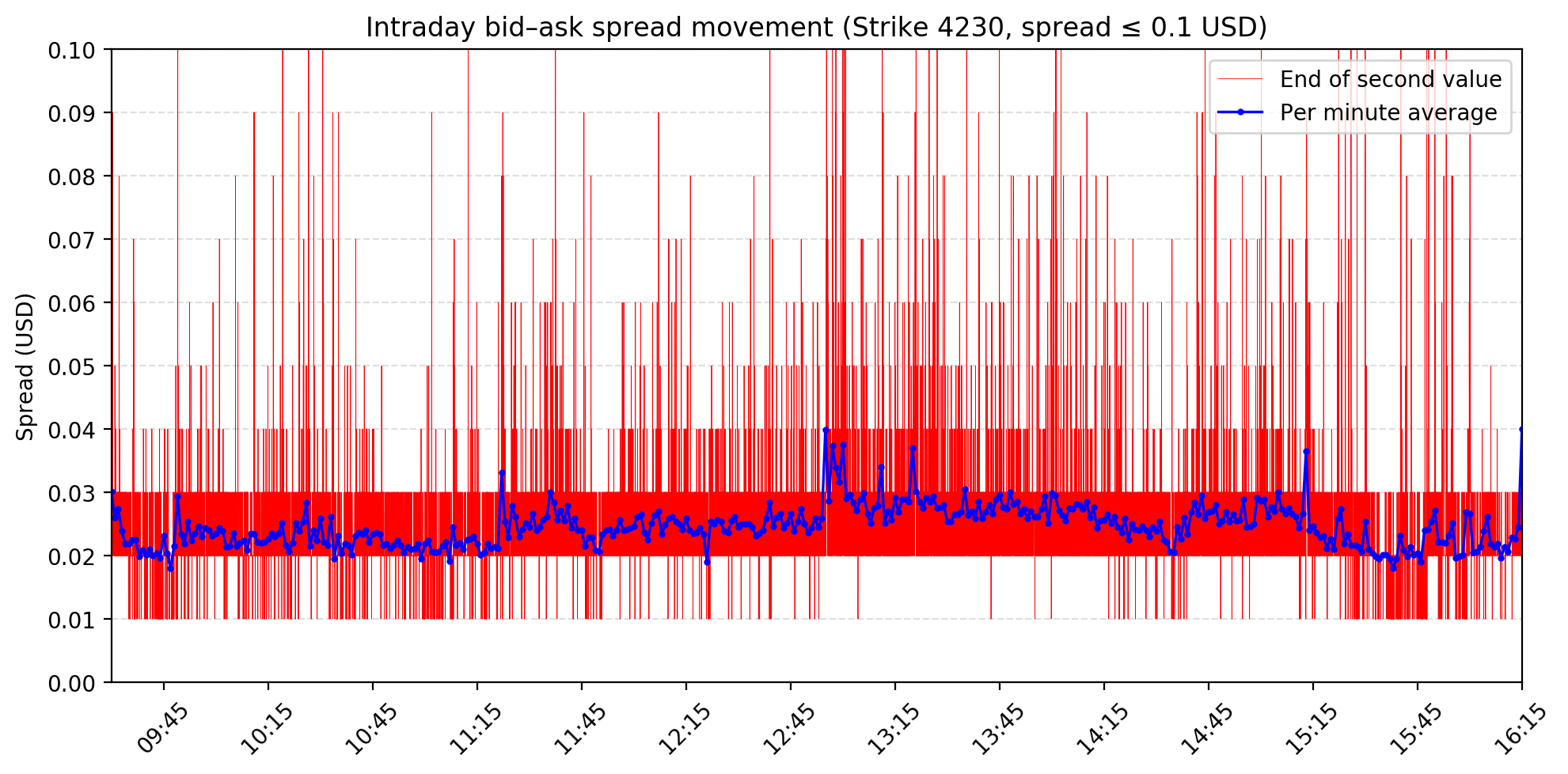}
  \end{subfigure}

  \caption{Dynamics of intraday bid-ask spreads in seconds.}
  \label{Fig 1}
\end{figure}

\begin{figure}[H]
  \centering
  \begin{subfigure}[b]{0.5\textwidth}
      \includegraphics[width=\linewidth]{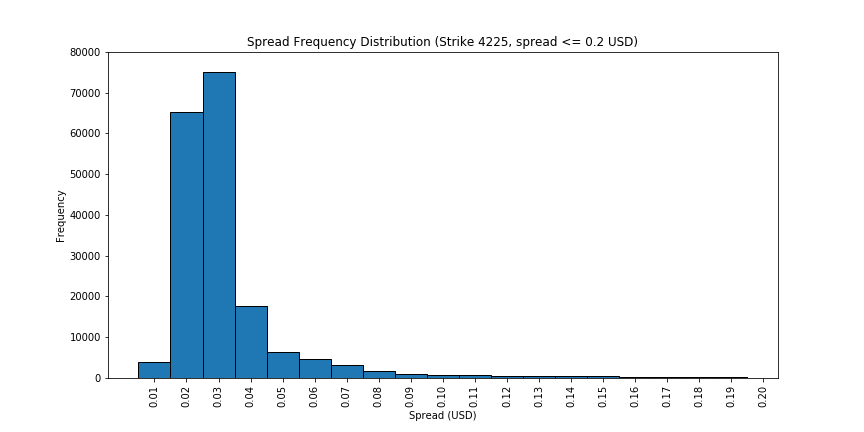}
  \end{subfigure}
  \hspace{0.1 cm}
  \begin{subfigure}[b]{0.5\textwidth}
      \includegraphics[width=\linewidth]{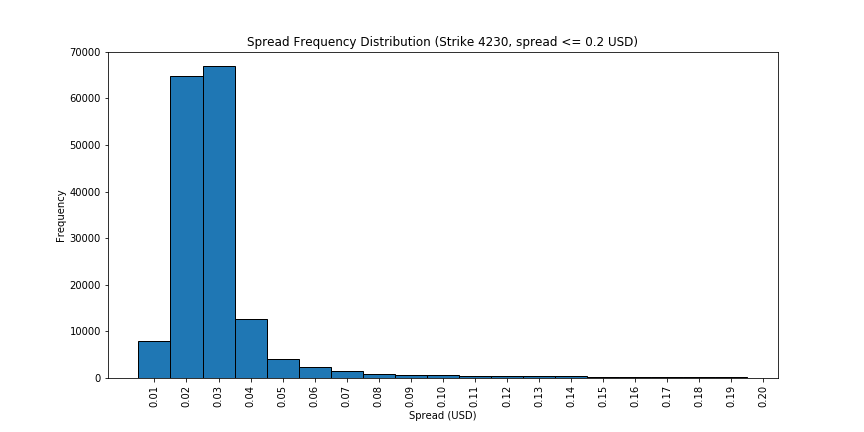}
  \end{subfigure}
  
  \caption{Empirical distribution of intraday bid-ask spreads.}
  \label{Fig 2}
\end{figure}

\begin{table}[h]
  \centering
  \setlength{\tabcolsep}{3pt}
  \begin{tabular}{p{0.2\textwidth} >{\raggedright}p{0.6\textwidth} >{\centering\arraybackslash}p{0.2\textwidth}}
      \toprule
      \textbf{Parameter} & \textbf{Description} & \textbf{Value} \\
      \midrule
      \multicolumn{3}{c}{\textbf{(a) Market Parameters}} \\
      \(T\) & Market making horizon & 10 s \\
      \(\sigma\) & Volatility & 0.22 \\
      \(\delta\) & Tick size & 0.01 dollars \\
      \(m\) & Number of spread values & 6 \\ 
      \(\bar{e}\) & Maximal volume of market order & 3 \\ 
      \(\epsilon\) & Transaction fee per unit of call & \(10^{-3}\) dollars \\
      \(\tilde \epsilon\) & Fixed brokerage fee & \(10^{-5}\) dollars \\ 
      \midrule
      \multicolumn{3}{c}{\textbf{(b) Optimization Parameters}} \\
      \(\gamma\) & Risk aversion towards running portfolio delta & 5 \\
      \(\gamma^\prime\) & Risk aversion towards terminal portfolio delta & 5 \\
      \(\Delta^1\)  & Initial delta of \(C^1\) & 0.50 \\
      \(\Delta^2\)  & Initial delta of \(C^2\) & 0.45 \\
      \midrule
      \multicolumn{3}{c}{\textbf{(c) Discretization Parameters}} \\
      \(\rho\) & Penalty scheme parameter & 1 \\
      \(N\) & Number of timesteps & 2000 \\
      \(\bar\Delta\) & Delta risk limit &  15\\
      \bottomrule
  \end{tabular}
  \caption{Parameters}
  \label{Table 1}
\end{table}

\subsection{Shape of value function}
In this section, we perform step (b) of our three-step scheme.
Specifically, given the fixed rebate functions \((F^{b,i},F^{a,i})_{i=1}^M\) proposed earlier,
we use Algorithm \ref{Algo 1} to 
numerically solve the system of QVIs \eqref{Eq 35} and \eqref{Eq 36} after penalizing the obstacle terms. The numerical result is the discretized value function \(\varphi = (\varphi_{d})_{d\in \mathbb{S}^M}\) 
defined over the domain \([0,T] \times ([-\bar \Delta,\bar \Delta]\cap \mathcal{Z})\). 

Figure \ref{Fig 3} displays the value function \(\varphi\) on a grid for various spread values. 
Each subfigure displays the value function \(\varphi_{d}\) for \(d = (d^1,d^2)=(k \delta,k \delta)\) with \(k=1,\ldots,m\) as a multiple of the tick size.
We observe that, in each subfigure, the value function \(\varphi\), which starts as a quadratic function at \(t=T\),
gradually bends downward as \(t\) approaches \(0\), forming an approximately quadratic shape.
This observation aligns with the quadratic penalty on portfolio delta in the performance criterion. 

\begin{figure}[H]
  \centering
  \begin{subfigure}{0.35\textwidth}
      \includegraphics[width=\linewidth]{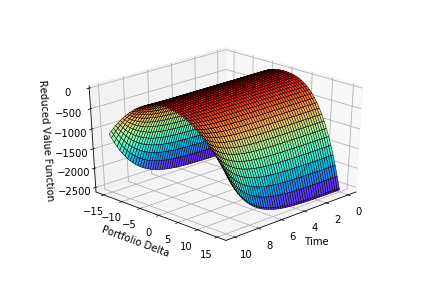}
      \caption{Spread = 1 tick}
      \label{fig:value_function_spread_1}
  \end{subfigure}
  \hspace{0.1\textwidth}
  \begin{subfigure}{0.35\textwidth}
      \includegraphics[width=\linewidth]{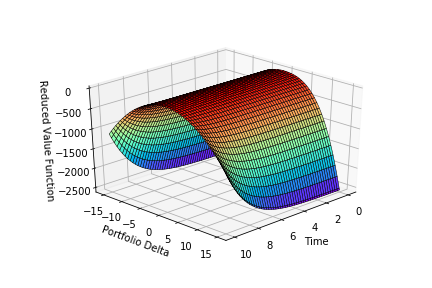}
      \caption{Spread = 2 ticks}
      \label{fig:value_function_spread_2}
  \end{subfigure}
  
  \begin{subfigure}{0.35\textwidth}
      \includegraphics[width=\linewidth]{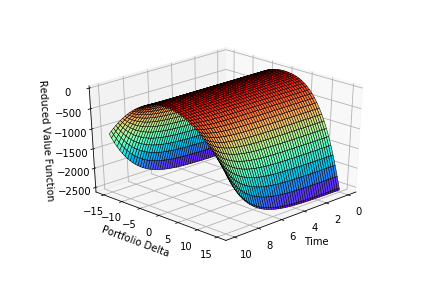}
      \caption{Spread = 3 ticks}
      \label{fig:value_function_spread_3}
  \end{subfigure}
  \hspace{0.1\textwidth}
  \begin{subfigure}{0.35\textwidth}
      \includegraphics[width=\linewidth]{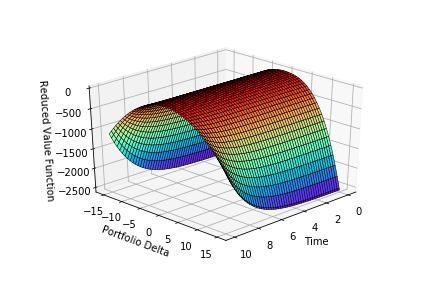}
      \caption{Spread = 4 ticks}
      \label{fig:value_function_spread_4}
  \end{subfigure}
  
  \begin{subfigure}{0.35\textwidth}
      \includegraphics[width=\linewidth]{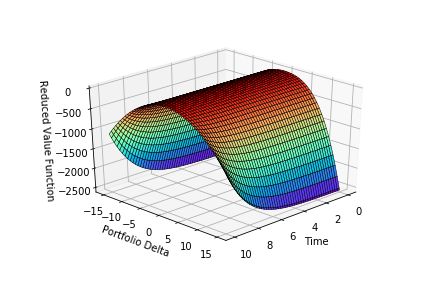}
      \caption{Spread = 5 ticks}
      \label{fig:value_function_spread_5}
  \end{subfigure}
  \hspace{0.1\textwidth}
  \begin{subfigure}{0.35\textwidth}
      \includegraphics[width=\linewidth]{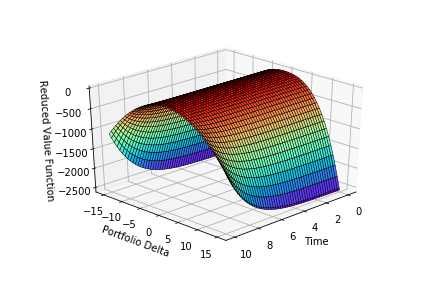}
      \caption{Spread = 6 ticks}
      \label{fig:value_function_spread_6}
  \end{subfigure}
  
  \caption{Plots of the reduced value function \(\varphi_d\), \(d = (d^1,d^2)=(k \delta, k \delta)\) with \(k=1,\cdots,6\).}
  \label{Fig 3}
\end{figure}


Moreover, the plots in Figure \ref{Fig 3} suggest a moderate dependence of the value function \(\varphi_d\) on the spread \(d\) (see \hyperref[Appendix C]{Appendix C} for numerical justification).
Indeed, we will see in \hyperref[Appendix B]{Appendix B} that the spread quickly mean reverts at the scale of seconds, so it has negligible impact on the value function.

\subsection{Maximal fee rebate}
In this section, we perform step (c) of our three-step scheme. 
Specifically, given the fixed rebate functions proposed earlier and the discretized value function \(\varphi\), 
we calculate the maximal fee rebates \(\text{LUB}^{b,i}_{d^i}\) and \(\text{LUB}^{a,i}_{d^i}\) \big(see \eqref{Eq 41} and \eqref{Eq 42}\big) for conservative quotes,
so that we may encourage market makers to quote aggressively by limiting the rebates 
they are given when quoting conservatively.

Taking option \(C^1\) as an example, Figure \ref{Fig 4} displays the least upper bounds \(\text{LUB}^{b,1}_{d^1}\) and \(\text{LUB}^{a,1}_{d^1}\) 
of the rebate functions \(F^{b,1}_{d^1}(0,\cdot)\) and \(F^{a,1}_{d^1}(0,\cdot)\) as a function of portfolio delta, for \(d^1=2\delta,\ldots,6 \delta\). 
For a given spread value \(d^1\), the least upper bound \(\text{LUB}^{b,1}_{d^1}\) 
(resp. \(\text{LUB}^{a,1}_{d^1}\)) represents the maximal fee rebate 
that the exchange should offer to the market maker if she does not improve the best bid (resp.\ ask) price when spread is \(d^1\). 
In each subfigure, the bid-side least upper bound \(\text{LUB}^{b,1}_{d^1}\) is marked in blue color, and the ask-side least upper bound \(\text{LUB}^{a,1}_{d^1}\) is marked in orange color. 

In each subfigure, we observe that the least upper bound for the bid (resp.\ ask) side 
is strictly positive when the portfolio delta is strictly negative (resp. positive). 
This suggests, in the case where the market maker holds negative (resp. positive) portfolio delta,
the exchange does not need to compensate the market maker on the bid (resp.\ ask) side,
since even without rebate incentives,
the market maker would place her limit buy (resp. sell) orders at the touch in order to reduce her delta exposure.

On the other hand, for \emph{small} spread values \(d^1=2 \delta\) and \(3 \delta\), the least upper bound for the bid (resp.\ ask) side 
is strictly negative when the portfolio delta is strictly positive (resp.\ negative) with large absolute value. 
This means the exchange should effectively ``punish'' the market maker if she fails to improve the best prices 
so that she is enforced to quote aggressively in order to avoid punishment. 

However, for \emph{large} spread values \(d^1 = 4 \delta,\ 5 \delta\) and \(6 \delta\), the least upper bound for the bid (resp.\ ask) side 
remains positive, and is indeed equal to the fee rebates for aggressive quotes
when the portfolio delta is strictly positive (resp.\ negative).
This initially counterintuitive fact can be rationalised by noticing that when the spread is at least four ticks,
the execution intensity (see Table \ref{Table 3}) is (almost) zero on both sides regardless of aggressiveness. 
In such cases, no liquidity taker arrives in the market and thus no transactions occur, 
so it makes no practical difference whether the market maker submits aggressive or conservative quotes — neither will be executed.
Therefore, an arbitrary fee rebate is sufficient to incentivize the market maker to quote at the touch, 
but cruciallly should not exceed the fee rebate for aggressive quotes (otherwise the market maker will send conservative quotes instead).
This explains why the maximal fee rebate when the market maker quotes conservatively is found to be equal to the fixed fee rebate when she quotes aggressively.

\begin{remark}{\label{MFR}}
  The rebate policy proposed above is somewhat unrealistic as it involves negative fee rebate in some circumstances.
  In practice, however, market makers are not penalized for simply adding liquidity at the touch. 
  To circumvent this issue, a natural idea is to \emph{lift} the entire rebate schedule, that is, to uniformly increase the baseline rebate for aggressive quotes 
  so that the maximal fee rebate for conservative quotes becomes positive.
Lifting the rebate schedule does not alter the core mechanism of the policy. What truly matters is the gap between the rebates assigned to conservative and aggressive quotes.
  This inspires us to consider a \emph{relative rebate policy} that focuses on the extra rewards needed to improve the best price, rather than on absolute rebate values.
  Precisely, we need to specify the \emph{marginal fee rebates}
  \[
  \begin{dcases}
    \text{MFR}^{b,i}_{d^i}(\Delta^\pi) \coloneqq F^{b,i}_{d^i}(1,\Delta^\pi) - \text{LUB}^{b,i}_{d^i}(\Delta^\pi),\\
    \text{MFR}^{a,i}_{d^i}(\Delta^\pi) \coloneqq F^{a,i}_{d^i}(1,\Delta^\pi) - \text{LUB}^{a,i}_{d^i}(\Delta^\pi),
  \end{dcases}
  \] which represent the minimal additional rewards the exchange should offer to incentivize a market maker to improve the best price.
\end{remark}

\begin{figure}[h]
  \centering
  \begin{subfigure}{\textwidth}
    \centering
    \includegraphics[width=\textwidth]{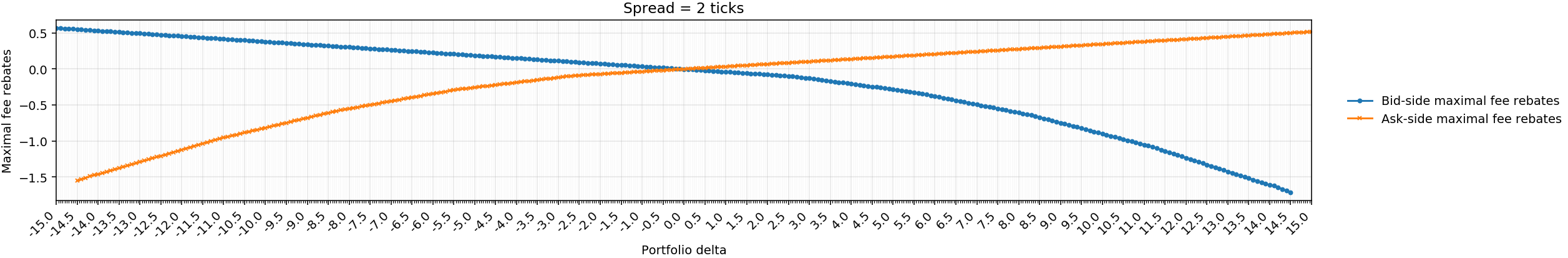}
  \end{subfigure}

  \vspace{0.6em}
  \begin{subfigure}{\textwidth}
    \centering
    \includegraphics[width=\textwidth]{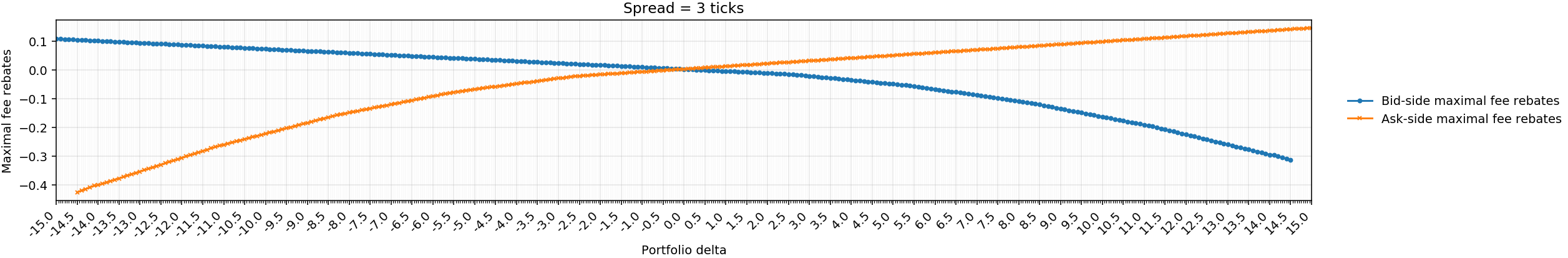}
  \end{subfigure}

  \vspace{0.6em}
  \begin{subfigure}{\textwidth}
    \centering
    \includegraphics[width=\textwidth]{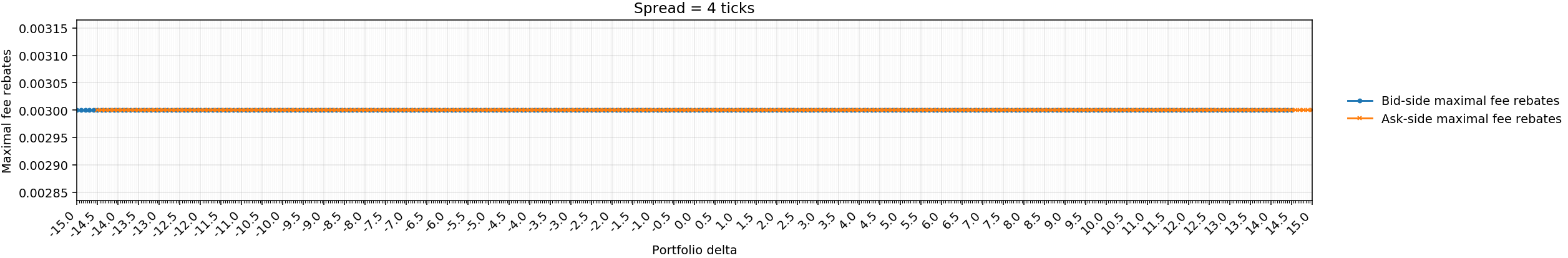}
  \end{subfigure}

  \vspace{0.6em}
  \begin{subfigure}{\textwidth}
    \centering
    \includegraphics[width=\textwidth]{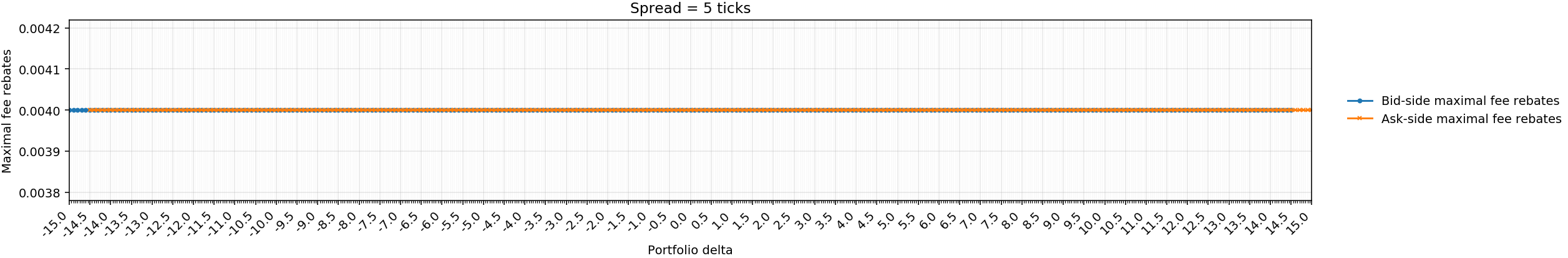}
  \end{subfigure}

  \vspace{0.6em}
  \begin{subfigure}{\textwidth}
    \centering
    \includegraphics[width=\textwidth]{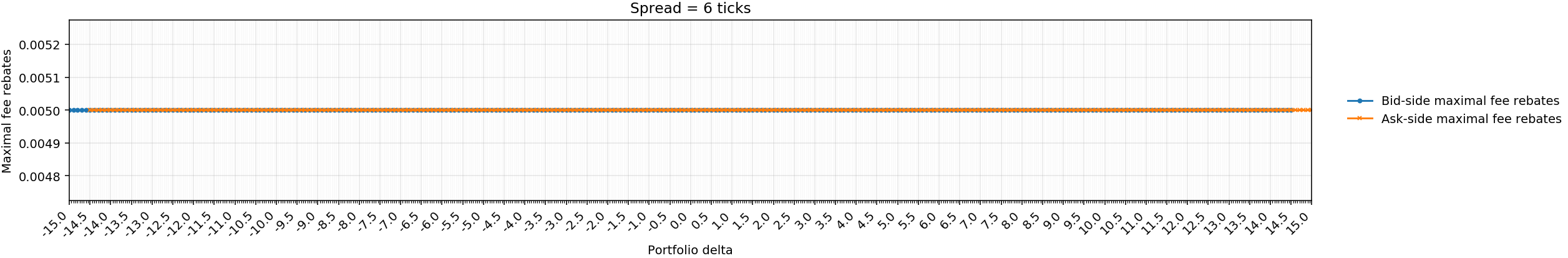}
  \end{subfigure}
  
  \caption{Plots of maximal fee rebates \(\text{LUB}^{b,1}_{d^1}\) and \(\text{LUB}^{a,1}_{d^1}\) as a function of portfolio delta
  when market maker quotes conservatively, \(d^1= 2\delta,\cdots,6\delta\).}
  \label{Fig 4}
\end{figure}

\section*{Acknowledgements}
This research was supported by CBOE, through the EPSRC Centre for Doctoral Training in Mathematics of Random Systems: Analysis, Modelling and Simulation (EP/S023925/1).

\section*{Contributions}
SC, ZG and CR developed the modelling framework, numerical schemes, and wrote the paper. LD and ZG processed the raw CBOE PITCH data underlying the numerical section, and implemented the computational algorithm. The authors thank Florian Huchede (CBOE), Leandro Sanchez Betancourt, and Alvaro Cartea for useful comments.
\clearpage

\newpage

\bibliography{ref}
\bibliographystyle{siam} 

\begin{appendix}

\section{More on space discretization} 
In this section, we discuss the finiteness of \([-\bar \Delta,\bar \Delta] \cap \mathcal{Z}\), where 
\[
  \mathcal{Z} \coloneqq \sum_{i=1}^{M} \Delta^i \mathbb{Z} \coloneqq \big\{\sum_{i=1}^{M} \Delta^i y^i:\ y^i \in \mathbb{Z},\ i=1,\dots,M\big\}  
\] with \(\Delta^i \in [0,1]\) for all \(i=1,\dots,M\), and \(\bar \Delta\) is assumed to belong to \(\mathcal{Z}\).  

\begin{lemma}{\label{Lemma A.1}}
  If \(\Delta^i \in \mathbb{Q}\) for all \(i=1,\dots,M\), then \([-\bar \Delta,\bar \Delta] \cap \mathcal{Z}\) is finite.  
\end{lemma}
\begin{proof}
  By assumption, for each \(i\) we can write \(\Delta^i = \frac{p_i}{q_i}\) for some positive integers \(p_i,q_i\) with \(p_i \leq q_i\).
  Any element \(z\) of \(\mathcal{Z}\) is of the form
  \begin{equation}{\label{Eq 47}}
    z = \sum_{i=1}^{M} y^i \Delta^i = \sum_{i=1}^{M} y^i \frac{p_i}{q_i} = \frac{\sum\limits_{i=1}^{M} y^i p_i \prod\limits_{j \neq i}q_j}{\prod\limits_{i=1}^M q_i},\ 
    y^i \in \mathbb{Z},\ i= 1,\dots,M.
  \end{equation}
  It follows that \(\mathcal{Z} \subset \frac{1}{\prod_{i=1}^M q_i} \mathbb{Z}\) since the numerator \(\sum\limits_{i=1}^{M} y^i p_i \prod\limits_{j \neq i}q_j\) in \eqref{Eq 47}
  is always an integer.\\
  Now if \(z \in \mathcal{Z} \cap [-\bar \Delta,\bar \Delta]\), then there exists \(k\in \mathbb{Z}\) such that  
  \[
    z = \frac{k}{\prod_{i=1}^M q_i} \in [-\bar \Delta,\bar \Delta],
  \]
  or equivalently,
  \[
    - \bar \Delta \prod_{i=1}^M q_i \leq k \leq \bar \Delta \prod_{i=1}^M q_i. 
  \]
  From this we deduce that the set \(\mathcal{Z} \cap [-\bar \Delta,\bar \Delta]\) can at most have \(2\left \lfloor{\bar \Delta \prod_{i=1}^M q_i}\right \rfloor + 1\) elements, 
  in particular, is finite. 
\end{proof}

\hyperref[Lemma A.1]{Lemma A.1} provides a sufficient condition for \([-\bar \Delta,\bar \Delta] \cap \mathcal{Z}\) to be finite. 
It turns out that, if the assumption of rational \(\Delta^i\) in \hyperref[Lemma A.1]{Lemma A.1} is violated, 
then condition \eqref{Eq 48} in the following \hyperref[Cor A.3]{Corollary A.3} may hold true, 
in which case \([-\bar \Delta,\bar \Delta] \cap \mathcal{Z}\) is countably infinite.

\begin{lemma}{\label{Lemma A.2}}
  If \(x\in \mathbb{R} \setminus \mathbb{Q}\), then for any \(\epsilon >0\) there exist \(n,m\in \mathbb{Z}\) such that 
  \[
    0< |nx-m| < \epsilon.
  \]   
\end{lemma}
\begin{proof}
  Let \(x\in \mathbb{R} \setminus \mathbb{Q}\). For each \(n\in \mathbb{Z}\) let \(f(n)\) be the unique integer such that 
  \[
    n x + f(n) \in [0,1].
  \]    
  The sequence \(\{nx+f(n)\}_{n\in \mathbb{Z}}\) has distinct elements, since if there exist \(n\neq m \in \mathbb{Z}\) such that 
  \[
    nx + f(n) = mx + f(m),
  \] then 
  \[
    x = \frac{f(m)-f(n)}{n-m}\in \mathbb{Q}
  \] forms a contradiction.\\
  Since \(\{nx+f(n)\}_{n\in \mathbb{Z}}\) is a bounded sequence in \([0,1]\) with distinct elements, Bolzano-Weierstrass theorem implies that 
  it has a convergent subsequence, which further implies that for any \(\epsilon >0\) there exist \(\tilde n,\tilde m\in \mathbb{Z}\) such that 
  \[
    0 < \big| \big(\tilde n x+f(\tilde n)\big) - \big(\tilde m x+f(\tilde m)\big)\big| <\epsilon,
  \] or 
  \[
    0 < |n x - m| < \epsilon
  \] with \(n \coloneqq \tilde n-\tilde m \in \mathbb{Z}\) and \(m \coloneqq f(\tilde m) - f(\tilde n)\in \mathbb{Z}\).  
\end{proof}

\begin{corollary}{\label{Cor A.3}}
  If \(\Delta^1,\dots,\Delta^M\) are such that 
  \begin{equation}{\label{Eq 48}}
    \frac{\sum_{i\in \tilde{\mathcal{I}}} \Delta^i}{\sum_{i \in \mathcal{I}\setminus \tilde{\mathcal{I}}} \Delta^i} \in \mathbb{R}\setminus \mathbb{Q}
  \end{equation}
  for some subset \(\tilde{\mathcal{I}}\) of \(\mathcal{I}=\{1,\dots,M\}\),  
  then \([-\bar \Delta,\bar \Delta] \cap \mathcal{Z}\) is countably infinite.
\end{corollary}
\begin{proof}
  By \hyperref[Lemma A.2]{Lemma A.2}, if \(\Delta^1,\dots,\Delta^M\) satisfy \eqref{Eq 48} for some subset \(\tilde{\mathcal{I}}\) of \(\mathcal{I}\), 
  then for any \(\epsilon>0\), there exist \(n,m\in \mathbb{Z}\) such that 
  \[
    0 < \bigg| n \frac{\sum_{i\in \tilde{\mathcal{I}}} \Delta^i}{\sum_{i \in \mathcal{I}\setminus \tilde{\mathcal{I}}} \Delta^i}  + m \bigg| < \frac{\epsilon}{\sum_{i \in \mathcal{I}\setminus \tilde{\mathcal{I}}} \Delta^i},
  \] or
  \[
    0 < \bigg| n \sum_{i\in \tilde{\mathcal{I}}} \Delta^i+ m \sum_{i \in \mathcal{I}\setminus \tilde{\mathcal{I}}} \Delta^i \bigg| <\epsilon. 
  \] This shows that there is a convergent sequence in \([-\bar \Delta,\bar \Delta]\cap \mathcal{Z}\) with limit \(0\), which holds only when 
  \([-\bar \Delta,\bar \Delta] \cap \mathcal{Z}\) is countably infinite.
\end{proof}

\section{Parameter estimation}{\label{Appendix B}}
In this section, we calibrate the model parameters associated to the spread Markov chain \(D = (D^i)_{i=1}^M\)
and the execution processes \((N^{b,i},N^{a,i})_{i=1}^M\) using the dataset described in Section \ref{Sec 5.1}.

\subsection{Estimation for spread Markov chains}
Given the continuously observed spread data \(D\), its jump times \(T_n,\ n \geq 0\) are given by  
\[
  T_0 \equiv 0, \quad T_{n+1} \coloneqq \inf\{t>T_n:D_t \neq D_{t^-}\},\quad n \geq 0,
\] then we can recover its jump chain \(\hat{D}\) by 
\[
  \hat{D}_n  =   D_{T_n},\quad n \geq 0.
\]  

Our goal is to estimate the jump rates \(q_{jk},\ j,k\in \mathbb{S}^M\) of the spread Markov chain \(D\),
from a path realization with the high frequency spread data in tick time, over a fixed market making period \([0,T]\). 
To do so, we fix \(j,k\in \mathbb{S}^M\), and denote by \(N_{jk}(t)\) the number of visits from state \(j\) to state \(k\) by time \(t\),
by \(N_{j\rightarrow \cdot} (t)\) the number of visits from state \(j\) to other states by time \(t\),
and by \(\mathcal{O}_j(t)\) the occupation time in state \(j\) by time \(t\):
\[
  \begin{dcases}
    N_{jk}(t)\coloneqq \sum_{n=1}^{\infty} \mathbb{I}_{\{T_{n} \leq t\}} \mathbb{I}_{\{\hat{D}_{n-1}=j,\hat{D}_{n}=k\}},\quad t\in [0,T]\\
    N_{j\rightarrow \cdot}(t) \coloneqq \sum_{n=1}^{\infty} \mathbb{I}_{\{T_n \leq t\}} \mathbb{I}_{\{\hat{D}_{n-1}=j\}},\quad t\in [0,T]\\
    \mathcal{O}_j(t) \coloneqq \int_0^t \mathbb{I}_{\{D_s=j\}} ds,\quad t\in [0,T].
  \end{dcases}
\]

We claim that the estimator 
\begin{equation}{\label{Eq 49}}
  \hat{q}_{jk}(t) \coloneqq \frac{N_{jk}(t)}{\mathcal{O}_j(t)} = \frac{N_{jk}(t)}{N_{j\rightarrow\cdot}(t)} \frac{N_{j\rightarrow\cdot}(t)}{\mathcal{O}_j(t)}
\end{equation}
converges \(\mathbb{P}\)-a.s. to the jump rate \(q_{jk}\) as \(t\to \infty\).  

By construction, the jump chain \(\hat{D}\) is homogeneous, irreducible, and takes values in the finite state space \(\mathbb{S}^M\), so it has a unique stationary distribution \(\pi\), 
and hence satisfies the ergodicity property (see \cite[Theorem 1.22]{durrett2016essentials}): 
\[
  \frac{1}{K} \sum_{n=0}^{K-1} f(\hat{D}_{n}) \to \sum\limits_{j\in \mathbb{S}^M} f(j) \pi(j)
\] \(\mathbb{P}\)-a.s. as \(K\to\infty\) for any \(f\in L^1(\pi)\), from which we deduce that with probability one
\begin{equation}{\label{Eq 50}}
  \lim\limits_{t \to \infty} \frac{N_{jk}(t)}{N_{j\rightarrow\cdot}(t)} = \lim\limits_{K \to \infty} \frac{\sum_{n=0}^{K-1} \mathbb{I}_{\{\hat{D}_{n}=j,\hat{D}_{n+1}=k\}}}{\sum_{n=0}^{K-1} \mathbb{I}_{\{\hat{D}_n=j\}}} = p_{jk}.
\end{equation}

To prove the convergence of the second term, in a similar spirit to \cite[Theorem 3.8.1]{norris1998markov}, 
we may assume that the Markov chain \(D\) starts at the state \(j\in \mathbb{S}^M\). 
Under this assumption, the \(n\)-th passage time \(T^j_n\) to state \(j\) and the length \(M^j_n\) of the \(n\)-th visit to state \(j\) are given by 
\[
  \begin{dcases}
    T^j_0 \coloneqq 0,\quad T^j_n \coloneqq \inf\{t>T^j_{n-1}+M^j_n:D_t=j\},\quad n\in \mathbb{N}\\
    M^j_n \coloneqq \inf\{t>T^j_{n-1}:D_t \neq j \} - T^j_{n-1},\quad n\in \mathbb{N}.
  \end{dcases}
\]  
By the strong Markov property at the stopping times \(T^j_n\), we see that \(M^j_1,M^j_2,\cdots\) are independent exponentials with rate \(q_j\),
which by the strong law of large numbers implies that with probability one
\begin{equation}{\label{Eq 51}}
  \lim\limits_{t \to \infty} \frac{N_{j\rightarrow\cdot}(t)}{\mathcal{O}_j(t)} = \lim\limits_{n \to \infty} \frac{N_{j\rightarrow\cdot}(T^j_n)}{\mathcal{O}_j(T^j_n)} 
  = \lim\limits_{n \to \infty} \frac{n}{M^j_1+\dots+M^j_n} = q_j.
\end{equation}

The claim now follows from \eqref{Eq 50} and \eqref{Eq 51}:
\[
  \lim\limits_{t \to \infty} \hat{q}_{jk}(t) = p_{jk} q_j = q_{jk}.
\]

Now we perform the estimation \eqref{Eq 49} on our data. 
Table \ref{Table 2} displays the estimated generator of the spread Markov chain \(D\), 
in which the \(j\)-th row represents the spread value (in ticks) before jump 
while the \(k\)-th column represents the spread value (in ticks) after jump,
with the entry \(\hat{q}_{jk}\) indicating the estimated jump intensity from spread value \(j\) to \(k\).
\begin{sidewaystable*}[p]
  \centering
  \caption{Estimated generator of the spread Markov chain, expressed in ticks and \(\text{second}^{-1}\).}
  \label{Table 2}
  \setlength{\tabcolsep}{1pt}
  \renewcommand{\arraystretch}{0.9}
  \makebox[\textheight][c]{%
    \resizebox{\textheight}{!}{\begin{tabular}{lrrrrrrrrrrrrrrrrrrrrrrrrrrrrrrrrrrrr}
\toprule
{} &  (1, 1) &  (1, 2) &  (1, 3) &  (1, 4) &  (1, 5) &  (1, 6) &  (2, 1) &  (2, 2) &  (2, 3) &  (2, 4) &  (2, 5) &  (2, 6) &  (3, 1) &  (3, 2) &  (3, 3) &  (3, 4) &  (3, 5) &  (3, 6) &  (4, 1) &  (4, 2) &  (4, 3) &  (4, 4) &  (4, 5) &  (4, 6) &  (5, 1) &  (5, 2) &  (5, 3) &  (5, 4) &  (5, 5) &  (5, 6) &  (6, 1) &  (6, 2) &  (6, 3) &  (6, 4) &  (6, 5) &  (6, 6) \\
\midrule
(1, 1) &  -24.11 &    8.18 &    0.03 &    0.00 &    0.00 &    0.00 &   12.96 &    2.63 &    0.00 &    0.00 &    0.00 &    0.00 &    0.20 &    0.03 &    0.00 &    0.00 &    0.00 &    0.00 &    0.03 &    0.00 &    0.05 &    0.03 &    0.00 &    0.00 &    0.00 &    0.00 &    0.00 &    0.00 &    0.00 &    0.00 &    0.00 &    0.00 &    0.00 &    0.00 &    0.00 &    0.00 \\
(1, 2) &    2.54 &  -26.37 &    0.87 &    0.01 &    0.00 &    0.01 &    0.00 &   22.38 &    0.30 &    0.00 &    0.00 &    0.00 &    0.00 &    0.24 &    0.00 &    0.00 &    0.00 &    0.00 &    0.00 &    0.02 &    0.01 &    0.00 &    0.00 &    0.00 &    0.00 &    0.00 &    0.00 &    0.00 &    0.00 &    0.00 &    0.00 &    0.00 &    0.00 &    0.00 &    0.00 &    0.00 \\
(1, 3) &    0.36 &   20.68 &  -61.13 &    1.27 &    0.54 &    0.18 &    0.00 &    0.00 &   34.47 &    0.00 &    0.00 &    0.00 &    0.00 &    0.00 &    2.72 &    0.00 &    0.00 &    0.00 &    0.00 &    0.00 &    0.54 &    0.00 &    0.00 &    0.00 &    0.00 &    0.00 &    0.00 &    0.00 &    0.00 &    0.18 &    0.00 &    0.00 &    0.18 &    0.00 &    0.00 &    0.00 \\
(1, 4) &    0.00 &    8.21 &   20.53 &  -69.79 &    4.11 &    0.00 &    0.00 &    0.00 &    0.00 &   16.42 &    0.00 &    0.00 &    0.00 &    0.00 &    0.00 &    4.11 &    0.00 &    0.00 &    0.00 &    0.00 &    0.00 &   12.32 &    0.00 &    0.00 &    0.00 &    0.00 &    0.00 &    4.11 &    0.00 &    0.00 &    0.00 &    0.00 &    0.00 &    0.00 &    0.00 &    0.00 \\
(1, 5) &    0.00 &    0.00 &   39.53 &   26.35 & -118.58 &   13.18 &    0.00 &    0.00 &    0.00 &    0.00 &   13.18 &    0.00 &    0.00 &    0.00 &    0.00 &    0.00 &   13.18 &    0.00 &    0.00 &    0.00 &    0.00 &    0.00 &   13.18 &    0.00 &    0.00 &    0.00 &    0.00 &    0.00 &    0.00 &    0.00 &    0.00 &    0.00 &    0.00 &    0.00 &    0.00 &    0.00 \\
(1, 6) &    0.00 &    0.00 &    0.00 &    2.50 &    1.25 &   -5.01 &    0.00 &    0.00 &    0.00 &    0.00 &    0.00 &    0.00 &    0.00 &    0.00 &    0.00 &    0.00 &    0.00 &    1.25 &    0.00 &    0.00 &    0.00 &    0.00 &    0.00 &    0.00 &    0.00 &    0.00 &    0.00 &    0.00 &    0.00 &    0.00 &    0.00 &    0.00 &    0.00 &    0.00 &    0.00 &    0.00 \\
(2, 1) &    1.51 &    0.00 &    0.00 &    0.00 &    0.00 &    0.00 &  -18.81 &   14.66 &    0.09 &    0.01 &    0.00 &    0.00 &    2.01 &    0.49 &    0.01 &    0.00 &    0.00 &    0.00 &    0.01 &    0.01 &    0.01 &    0.00 &    0.00 &    0.00 &    0.01 &    0.00 &    0.00 &    0.00 &    0.00 &    0.00 &    0.00 &    0.00 &    0.00 &    0.00 &    0.00 &    0.00 \\
(2, 2) &    0.01 &    0.36 &    0.00 &    0.00 &    0.00 &    0.00 &    0.70 &   -8.61 &    2.63 &    0.01 &    0.00 &    0.00 &    0.00 &    4.66 &    0.19 &    0.00 &    0.00 &    0.00 &    0.00 &    0.02 &    0.00 &    0.00 &    0.00 &    0.00 &    0.00 &    0.00 &    0.00 &    0.00 &    0.00 &    0.00 &    0.00 &    0.00 &    0.00 &    0.00 &    0.00 &    0.00 \\
(2, 3) &    0.00 &    0.01 &    0.08 &    0.00 &    0.00 &    0.00 &    0.01 &    8.63 &  -17.41 &    0.14 &    0.01 &    0.00 &    0.00 &    0.07 &    8.32 &    0.02 &    0.00 &    0.00 &    0.00 &    0.00 &    0.10 &    0.01 &    0.00 &    0.00 &    0.00 &    0.00 &    0.01 &    0.00 &    0.00 &    0.00 &    0.00 &    0.00 &    0.00 &    0.00 &    0.00 &    0.00 \\
(2, 4) &    0.00 &    0.00 &    0.12 &    0.12 &    0.00 &    0.00 &    0.06 &    5.44 &   34.72 &  -72.05 &    1.60 &    0.24 &    0.00 &    0.00 &    0.24 &   22.77 &    0.06 &    0.00 &    0.00 &    0.00 &    0.00 &    5.03 &    0.00 &    0.12 &    0.00 &    0.00 &    0.00 &    1.24 &    0.00 &    0.06 &    0.00 &    0.00 &    0.00 &    0.18 &    0.00 &    0.06 \\
(2, 5) &    0.00 &    0.00 &    0.00 &    0.00 &    1.39 &    0.00 &    0.00 &    6.93 &   16.63 &   23.56 &  -94.22 &    4.16 &    0.00 &    0.00 &    0.00 &    0.69 &   18.71 &    1.39 &    0.00 &    0.00 &    0.00 &    0.00 &    9.70 &    0.00 &    0.00 &    0.00 &    0.00 &    0.00 &    7.62 &    0.00 &    0.00 &    0.00 &    0.00 &    0.00 &    3.46 &    0.00 \\
(2, 6) &    0.00 &    0.00 &    0.00 &    0.00 &    0.00 &    0.00 &    0.00 &    6.84 &   22.80 &   11.40 &    9.12 &  -88.93 &    0.00 &    0.00 &    0.00 &    0.00 &    0.00 &   18.24 &    0.00 &    0.00 &    0.00 &    0.00 &    0.00 &    9.12 &    0.00 &    0.00 &    0.00 &    0.00 &    0.00 &    6.84 &    0.00 &    0.00 &    0.00 &    0.00 &    0.00 &    4.56 \\
(3, 1) &    0.11 &    0.00 &    0.00 &    0.00 &    0.00 &    0.00 &   12.41 &    0.07 &    0.00 &    0.00 &    0.00 &    0.00 &  -34.99 &   21.30 &    0.40 &    0.01 &    0.00 &    0.00 &    0.48 &    0.11 &    0.01 &    0.00 &    0.00 &    0.00 &    0.04 &    0.00 &    0.01 &    0.00 &    0.00 &    0.00 &    0.01 &    0.01 &    0.00 &    0.00 &    0.00 &    0.00 \\
(3, 2) &    0.00 &    0.01 &    0.00 &    0.00 &    0.00 &    0.00 &    0.01 &    6.66 &    0.03 &    0.00 &    0.00 &    0.00 &    0.29 &  -12.48 &    4.89 &    0.03 &    0.01 &    0.00 &    0.00 &    0.46 &    0.06 &    0.00 &    0.00 &    0.00 &    0.00 &    0.02 &    0.00 &    0.00 &    0.00 &    0.00 &    0.00 &    0.00 &    0.00 &    0.00 &    0.00 &    0.00 \\
(3, 3) &    0.00 &    0.00 &    0.01 &    0.00 &    0.00 &    0.00 &    0.00 &    0.27 &    4.42 &    0.00 &    0.00 &    0.00 &    0.01 &    6.18 &  -13.65 &    0.83 &    0.03 &    0.02 &    0.00 &    0.00 &    1.69 &    0.09 &    0.00 &    0.00 &    0.00 &    0.00 &    0.06 &    0.01 &    0.01 &    0.00 &    0.00 &    0.00 &    0.02 &    0.00 &    0.00 &    0.00 \\
(3, 4) &    0.00 &    0.00 &    0.00 &    0.01 &    0.00 &    0.00 &    0.00 &    0.02 &    0.28 &    2.51 &    0.00 &    0.00 &    0.00 &    1.33 &   16.57 &  -31.22 &    0.86 &    0.25 &    0.00 &    0.00 &    0.09 &    8.06 &    0.08 &    0.01 &    0.00 &    0.00 &    0.02 &    0.71 &    0.04 &    0.01 &    0.00 &    0.00 &    0.00 &    0.29 &    0.02 &    0.04 \\
(3, 5) &    0.00 &    0.00 &    0.00 &    0.00 &    0.05 &    0.00 &    0.00 &    0.00 &    0.00 &    0.05 &    1.26 &    0.00 &    0.11 &    1.15 &    8.87 &   11.01 &  -38.77 &    3.01 &    0.00 &    0.00 &    0.05 &    0.16 &    6.79 &    0.11 &    0.00 &    0.00 &    0.05 &    0.05 &    3.40 &    0.05 &    0.00 &    0.00 &    0.05 &    0.05 &    2.35 &    0.11 \\
(3, 6) &    0.00 &    0.00 &    0.00 &    0.00 &    0.08 &    0.08 &    0.00 &    0.00 &    0.00 &    0.00 &    0.17 &    0.84 &    0.00 &    0.93 &    5.24 &    3.63 &    5.66 &  -27.80 &    0.00 &    0.00 &    0.25 &    0.25 &    0.08 &    3.80 &    0.00 &    0.00 &    0.17 &    0.25 &    0.25 &    1.77 &    0.00 &    0.00 &    0.00 &    0.17 &    0.93 &    3.21 \\
(4, 1) &    0.61 &    0.00 &    0.00 &    0.00 &    0.00 &    0.00 &    4.25 &    0.00 &    0.00 &    0.00 &    0.00 &    0.00 &   27.95 &    0.00 &    0.00 &    0.61 &    0.00 &    0.00 &  -70.47 &   25.52 &    5.47 &    3.04 &    0.00 &    0.61 &    1.82 &    0.00 &    0.00 &    0.00 &    0.00 &    0.00 &    0.61 &    0.00 &    0.00 &    0.00 &    0.00 &    0.00 \\
(4, 2) &    0.00 &    0.00 &    0.00 &    0.00 &    0.00 &    0.00 &    0.01 &    2.06 &    0.00 &    0.00 &    0.00 &    0.00 &    0.02 &   18.37 &    0.04 &    0.01 &    0.00 &    0.00 &    0.33 &  -34.61 &   11.34 &    0.61 &    0.08 &    0.02 &    0.00 &    1.34 &    0.12 &    0.01 &    0.00 &    0.00 &    0.00 &    0.24 &    0.02 &    0.01 &    0.01 &    0.00 \\
(4, 3) &    0.00 &    0.00 &    0.00 &    0.00 &    0.00 &    0.00 &    0.00 &    0.02 &    0.87 &    0.00 &    0.00 &    0.00 &    0.00 &    0.24 &   10.78 &    0.02 &    0.00 &    0.00 &    0.02 &    4.89 &  -24.04 &    4.58 &    0.27 &    0.13 &    0.00 &    0.00 &    1.65 &    0.15 &    0.02 &    0.01 &    0.00 &    0.00 &    0.32 &    0.01 &    0.02 &    0.02 \\
(4, 4) &    0.00 &    0.00 &    0.00 &    0.00 &    0.00 &    0.00 &    0.00 &    0.02 &    0.03 &    0.57 &    0.00 &    0.00 &    0.00 &    0.04 &    0.80 &    5.54 &    0.01 &    0.01 &    0.00 &    0.73 &    9.40 &  -24.82 &    2.03 &    0.62 &    0.00 &    0.00 &    0.03 &    3.87 &    0.29 &    0.02 &    0.00 &    0.00 &    0.00 &    0.61 &    0.05 &    0.15 \\
(4, 5) &    0.00 &    0.00 &    0.00 &    0.00 &    0.00 &    0.00 &    0.00 &    0.00 &    0.04 &    0.00 &    0.42 &    0.00 &    0.00 &    0.04 &    0.31 &    0.37 &    3.19 &    0.00 &    0.00 &    0.33 &    2.70 &   13.55 &  -36.57 &    2.95 &    0.00 &    0.00 &    0.07 &    0.33 &    8.64 &    0.15 &    0.00 &    0.00 &    0.04 &    0.07 &    3.17 &    0.20 \\
(4, 6) &    0.00 &    0.00 &    0.00 &    0.00 &    0.00 &    0.00 &    0.00 &    0.00 &    0.09 &    0.00 &    0.00 &    0.09 &    0.00 &    0.00 &    0.27 &    0.13 &    0.04 &    1.92 &    0.00 &    0.45 &    1.97 &    7.07 &    8.05 &  -31.27 &    0.00 &    0.00 &    0.00 &    0.13 &    0.09 &    3.58 &    0.00 &    0.00 &    0.00 &    0.31 &    0.72 &    6.35 \\
(5, 1) &    0.00 &    0.00 &    0.00 &    0.00 &    0.00 &    0.00 &    1.46 &    0.00 &    0.00 &    0.00 &    0.00 &    0.00 &    2.92 &    0.00 &    0.00 &    0.00 &    0.00 &    0.00 &    1.46 &    0.00 &    0.00 &    0.00 &    0.00 &    0.00 &  -27.70 &   16.03 &    1.46 &    2.92 &    0.00 &    0.00 &    1.46 &    0.00 &    0.00 &    0.00 &    0.00 &    0.00 \\
(5, 2) &    0.00 &    0.00 &    0.00 &    0.00 &    0.00 &    0.00 &    0.00 &    0.14 &    0.00 &    0.00 &    0.00 &    0.00 &    0.00 &    0.49 &    0.01 &    0.00 &    0.00 &    0.00 &    0.00 &    1.59 &    0.00 &    0.00 &    0.00 &    0.00 &    0.03 &  -10.69 &    7.37 &    0.22 &    0.09 &    0.01 &    0.00 &    0.72 &    0.03 &    0.01 &    0.00 &    0.00 \\
(5, 3) &    0.00 &    0.00 &    0.00 &    0.00 &    0.00 &    0.00 &    0.00 &    0.00 &    0.08 &    0.00 &    0.00 &    0.00 &    0.00 &    0.01 &    0.75 &    0.00 &    0.01 &    0.00 &    0.00 &    0.02 &    2.81 &    0.02 &    0.01 &    0.00 &    0.01 &    3.67 &  -12.08 &    2.22 &    0.30 &    0.10 &    0.00 &    0.00 &    1.98 &    0.03 &    0.02 &    0.01 \\
(5, 4) &    0.00 &    0.00 &    0.00 &    0.00 &    0.00 &    0.00 &    0.00 &    0.01 &    0.00 &    0.20 &    0.00 &    0.00 &    0.00 &    0.02 &    0.17 &    1.54 &    0.00 &    0.00 &    0.00 &    0.01 &    0.51 &    8.74 &    0.03 &    0.03 &    0.00 &    0.57 &    7.15 &  -29.81 &    5.45 &    1.28 &    0.00 &    0.01 &    0.00 &    3.60 &    0.38 &    0.13 \\
(5, 5) &    0.00 &    0.00 &    0.00 &    0.00 &    0.00 &    0.00 &    0.00 &    0.02 &    0.01 &    0.00 &    0.12 &    0.00 &    0.01 &    0.02 &    0.30 &    0.03 &    1.08 &    0.01 &    0.00 &    0.01 &    0.12 &    0.81 &    3.94 &    0.02 &    0.00 &    0.20 &    1.51 &    6.88 &  -23.49 &    3.06 &    0.00 &    0.00 &    0.07 &    0.07 &    4.44 &    0.72 \\
(5, 6) &    0.00 &    0.00 &    0.00 &    0.00 &    0.00 &    0.00 &    0.00 &    0.00 &    0.03 &    0.00 &    0.00 &    0.05 &    0.00 &    0.00 &    0.14 &    0.05 &    0.05 &    0.71 &    0.00 &    0.00 &    0.05 &    0.14 &    0.38 &    2.78 &    0.00 &    0.14 &    0.87 &    3.68 &    7.36 &  -21.97 &    0.00 &    0.00 &    0.03 &    0.16 &    0.30 &    5.04 \\
(6, 1) &    0.00 &    0.00 &    0.00 &    0.00 &    0.00 &    0.00 &   11.86 &    0.00 &    0.00 &    0.00 &    0.00 &    0.00 &    5.93 &    0.00 &    0.00 &    0.00 &    0.00 &    0.00 &    0.00 &    0.00 &    0.00 &    0.00 &    0.00 &    0.00 &    5.93 &    0.00 &    0.00 &    0.00 &    0.00 &    0.00 &  -82.99 &   47.43 &    0.00 &    5.93 &    5.93 &    0.00 \\
(6, 2) &    0.00 &    0.01 &    0.00 &    0.00 &    0.00 &    0.00 &    0.00 &    0.05 &    0.00 &    0.00 &    0.00 &    0.00 &    0.00 &    0.17 &    0.00 &    0.00 &    0.00 &    0.00 &    0.00 &    0.34 &    0.00 &    0.00 &    0.00 &    0.00 &    0.00 &    1.42 &    0.00 &    0.00 &    0.00 &    0.00 &    0.05 &  -11.67 &    9.27 &    0.23 &    0.05 &    0.09 \\
(6, 3) &    0.00 &    0.00 &    0.00 &    0.00 &    0.00 &    0.00 &    0.00 &    0.00 &    0.02 &    0.00 &    0.00 &    0.00 &    0.00 &    0.00 &    0.19 &    0.00 &    0.00 &    0.00 &    0.00 &    0.01 &    0.57 &    0.00 &    0.00 &    0.00 &    0.00 &    0.00 &    1.84 &    0.00 &    0.01 &    0.00 &    0.00 &    3.87 &   -8.89 &    1.92 &    0.24 &    0.20 \\
(6, 4) &    0.00 &    0.00 &    0.00 &    0.00 &    0.00 &    0.00 &    0.00 &    0.00 &    0.00 &    0.14 &    0.00 &    0.00 &    0.00 &    0.01 &    0.06 &    0.53 &    0.00 &    0.00 &    0.00 &    0.02 &    0.10 &    2.48 &    0.04 &    0.02 &    0.00 &    0.00 &    0.14 &    4.04 &    0.06 &    0.05 &    0.00 &    0.85 &    8.89 &  -24.08 &    4.56 &    2.08 \\
(6, 5) &    0.00 &    0.00 &    0.00 &    0.00 &    0.00 &    0.00 &    0.00 &    0.00 &    0.00 &    0.00 &    0.07 &    0.00 &    0.00 &    0.05 &    0.20 &    0.02 &    0.37 &    0.00 &    0.00 &    0.01 &    0.23 &    0.24 &    1.78 &    0.01 &    0.00 &    0.00 &    0.01 &    0.46 &    4.01 &    0.05 &    0.02 &    0.24 &    1.62 &    6.27 &  -21.34 &    5.66 \\
(6, 6) &    0.00 &    0.00 &    0.00 &    0.00 &    0.00 &    0.00 &    0.00 &    0.01 &    0.01 &    0.00 &    0.00 &    0.07 &    0.00 &    0.01 &    0.08 &    0.03 &    0.00 &    0.33 &    0.00 &    0.00 &    0.04 &    0.12 &    0.04 &    1.08 &    0.00 &    0.00 &    0.00 &    0.13 &    0.54 &    1.76 &    0.00 &    0.13 &    0.72 &    2.63 &    4.35 &  -12.06 \\
\bottomrule
\end{tabular}
}%
  }
\end{sidewaystable*}

Moreover, the estimated transition probabilities \eqref{Eq 50} reveal the correlation in spreads.
To see this, we sort the pairs of spreads by their occupation time,
and pick the six most frequently visited pairs for inspection.

For each selected pair \((d^1,d^2)\) (in ticks), 
Figure \ref{Fig 5} plots its heatmap of transition probabilities to neighboring states, 
where \((\Delta d^1, \Delta d^2)\) denotes its increment (in ticks) after transition. 
These plots give an overview of the joint spread evolution in high frequency. 
Given a transition, we observe that, either the smaller spread tends to widen or 
the larger spread tends to reduce, so as to move towards the other one.
In other words, the transitions quickly ``pull'' the spreads back towards each other so that the joint spread process is mean reverting at the microscopic level. 

\begin{figure}[h]
  \centering
  \begin{subfigure}{0.4\textwidth}
      \includegraphics[width=\linewidth]{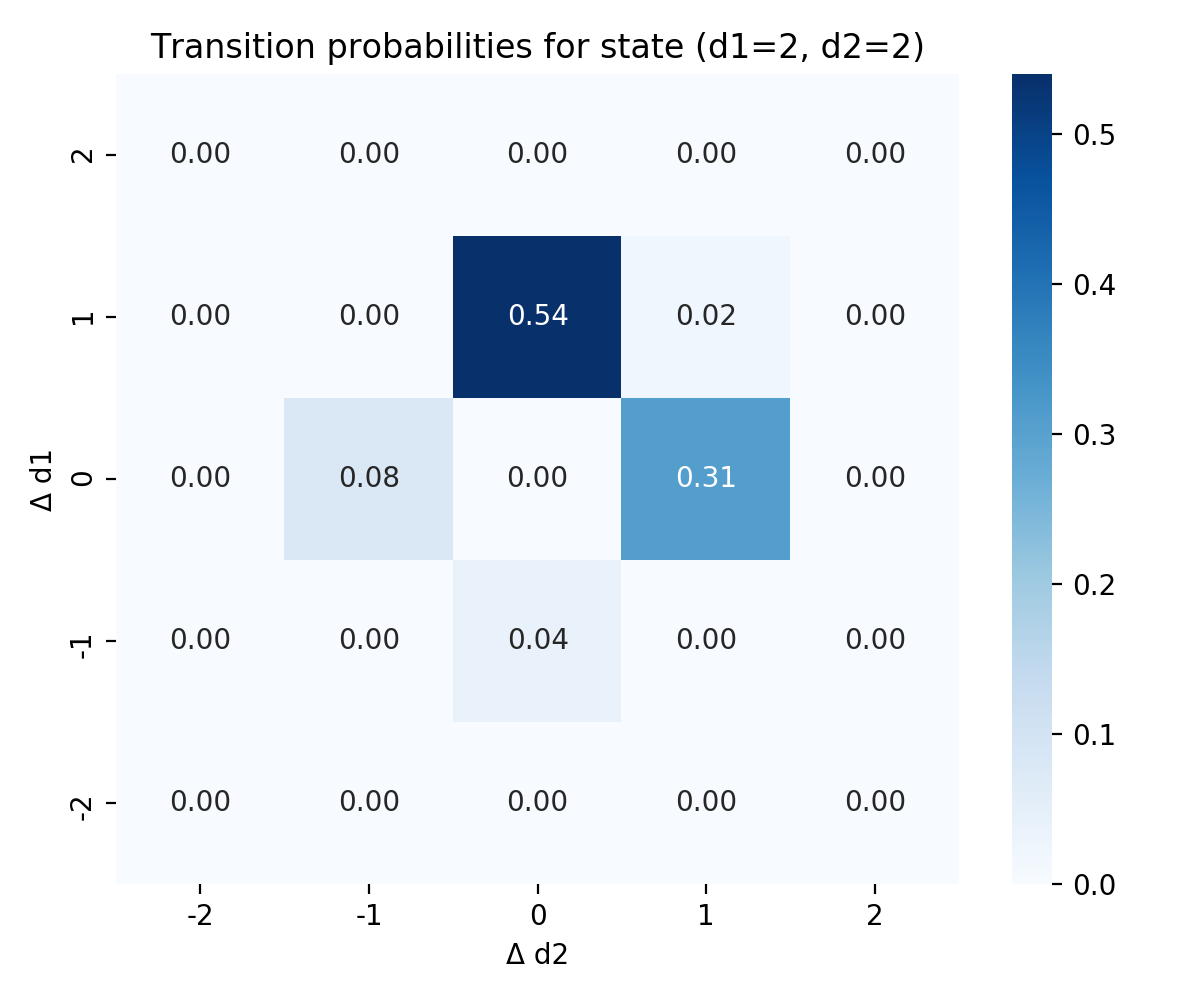}
  \end{subfigure}
  \hspace{0.1\textwidth}
  \begin{subfigure}{0.4\textwidth}
      \includegraphics[width=\linewidth]{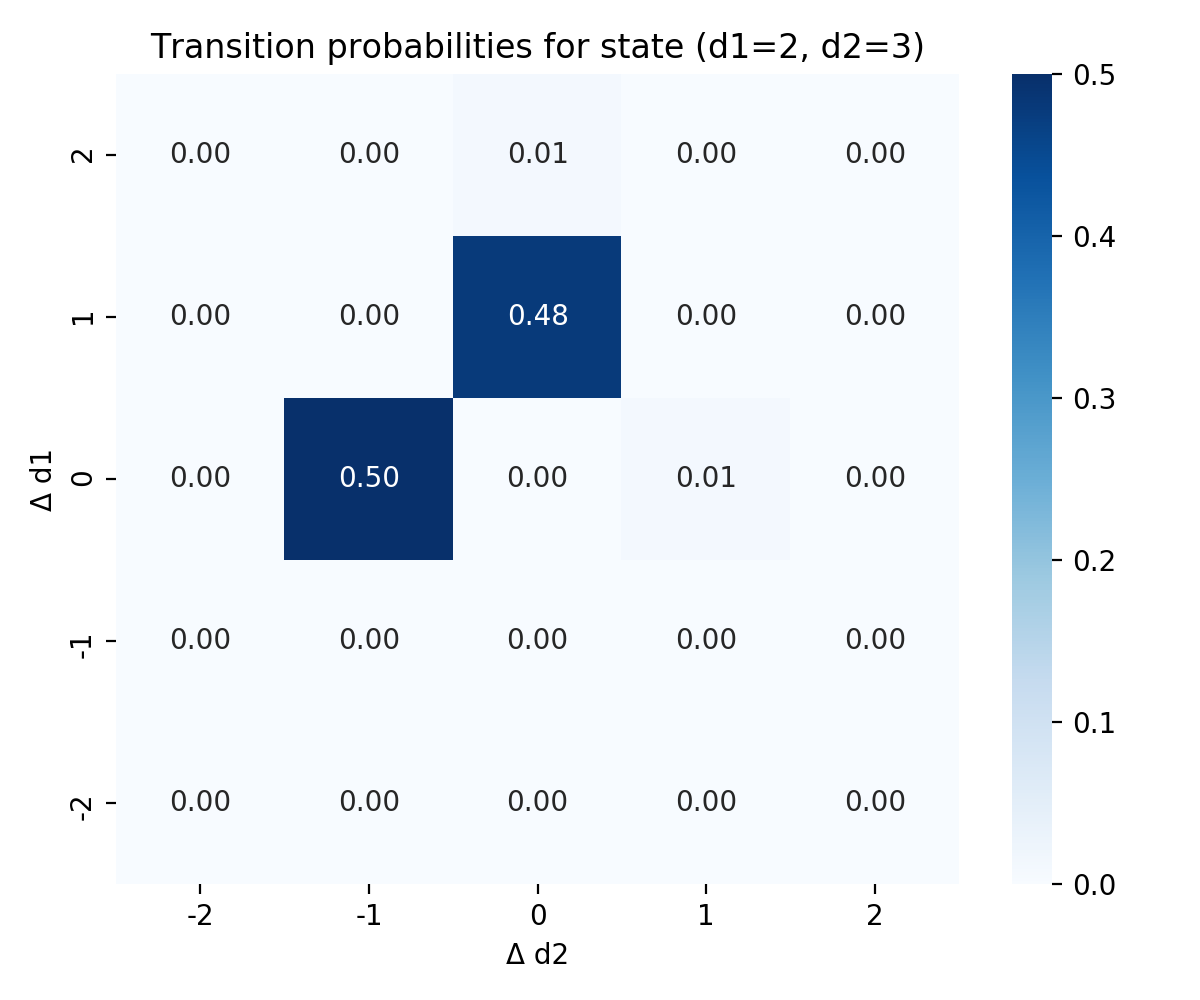}
  \end{subfigure}
  
  \begin{subfigure}{0.4\textwidth}
      \includegraphics[width=\linewidth]{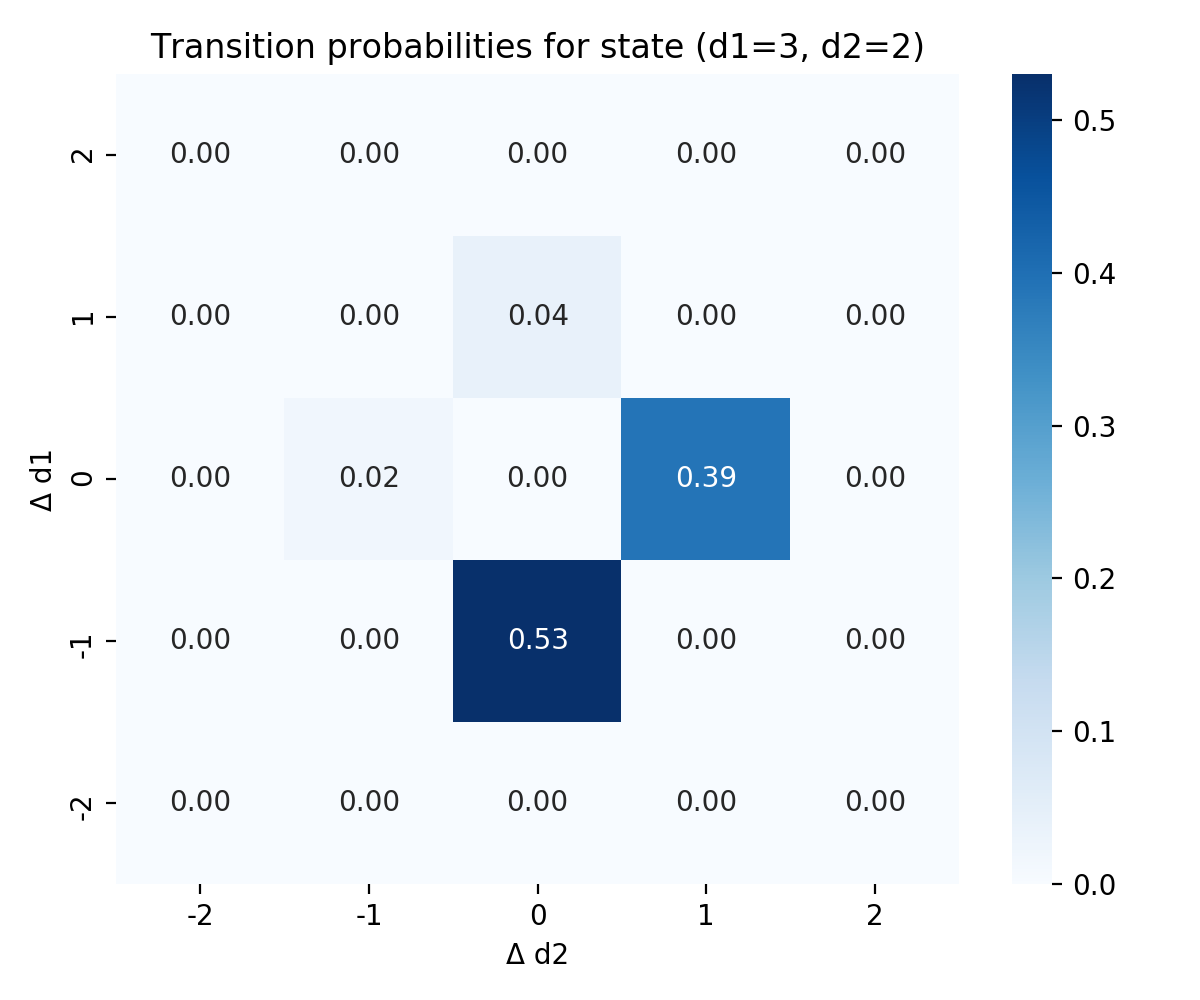}
  \end{subfigure}
  \hspace{0.1\textwidth}
  \begin{subfigure}{0.4\textwidth}
      \includegraphics[width=\linewidth]{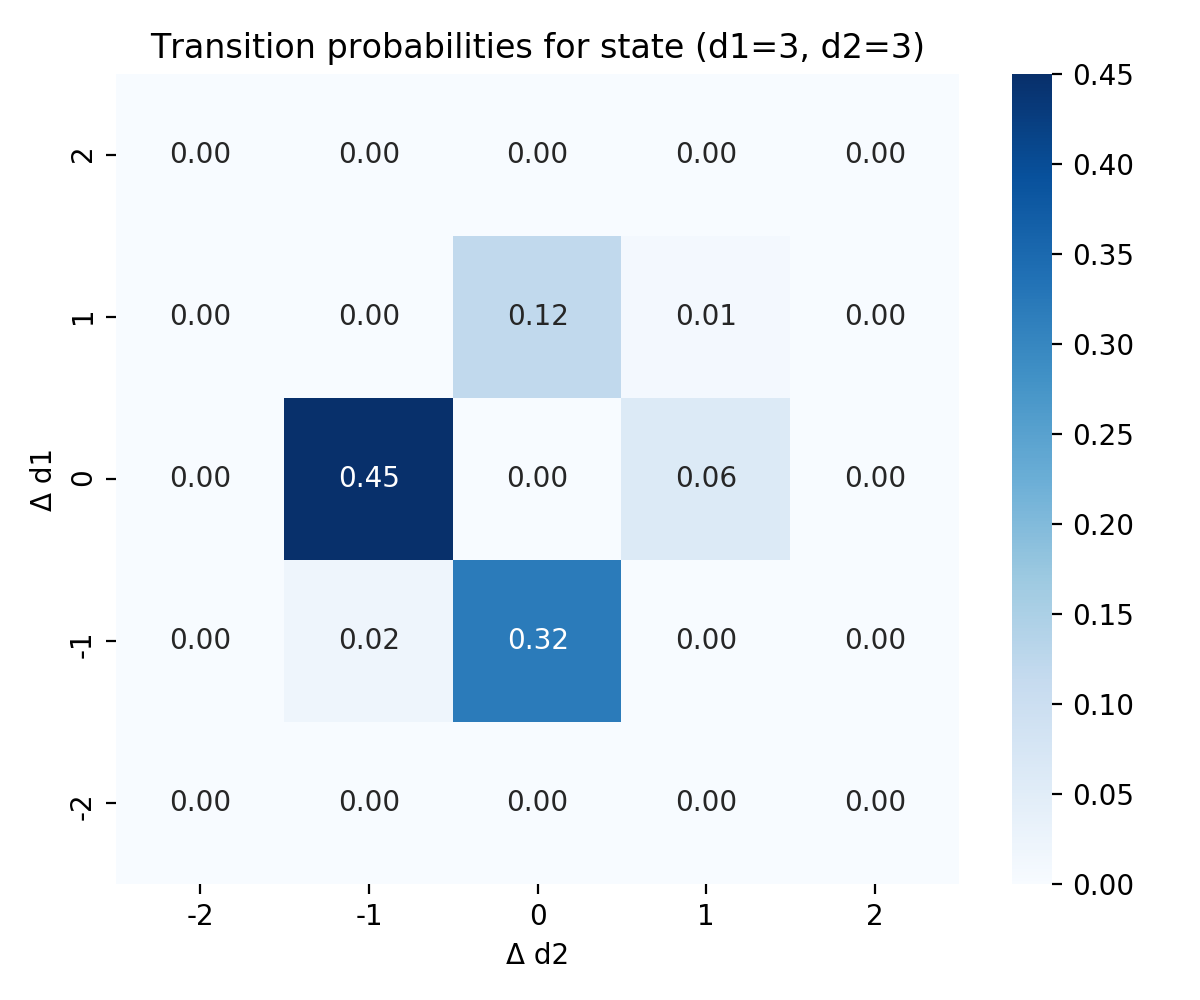}
  \end{subfigure}
  
  \begin{subfigure}{0.4\textwidth}
      \includegraphics[width=\linewidth]{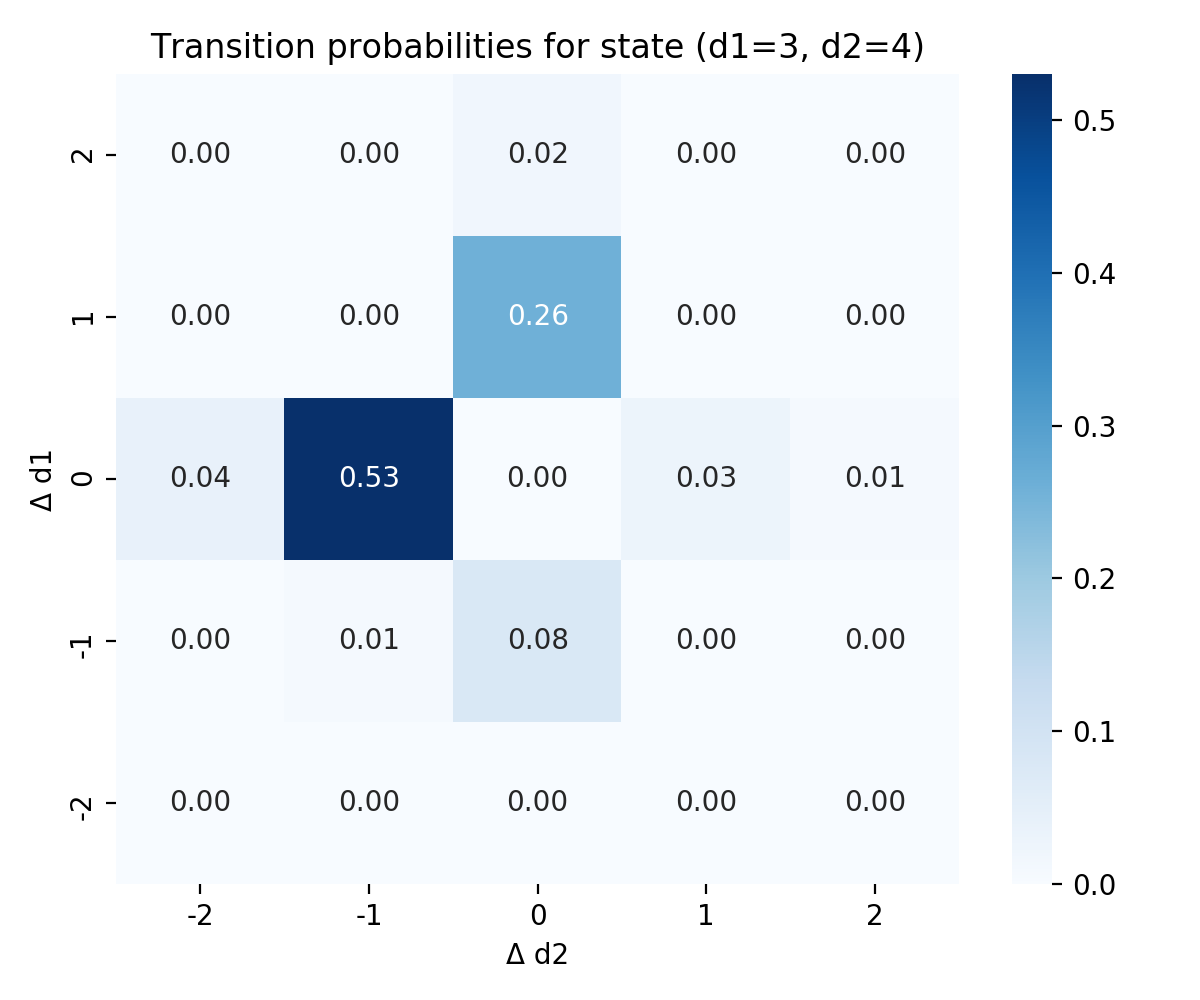}
  \end{subfigure}
  \hspace{0.1\textwidth}
  \begin{subfigure}{0.4\textwidth}
      \includegraphics[width=\linewidth]{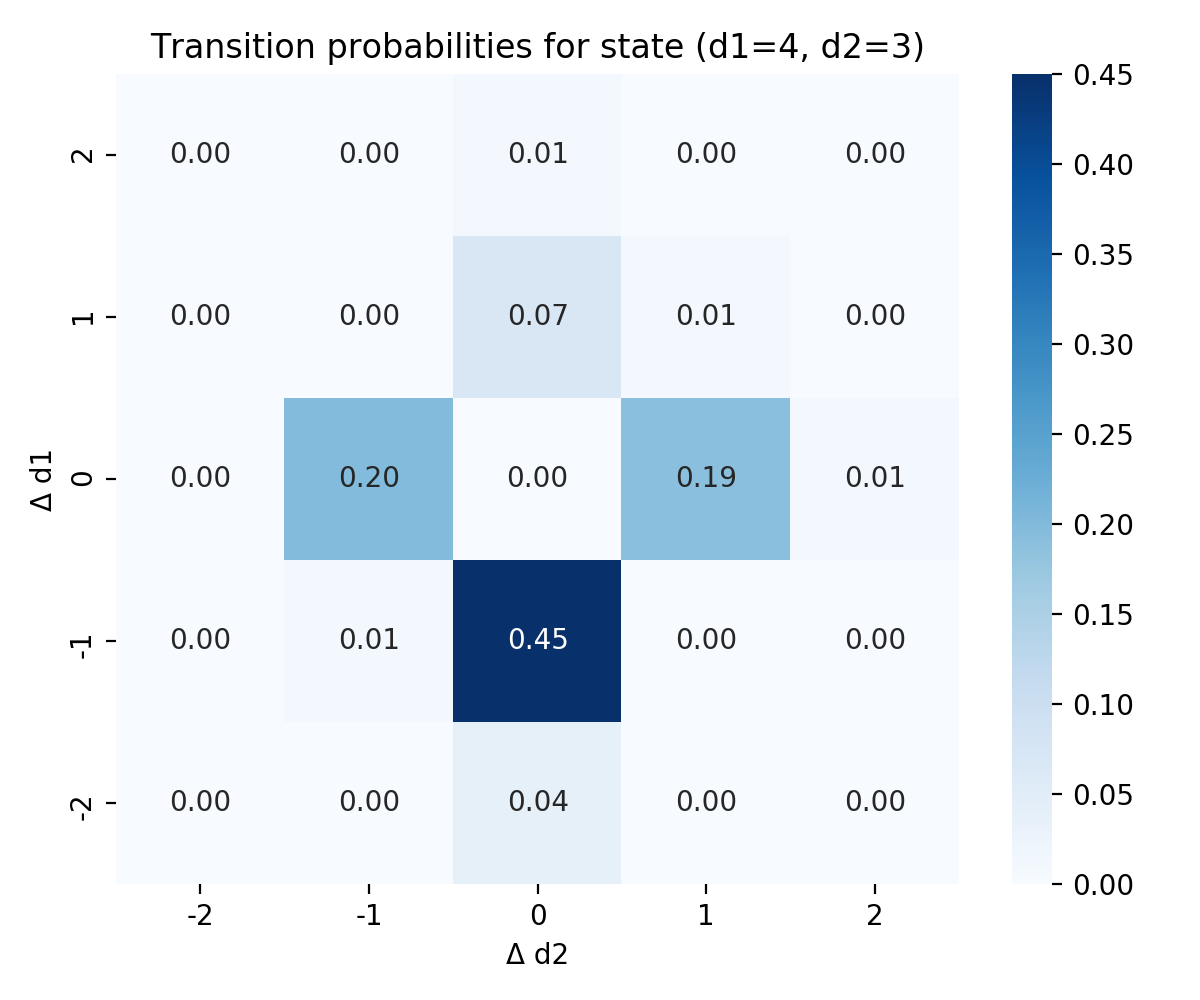}
  \end{subfigure}
  
  \caption{Plots of the estimated transition probabilities,
  expressed in ticks and \(\text{second}^{-1}\).}
  \label{Fig 5}
\end{figure}

\subsection{Estimation for limit order execution processes}
To estimate the intensities of the limit order execution processes, we follow the arguments in \cite{guilbaud2013optimal}.
Ideally, if the market maker can continuously observe her execution point processes \((N^{b,i},N^{a,i})_{i=1}^M\) while applying a make strategy \((\alpha^{b,i},\alpha^{a,i})_{i=1}^M\),
then for each call \(C^i\) the intensity function \(\lambda^{b,i}_{d^i}(\alpha^{b,i})\) (resp. \(\lambda^{a,i}_{d^i}(\alpha^{a,i})\)) 
can be estimated by first counting the number of executions at bid (resp.\ ask) only when the system is in the state \((\alpha^{b,i},d^i)\) (resp. \((\alpha^{a,i},d^i)\)) 
and then normalizing this quantity by the time the system spends in the state \((\alpha^{b,i},d^i)\) (resp. \((\alpha^{a,i},d^i)\)).
Mathematically, on the bid side, for each call \(C^i\), given \(\alpha^{b,i} \in \{0,1\}, d^i \in \mathbb{S}\),
\begin{equation}{\label{Eq 52}}
  N_t^{b,i,\alpha^{b,i},d^i} \coloneqq  \int_0^t \mathbb{I}_{\{\alpha^{b,i}_u=\alpha^{b,i},D^i_{u^-}=d^i\}} dN_u^{b,i},\quad t \geq 0 
\end{equation}
counts the number of executions at bid when the system is in the state \((\alpha^{b,i},d^i)\) over \([0,t]\), and 
\begin{equation}{\label{Eq 53}}
  \mathcal{T}_t^{b,i,\alpha^{b,i},d^i} \coloneqq \int_0^t \mathbb{I}_{\{\alpha^{b,i}_u=\alpha^{b,i},D^i_{u^-}=d^i\}} du,\quad t \geq 0
\end{equation}
represents the time that the system spends in the state \((\alpha^{b,i},d^i)\) over \([0,t]\). Then
\[
  \hat{\lambda}^{b,i}_{d^i}(\alpha^{b,i}) \coloneqq \frac{N_T^{b,i,\alpha^{b,i},d^i}}{\mathcal{T}_T^{b,i,\alpha^{b,i},d^i}},\quad \alpha^{b,i} \in \{0,1\},\quad d^i\in \mathbb{S}
\] estimates the intensity function \(\lambda^{b,i}_{d^i}(\alpha^{b,i})\).
Similarly on the ask side, for each call \(C^i\), 
\[
  \hat{\lambda}^{a,i}_{d^i}(\alpha^{a,i}) \coloneqq \frac{N_T^{a,i,\alpha^{a,i},d^i}}{\mathcal{T}_T^{a,i,\alpha^{a,i},d^i}},\quad \alpha^{a,i} \in \{0,1\},\quad d^i\in \mathbb{S}
\] estimates the intensity function \(\lambda^{a,i}_{d^i}(\alpha^{a,i})\), where 
\(N_T^{a,i,\alpha^{a,i},d^i}\) and \(\mathcal{T}_T^{a,i,\alpha^{a,i},d^i}\) are defined analogously by replacing \((b,\alpha^{b,i},N^{b,i})\) with \((a,\alpha^{a,i},N^{a,i})\) 
in \eqref{Eq 52} and \eqref{Eq 53}. 

Unfortunately, since no one precisely performs the proposed make strategy on the real-world order book, 
we cannot directly observe the actual execution processes \((N^{b,i},N^{a,i})_{i=1}^M\).
Instead, for any \(i\in \mathcal{I}, (\alpha^{b,i},\alpha^{a,i})\in \{0,1\} \times \{0,1\},d^i\in \mathbb{S}\), 
we shall consider reasonable proxies \(\tilde{N}^{b,i,\alpha^{b,i},d^i}\) and \(\tilde{N}^{a,i,\alpha^{a,i},d^i}\) 
for \(N^{b,i,\alpha^{b,i},d^i}\) and \(N^{a,i,\alpha^{a,i},d^i}\),  
which yield proxy estimates \(\tilde{\lambda}^{b,i}_{d^i}(\alpha^{b,i})\) and \(\tilde{\lambda}^{a,i}_{d^i}(\alpha^{a,i})\). 

Suppose that in addition to the spread Markov chain \((D_{T_n})_n\), we can also observe, at each jump time \(T_n,\ n \geq 1\), 
the volumes \((V^{b,i}_{T_n},V^{a,i}_{T_n})_{i=1}^M\) offered at the best available prices,
as well as the total volumes \((V^{\mathrm{buy},i}_{T_{n-1},T_n},V^{\mathrm{sell},i}_{T_{n-1},T_n})_{i=1}^M\) 
of market orders that arrive at the best available prices between two consecutive jump times \(T_{n-1}\) and \(T_n\). 
As the market maker buys and sells one unit of call per trade, with a virtual reference volume \(V_0=1\) on the order book, 
we can then define for each \(i\in \mathcal{I}, d^i\in \mathbb{S}\) the proxies \(\tilde{N}^{b,i,\alpha^{b,i},d^i}\) and \(\tilde{N}^{a,i,\alpha^{a,i},d^i}\) 
at jump times \(T_n\) by 
\[
  \begin{dcases}
    \tilde{N}^{b,i,0,d^i}_0 \coloneqq 0,\quad \tilde{N}^{b,i,0,d^i}_{T_{n+1}} \coloneqq \tilde{N}^{b,i,0,d^i}_{T_n} + \mathbb{I}_{\{V_0 + V^{b,i}_{T_n} < V^{\mathrm{sell},i}_{T_n,T_{n+1}},D^i_{T_n}=d^i\}},\quad n \geq 0\\ 
    \tilde{N}^{b,i,1,d^i}_0 \coloneqq 0,\quad \tilde{N}^{b,i,1,d^i}_{T_{n+1}} \coloneqq \tilde{N}^{b,i,1,d^i}_{T_n} + \mathbb{I}_{\{V_0< V^{\mathrm{sell},i}_{T_n,T_{n+1}},D^i_{T_n}=d^i\}},\quad n \geq 0\\
    \tilde{N}^{a,i,0,d^i}_0 \coloneqq 0,\quad \tilde{N}^{a,i,0,d^i}_{T_{n+1}} \coloneqq \tilde{N}^{a,i,0,d^i}_{T_n} + \mathbb{I}_{\{V_0 + V^{a,i}_{T_n} < V^{\mathrm{buy},i}_{T_n,T_{n+1}},D^i_{T_n}=d^i\}},\quad n \geq 0\\
    \tilde{N}^{a,i,1,d^i}_0 \coloneqq 0,\quad \tilde{N}^{a,i,1,d^i}_{T_{n+1}} \coloneqq \tilde{N}^{a,i,1,d^i}_{T_n} + \mathbb{I}_{\{V_0 < V^{\mathrm{buy},i}_{T_n,T_{n+1}},D^i_{T_n}=d^i\}},\quad n \geq 0\\
  \end{dcases}
\]
along with the proxy \(\tilde{\mathcal{T}}^{i,d^i}\) for the time that the spread \(D^i\) spends in the state \(d^i\), inductively defined by
\[
  \tilde{\mathcal{T}}^{i,d^i}_0 \coloneqq 0, \quad \tilde{\mathcal{T}}^{i,d^i}_{T_{n+1}} \coloneqq \tilde{\mathcal{T}}^{i,d^i}_{T_n} + (T_{n+1}-T_n) \mathbb{I}_{\{D^i_{T_n}=d^i\}},\ n \geq 0. 
\]  

The interpretation of these proxies is rooted in the assumed \emph{price-time microstructure} of the LOB.
We consider a special case where the \emph{small} market maker instantaneously updates her make strategy \((\alpha^{b,i},\alpha^{a,i})_{i=1}^M\) 
only when the spreads change \emph{exogenously}, i.e. at the times \((T_n)_n\), so that the spreads (excluding her own quotes) remain constant between updates.        
If the market maker adopts the aggressive make strategy \((\alpha^{b,i}_{T_n},\alpha^{a,i}_{T_n})=(1,1)\),
she improves the best available prices, thereby receiving top execution priority and capturing all incoming market order flow.
We thus increment \(\tilde{N}^{b,i,1,d^i}\) (resp. \(\tilde{N}^{a,i,1,d^i}\)) at \(T_n\) only when the total volume \(V^{\mathrm{sell},i}_{T_{n-1},T_n}\)
(resp. \(V^{\mathrm{buy},i}_{T_{n-1},T_n}\)) of counterpart market orders covers the order size \(V_0\).
In contrast, if the market maker adopts the conservative make strategy \((\alpha^{b,i}_{T_n},\alpha^{a,i}_{T_n})=(0,0)\),
she only adds liquidity to the best available prices, so her limit orders are ranked behind the existing ones 
and will be executed only after all previously posted orders have been filled. 
We thus increment \(\tilde{N}^{b,i,0,d^i}\) (resp. \(\tilde{N}^{a,i,0,d^i}\)) at \(T_n\) 
only when the remaining volume \(V^{\mathrm{sell},i}_{T_{n-1},T_n}-V^{b,i}_{T_{n-1}}\) 
(resp. \(V^{\mathrm{buy},i}_{T_{n-1},T_n}-V^{a,i}_{T_{n-1}}\)) of counterpart market orders is sufficient to cover the order size \(V_0\).

Finally, with the proxies \(\tilde{N}^{b,i,\alpha^{b,i},d^i},\tilde{N}^{a,i,\alpha^{a,i},d^i}\) and \(\tilde{\mathcal{T}}^{i,d^i}\) on hand, 
we may define (for \(n\) large enough) the proxy estimates
\begin{equation}{\label{Eq 54}}
  \tilde{\lambda}^{b,i}_{d^i}(\alpha^{b,i}) \coloneqq \frac{\tilde{N}^{b,i,\alpha^{b,i},d^i}_{T_{n}}}{\tilde{\mathcal{T}}^{i,d^i}_{T_n}},\quad 
  \tilde{\lambda}^{a,i}_{d^i}(\alpha^{a,i}) \coloneqq \frac{\tilde{N}^{a,i,\alpha^{a,i},d^i}_{T_{n}}}{\tilde{\mathcal{T}}^{i,d^i}_{T_n}}
\end{equation}
for the execution intensities \(\lambda^{b,i}_{d^i}(\alpha^{b,i})\) and \(\lambda^{a,i}_{d^i}(\alpha^{a,i})\). 
Note that the proxy estimates satisfy the assumed property \eqref{Eq 3} of the execution intensities: 
\[
  \tilde{\lambda}^{b,i}_{d^i}(0) \leq \tilde{\lambda}^{b,i}_{d^i}(1), \quad \tilde{\lambda}^{a,i}_{d^i}(0) \leq \tilde{\lambda}^{a,i}_{d^i}(1), \quad i\in \mathcal{I},\quad d^i \in \mathbb{S}.
\]

Now we perform the estimation \eqref{Eq 54} on our data. 
Table \ref{Table 3} displays the estimated intensities of the execution processes \((N^{b,i},N^{a,i})_{i=1}^M\) 
under different spread values and strategies; ``aggressive'' means improving the best price by one tick and ``conservative'' means quoting at the touch. 
Figure \ref{Fig 6} plots the affine interpolation of the estimated execution intensities as a function of the spread.

\begin{table}[ht]
  \centering
  \resizebox{\textwidth}{!}{\begin{tabular}{lcccccccc}
\toprule
{} & \multicolumn{4}{c}{strike 4225} & \multicolumn{4}{c}{strike 4230} \\
\textbf{Side} & \multicolumn{2}{c}{ask} & \multicolumn{2}{c}{bid} & \multicolumn{2}{c}{ask} & \multicolumn{2}{c}{bid} \\
\textbf{Strategy} &  aggressive & conservative & aggressive & conservative &  aggressive & conservative & aggressive & conservative \\
\textbf{Spread (USD)} &             &              &            &              &             &              &            &              \\
\midrule
\textbf{0.01        } &      0.1821 &       0.0000 &     0.1252 &       0.0114 &      0.1484 &       0.0000 &     0.0691 &       0.0041 \\
\textbf{0.02        } &      0.0075 &       0.0006 &     0.0087 &       0.0011 &      0.0100 &       0.0008 &     0.0074 &       0.0001 \\
\textbf{0.03        } &      0.0019 &       0.0000 &     0.0014 &       0.0000 &      0.0016 &       0.0000 &     0.0005 &       0.0004 \\
\textbf{0.04        } &      0.0000 &       0.0000 &     0.0000 &       0.0000 &      0.0014 &       0.0000 &     0.0000 &       0.0000 \\
\textbf{0.05        } &      0.0000 &       0.0000 &     0.0000 &       0.0000 &      0.0000 &       0.0000 &     0.0000 &       0.0000 \\
\textbf{0.06        } &      0.0000 &       0.0000 &     0.0000 &       0.0000 &      0.0000 &       0.0000 &     0.0000 &       0.0000 \\
\bottomrule
\end{tabular}
}
  \caption{Estimated limit order execution intensities, expressed in second\(^{-1}\).}
  \label{Table 3}
\end{table}

\begin{figure}[H]
  \centering
  \includegraphics[width=\textwidth]{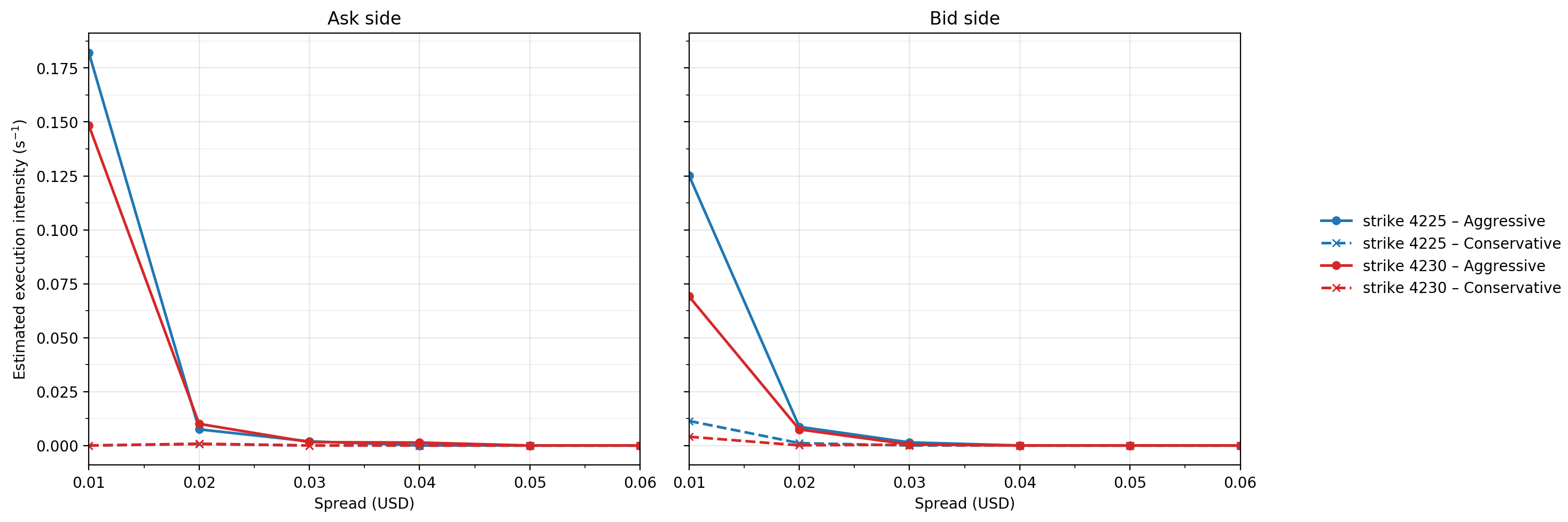}
  \caption{Plot of estimated limit order execution intensities, expressed in second\(^{-1}\).}
  \label{Fig 6}
\end{figure}

\begin{remark}
  The estimated execution intensities exhibit an asymmetry between the bid and ask sides. 
  This may be attributed to the unbalanced or non-stationary price movement observed on the specific trading day used for calibration. 
  We emphasize that this phenomenon is unlikely to be typical across trading days and should be interpreted with caution.
\end{remark}

\section{Moderate dependence of value function on spread}{\label{Appendix C}}
\begin{figure}[H]
    \centering
    \begin{subfigure}[t]{0.45\textwidth}
        \centering
        \includegraphics[width=\linewidth]{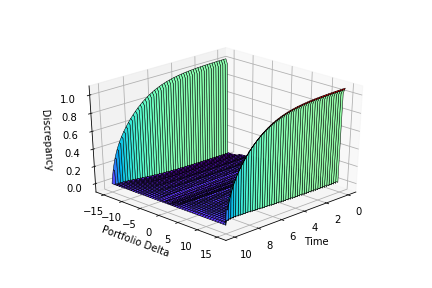}
        \caption{$d=(2,2),\ d'=(6,6)$}
    \end{subfigure}\hfill
    \begin{subfigure}[t]{0.45\textwidth}
        \centering
        \includegraphics[width=\linewidth]{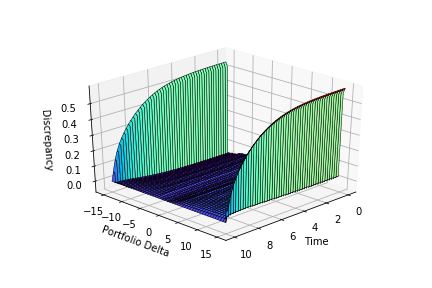}
        \caption{$d=(3,3),\ d'=(5,5)$}
    \end{subfigure}
    \caption{Discrepancy between reduced value functions $\varphi_d$ and $\varphi_{d'}$,
    with $d$ and $d'$ expressed in ticks.}
    \label{Fig 7}
\end{figure}

In this section, we provide a numerical justification for the moderate dependence of the value function on the spread variable.
Specifically, we compare the value functions \(\varphi_d\) and \(\varphi_{d^\prime}\)
for distinct spread values \(d\) and \(d^\prime\), and examine their difference \(\varphi_{d}\) and \(\varphi_{d^\prime}\).  

For simplicity, we take \(d=(2,2), d^\prime =(6,6)\) and \(d=(3,3),d^\prime =(5,5)\) (in ticks) as examples. As illustrated in \ref{Fig 7}, the discrepancy between the value functions \(\varphi_d\) and \(\varphi_{d^\prime}\) is negligible in the interior of the state space and primarily arises near the boundary.
This observation suggests that the difference is probably attributed to truncation effects at the boundary, rather than to a genuine dependence on the spread variable.
Therefore, up to a tolerable boundary error, we may conclude that the value function \(\varphi_d\) mildly depends on the choice of the spread \(d\).

\end{appendix}

\end{document}